\documentclass{lmcs}

\usepackage[utf8]{inputenc}

\usepackage{lastpage}
\lmcsdoi{17}{2}{11}
\lmcsheading{}{\pageref{LastPage}}{}{}%
{Feb.~26,~2020}{Apr.~23,~2021}{}

\usepackage[strings]{underscore} 

\keywords{Concurrency, process algebra, time-outs, priority, CCSP, labelled transition systems, semantic equivalences, linear time, branching time, failure trace semantics, absolute expressiveness, full abstraction, safety properties, may testing.}
\usepackage{hyperref,breakurl,colordvi}

\def\eg{{\em e.g.}}
\newfont{\bbb}{bbm10 scaled 1100}                       
\newfont{\bbbs}{bbm10 scaled 900}                       
\newcommand{\IN}{\mbox{\bbb N}}                         
\newcommand{\INs}{\mbox{\bbbs N}}                       
\newcommand{\IP}{\mbox{\bbb P}}                         
\DeclareSymbolFont{frenchscript}{OMS}{ztmcm}{m}{n}
\DeclareMathSymbol{\Pow}{\mathord}{frenchscript}{80}    
\DeclareMathSymbol{\R}{\mathord}{frenchscript}{82}      
\DeclareMathSymbol{\fO}{\mathord}{frenchscript}{79}     
\DeclareMathSymbol{\fG}{\mathord}{frenchscript}{71}     
\newcommand{\Rn}{\mathcal{R}}                           
\newcommand{\I}{\mathcal{I}}                            
\newenvironment{definition}[1]{\begin{defi} \label{df:#1} }{\end{defi}}
\newenvironment{theorem}[1]{\begin{thm} \rm \label{thm:#1} }{\end{thm}}
\newenvironment{proposition}[1]{\begin{prop} \rm \label{pr:#1} }{\end{prop}}
\newenvironment{lemma}[1]{\begin{lem} \rm \label{lem:#1} }{\end{lem}}
\newenvironment{corollary}[1]{\begin{cor} \rm \label{cor:#1} }{\end{cor}}
\newenvironment{example}[1]{\begin{exa} \rm \label{ex:#1} }{\end{exa}}
\newenvironment{observation}[1]{\begin{obs} \rm \label{obs:#1} }{\end{obs}}

\newcommand{\df}[1]{Definition~\ref{df:#1}}
\newcommand{\thrm}[1]{Theorem~\ref{thm:#1}}
\newcommand{\pr}[1]{Proposition~\ref{pr:#1}}
\newcommand{\lemref}[1]{Lemma~\ref{lem:#1}}
\newcommand{\corref}[1]{Corollary~\ref{cor:#1}}

\newcommand{\obsref}[1]{Observation~\ref{obs:#1}}
\newcommand{\tab}[1]{Table~\ref{tab:#1}}
\makeatletter
\def\comesfrom{\@transition\leftarrowfill}
\def\goesto{\@transition\rightarrowfill}
\def\ngoesto{\@transition\nrightarrowfill}
\def\Goesto{\@transition\Rightarrowfill}
\def\nGoesto{\@transition\nRightarrowfill}
\def\xmapsto{\@transition\mapstofill}
\def\nxmapsto{\@transition\nmapstofill}
\def\@transition#1{\@@transition{#1}}
\newbox\@transbox
\newbox\@arrowbox
\def\@@transition#1#2%
   {\setbox\@transbox\hbox
      {\vrule height 1.5ex depth .8ex width 0ex\hskip0.25em$\scriptstyle#2$\hskip0.25em}
   \ifdim\wd\@transbox<1.5em
      \setbox\@transbox\hbox to 1.5em{\hfil\box\@transbox\hfil}\fi
   \setbox\@arrowbox\hbox to \wd\@transbox{#1}
   \ht\@arrowbox\z@\dp\@arrowbox\z@
   \setbox\@transbox\hbox{$\mathop{\box\@arrowbox}\limits^{\box\@transbox}$}
   \dp\@transbox\z@\ht\@transbox 10pt
   \mathrel{\box\@transbox}}
\def\nrightarrowfill{$\m@th\mathord-\mkern-6mu%
  \cleaders\hbox{$\mkern-2mu\mathord-\mkern-2mu$}\hfill
  \mkern-6mu\mathord\not\mkern-2mu\mathord\rightarrow$}
\def\Rightarrowfill{$\m@th\mathord=\mkern-6mu%
  \cleaders\hbox{$\mkern-2mu\mathord=\mkern-2mu$}\hfill
  \mkern-6mu\mathord\Rightarrow$}
\def\nRightarrowfill{$\m@th\mathord=\mkern-6mu%
  \cleaders\hbox{$\mkern-2mu\mathord=\mkern-2mu$}\hfill
  \mkern-6mu\mathord\not\mathord\Rightarrow$}
\def\mapstofill{$\m@th\mathord\mapstochar\mathord-\mkern-6mu%
  \cleaders\hbox{$\mkern-2mu\mathord-\mkern-2mu$}\hfill
  \mkern-6mu\mathord\rightarrow$}
\def\nmapstofill{$\m@th\mathord\mapstochar\mathord-\mkern-6mu%
  \cleaders\hbox{$\mkern-2mu\mathord-\mkern-2mu$}\hfill
  \mkern-6mu\mathord\not\mkern-2mu\mathord\rightarrow$}
\makeatother 

\newcommand{\plat}[1]{\raisebox{0pt}[0pt][0pt]{#1}}     
\newcommand{\rec}[1]{\plat{$                            
        \stackrel{\mbox{\tiny $/$}}
        {\raisebox{-.3ex}[.3ex]{\tiny $\backslash$}}
        \!\!#1\!\!
        \stackrel{\mbox{\tiny $\backslash$}}
        {\raisebox{-.3ex}[.3ex]{\tiny $/$}}$}}
\newcommand{\obis}[2]{\mathrel{_{#1}\,                  
        \raisebox{.3ex}{$\underline{\makebox[.7em]{$\leftrightarrow$}}$}
                  \,_{#2}}}
\newcommand{\bis}[1]{\obis{}{#1}}                       

\newcommand{\us}{me}
\newcommand{\we}{I}
\newcommand{\We}{I}
\newcommand{\our}{my}
\newcommand{\rt}{{\rm t}}                         
\newcommand{\spar}[1]{\mathbin{\|^{}_{#1}}}        
\newcommand{\ini}[1]{\textit{ini}^{}_{#1}}         
\newcommand{\fft}{FT^*}                           
\newcommand{\ffte}{FT^e}                          
\newcommand{\rfft}{FT^{r*}}                        
\newcommand{\rffte}{FT^{re}}                       
\newcommand{\ffteq}{\equiv_{FT}^*}                 
\newcommand{\rffteq}{\equiv_{FT}^{r*}}              
\newcommand{\RS}{\mathcal{S}}                     
\newcommand{\col}{\textit{col}}                   
\newcommand{\st}{\mbox{\scriptsize\sc stab}}      
\newcommand{\pst}{\mbox{\scriptsize\sc poststab}}  

\begin{document}

\title{Failure Trace Semantics{\\} for a Process Algebra with Time-outs}
\titlecomment{This paper is dedicated to Jos Baeten, at the occasion of his retirement.}

\author[R.J. van Glabbeek]{Rob van Glabbeek}
\address{Data61, CSIRO, Sydney, Australia}
\address{School of Computer Science and Engineering,
University of New South Wales, Sydney, Australia}
\email{rvg@cs.stanford.edu}


\begin{abstract}
This paper extends a standard process algebra with a time-out operator, thereby increasing its
absolute expressiveness, while remaining within the realm of untimed process algebra, in the sense
that the progress of time is not quantified. Trace and failures equivalence fail to be congruences
for this operator; their congruence closure is characterised as failure trace equivalence.
\end{abstract}

\maketitle

\section{Motivation}

This work has four fairly independent motivations, in the sense that it aims to provide a common solution
to four different problems. How it yields a solution to the last two of these problems is left for future work, however.

\subsection{The passage of time in processes modelled as labelled transition systems}

The standard semantics of untimed non-probabilistic process algebras is given in terms of labelled
transition systems (LTSs), consisting of a set of states, with action-labelled transitions between them.
The behaviour of a system modelled as a state in an LTS can be visualised as a token, starting in
this initial state, that travels through the LTS by following its transitions. A question that has
plagued me since I first heard of this model, in December 1984, is where exactly (and how) does time
progress when this token travels: in the states, during the transitions, or maybe both?

In some sense this question is more of a philosophical than an practical nature. Process algebra has
flourished, without any convincing answer to this question sitting in its way.

One reasonable attitude is that in untimed progress algebra we \emph{abstract} from time, which makes
the entire question meaningless. While this does make sense, I still desire seeing an untimed
process as a special case of a timed process, one where any period of time that
might be specified in a timed formalism is instantiated by a nondeterministic choice allowing
\emph{any} amount of time. Such a point of view would certainly assist in relating timed and
untimed process algebras. It is from this perspective that my question is in need of an answer.

The common answer given by the process algebra community is that there is no need to model
transitions that take time, for a durational transition can simply be modelled as a pair of
instantaneous transitions with a time-consuming state in between. Here I follow this train of
thought, and regard transitions as occurring instantaneously. This implies that time must elapse in
states, if at all.

Now visualise a system reaching a state $s$ in an LTS, that has a single outgoing transition $u$
leading to a desired goal state. Assume, moreover, that this transition cannot be blocked by the environment.
A fundamental assumption made in most process algebraic formalisms is that the system will
eventually reach this goal state, that is, that it will not stay forever in state $s$.
In \cite{GH19} this assumption is called \emph{progress}. Without it it is not possible
to prove meaningful \emph{liveness properties} \cite{Lam77}, saying that ``something [good] must eventually happen''.

When our LTS would model all activities of the represented system, it is hard to believe how the
system, upon reaching state $s$, comes to a rest, spends a finite amount of time in state $s$ while
doing nothing whatsoever, and then suddenly, and without any clear reason, takes the transition to
the desired goal state. It appears more plausible that if we allow the system to stay in state $s$ for a
finite amount of time, it can (and perhaps even must) stay there for an infinite amount of time.

The way out of this paradox is that our system, being discrete, necessarily abstracts from a lot of activity.
The abstracted activity that is relevant for the visit of the system to state $s$ must be some durational
activity, an amount of work to be done, that for a while sits in the way of taking transition $u$.
As soon as that work is done, and nothing further sits in the way, transition $u$ takes place.

In this paper this abstracted activity is made explicit, in the shape of a time-out transition $\goesto{\rt}$.
Similar to the internal transition $\goesto\tau$, modelling the occurrence of an instantaneous
action from which we abstract, the time-out transition $\goesto{\rt}$ models the end of a time-consuming
activity from which we abstract. A state $s$ in which the system is supposed to spend some time---a
positive and finite, but otherwise unquantified amount---can now be modelled as a pair of states
$s_1$ and $s_2$ with a time-out transition between them; the outgoing transitions of $s$ are then
moved to $s_2$. Time will be spent in state $s_1$ only, where some abstracted time-consuming
activity takes place. The only thing in the model that testifies to this activity is the time-out
transition $s_1\goesto\rt s_2$, that is guaranteed to occur as soon as the time-consuming work is done.
Immediately afterwards, the system will take one of the outgoing transitions of $s_2$, unless they
are all blocked by the environment in which our reactive system is running.

An alternative proposal on how systems spend finite amounts of time in states is essentially due to
Milner \cite{Mi90ccs}. It says that the states in an LTS model \emph{reactive systems} that merely
react on stimuli from the environment. This reaction may be instantaneous, but the environment
itself takes time between issuing stimuli. Usually, each outgoing transition of a state $s$ is
labelled by a visible action $a$, and only when the environment chooses to synchronise on channel
$a$ will such a transition occur. Normally, it may take a while before the environment is ready to
synchronise with any outgoing transition of state $s$, and this is exactly what accounts for the time
spent in state $s$. A special case occurs when state $s$ has an outgoing $\tau$-transition.
Since such transitions do not require synchronisation with the environment, it appears that the
represented system is not allowed to linger in such a state.

The current paper combines the time-out transitions proposed above with the reactive viewpoint adopted by
Milner. So any time the system spends in a state is either due to the environment not allowing an
outgoing transition, or the system waiting for a time-out transition to occur (or both).

Admitting time-out transitions in an LTS not only allows the specification of states in which the
system spends a positive but finite amount of time, as described above; it also allows the
specification of states in which the system spends no time whatsoever, unless forced by the
environment---this is achieved by not using any time-out transitions leaving that state. Moreover,
we can model states $s$ in which some outgoing transitions have to wait for a time-out to
occur, and others do not.\footnote{Enriching the model with this kind of of states is not merely a
  luxury in which {\we} indulge;\vspace{-1pt} such states arise naturally through parallel composition.
  In {\our} interleaving semantics it will turn out that
  $t.a \| t.b \goesto{t} a\|t.b = a.t.b + t.(a.b+b.a)$.}
This is done by modelling $s$ as $s_1\goesto\rt s_2$, were the outgoing
transitions of $s_2$ have to wait for the time-out, and the ones of $s_1$ do not.
In that case the time-out will occur only if the environment fails to
synchronise in time with any of the outgoing transitions from $s_1$.
This implements a priority mechanism, in which from the system's perspective the outgoing
transitions of $s_1$ have priority over the ones of~$s_2$.

\subsection{Failure trace semantics}\label{testing scenarios}

In \cite{vG01,vG93} I classified many semantic equivalences on processes, and proposed testing
scenarios for these semantics in terms of button-pushing experiments on reactive and generative
machines, such that two processes are inequivalent iff their difference can be detected through the
associated testing scenario. Figure~\ref{spectrum} shows the various equivalences in the absence of
internal transitions $\tau$ (or $\rt$). For equivalences such as simulation equivalence (point $S$
in Figure~\ref{spectrum}), or bisimilarity ($B$), the testing scenarios require a replication facility,
allowing the experimenter to regularly make copies of a system \emph{in its current state}, and
expose each of those copies to further tests. This facility could be regarded as fairly unrealistic,
in the sense that one would not expect an actual implementation of such a facility to be available.
When one takes this opinion, (bi)similarity makes distinctions between processes that are not well justified.\linebreak[3]
Likewise, readiness equivalence ($R^*$) has a testing scenario that allows us, in certain states,
to see the menu of all available actions the environment (= tester) can synchronise with.
It may be deemed unrealistic to assume such a menu to be generally available. Or for systems in
which such a menu is available, one might say that it has to be specified as part of the system
behaviour, so that we do not need it explicitly in a testing scenario. When one takes this opinion,
readiness equivalence, and its finer variants, also makes distinctions between processes that are not well justified.

Following this line of reasoning, the finest (= most discriminating) semantic equivalence that does
have a realistic testing scenario is \emph{failure trace semantics} ($FT$). Its testing scenario
allows the tester to control which actions may take place (are available for synchronisation) and
which are not. It also allows the tester to record sequences of actions that take place during a run
of the system, interspersed with periods of idling, where each idle period is annotated with the set
of actions made available for synchronisation by the tester during this period. Such a sequence is a
\emph{failure trace}, and two systems are distinguished iff one has a failure trace that the other
has not. In Figure~\ref{spectrum} two variants of failure trace equivalence are recorded ($FT^*$ and
$FT^\infty$), depending on whether one only allows finite observations, leading to \emph{partial}
failure traces, or also infinite ones.

Failures semantics ($F^*$) is the default semantics of the process algebra CSP \cite{BHR84,Ho85}.
It is coarser than failure trace semantics, in the sense that it makes more identifications.
Its testing scenario is the variant of the one for failure trace semantics described above, in which
the environment/tester cannot alter the set of allowed actions once the system idles (= reaches a
state of deadlock). Thus, when two systems are failures inequivalent, one can think of an
environment such that when placed in that environment one of them deadlocks and the other does not.
An important insight is that such an environment can always be built from some basic process
algebraic operators. So when $P$ are $Q$ are failures inequivalent, there is process algebraic
context $\mathcal{C}[\_\!\_\,]$, such that $\mathcal{C}[P]$ has a deadlock which $\mathcal{C}[Q]$
has not (or vise versa).

Given that two processes $P$ are $Q$ are failure trace inequivalent iff, under a realistic testing scenario,
an environment can see their difference, one would expect, analogously to the situation with failures
semantics, that one can build an environment from process algebraic operators that exploits that
difference, i.e., that one can find a context $\mathcal{C}[\_\!\_\,]$, such that the difference
between $\mathcal{C}[P]$ and $\mathcal{C}[P]$ becomes much more tangible.
However, standard process algebras lack the expressiveness to achieve this.

So far the search for suitable process algebraic operators that allow us to exploit the difference
between failure trace inequivalent processes has not born fruit. Most operators fall short in doing
so, whereas others have too great a discriminating power. A prime example of the latter is the
priority operator of Baeten, Bergstra \& Klop \cite{BBK87b}. It allows building contexts that
distinguish processes whenever they are \emph{ready trace} equivalent. Here ready trace semantics
$(RT^*)$ is a bit finer than failure trace semantics $(FT^*)$. When analysing how priority
operators distinguish processes that are failure trace equivalent, one finds that the priority
operator in fact has unrealistic powers, and is not likely implementable.

The present paper shows that the addition of a time-out operator to a standard process algebra
suffices to exploit the difference between failure trace inequivalent processes. Here the time-out
operator is simply an instance of the traditional action prefixing operator $\alpha.\_\!\_\,$,
taking for $\alpha$ the time-out action $\rt$. {\We} show that with this operator one can build,
for any two failure trace inequivalent processes $P$ and $Q$, a context $\mathcal{C}[\_\!\_\,]$
such that  $\mathcal{C}[P]$ and $\mathcal{C}[P]$ are trace inequivalent. In fact, $P$ and $Q$ can be
distinguished with \emph{may testing}, as proposed by De Nicola \& Hennessy \cite{DH84}.
Without the time-out operator, may testing merely distinguishes processes when they are trace inequivalent.

\subsection{Capturing liveness properties while assuming justness}\label{justness}

In \cite{vG19c} I present a research agenda aiming at laying the foundations of a theory of
concurrency that is equipped to ensure liveness properties of distributed systems without making
fairness assumptions. The reason is that fairness assumptions, while indispensable for some applications,
in many situations lead to false conclusions \cite{GH19,vG19c}. Merely assuming progress, on the other
hand, is often insufficient to obtain intuitively valid liveness properties. As an alternative to
fairness assumptions, \cite{GH19} proposes the weaker assumption of \emph{justness}, that does not
lead to false conclusions.

As pointed out in \cite{vG19c}, when assuming justness but not fairness, liveness properties are not
preserved by strong bisimilarity, let alone by any of the other equivalences in the linear time --
branching time spectrum. An adequate treatment of liveness therefore calls for different semantic
equivalences, ones that do not make all the identifications of strong bisimilarity.

Whereas adapting the notion of bisimilarity to take justness considerations properly into account is
far from trivial, it appears that a version of failure trace semantics that preserves liveness
properties when assuming justness but not fairness comes naturally, and is in a denotational
approach to semantics hard to avoid. Here it is crucial that the semantics is based on
\emph{complete} failure traces, modelling potentially infinite observations of systems.

The present paper paves the way for such an approach by considering \emph{partial} failure traces.
This yields a semantic equivalence that is coarser than strong bisimilarity, and can be justified
from an operational point of view, although it completely fails to respect liveness properties.
The move from partial to complete failure trace equivalence is left for future work.

\subsection{Absolute expressiveness}\label{expressiveness}

In comparing the expressiveness of process algebraic languages, an important distinction between
absolute and relative expressiveness is made \cite{Parrow08}. Relative expressiveness deals with the
question whether an operator in one language can be faithfully mimicked by a context or open term in
another.  It is usually studied by means of valid encodings between languages
\cite{Gorla:unified,vG18e}.  As there is a priori no upper bound on how convoluted an operator one
can add to a process algebra, it may not be reasonable to aim for a universally expressive language,
in which all others can be expressed, unless a thorough discipline is proclaimed on which class of
operators is admissible.

Absolute expressiveness deals with the question whether there are systems that can be denoted by a
closed term in one language but not in another. It can for instance be argued that up to strong
bisimilarity the absolute expressiveness of a standard process algebra like CCS with guarded
recursion is not decreased upon omitting the parallel composition (although its relative
expressiveness surely is). Namely, up to strong bisimilarity, each transition system that can be
denoted by a closed CCS expression can already be denoted by such a CCS expression that does not
employ parallel composition. This consideration occurs in the proof of \thrm{completeness}.
When using absolute expressiveness, the goal of a universally expressive process algebra becomes
much less unrealistic. In fact, when not considering time, probabilities, or other features that are
alien to process algebras like CCS, it could be argued, and is perhaps widely believed, that the
version of CCS with arbitrary infinite sums and arbitrary systems of recursive equations is already
universally expressive. The reason is that all that can be expressed by CCS-like languages are
states in LTSs, and up to strong bisimilarity, each state in each LTS can already be denoted by such
a CCS expression.

This conclusion loses support when factoring in justness, as discussed in Section~\ref{justness},
for in this setting strong bisimilarity becomes too coarse an equivalence to base a theory of
expressiveness upon. In fact, the idea that standard process algebras like CCS are universally
expressive has been challenged in \cite{vG05d} and \cite{GH15b}, where concrete and useful
distributed systems are proposed that can not be modelled in such languages.  In \cite{vG05d} this
concerns a system that stops executing a repeating task after a time-out goes off, whereas the
systems considered in \cite{GH15b} are fair schedulers and mutual exclusion protocols.  In both
cases a proper treatment of justness appears to be a prerequisite for the correct modelling of such systems,
and for this reason this task falls outside the scope of the current paper. However, additionally,
both examples need something extra beyond what CCS has to offer, and both papers indicate that a
simple priority mechanism would do the job.

The current paper offers such a simple priority mechanism, and thereby paves the way for expressing
the systems that have been proposed as witnesses for the lack of universal expressiveness of
standard process algebras.

\section{The process algebra CCSP\texorpdfstring{$_\rt$}{}}

Let $A$ and $V$ be countably infinite sets of \emph{visible actions} and \emph{variables}, respectively.
The syntax of CCSP$_\rt$ is given by\vspace{-1pt}
$$E ::= 0 ~\mbox{\Large $\,\mid\,$}~ \alpha.E ~\mbox{\Large $\,\mid\,$}~ E+E
~\mbox{\Large $\,\mid\,$}~ E \spar{S} E ~\mbox{\Large $\,\mid\,$}~ \tau_I(E) ~\mbox{\Large $\,\mid\,$}~\Rn(E) \mbox{\Large
~$\,\mid\,$}~ X ~\mbox{\Large $\,\mid\,$}~ \rec{X|\RS}\mbox{ (with }X \mathbin\in V_\RS)$$
with $\alpha \mathbin\in Act := A \uplus\{\tau,\rt\}$, $S,I\mathbin\subseteq A$, $\Rn \mathbin\subseteq A \mathop\times A$,
$X \mathbin\in V$ and $\RS$ a {\em recursive specification}: a set of equations
$\{Y = \RS_{Y} \mid Y \mathbin\in V_\RS\}$ with $V_\RS \subseteq V$
(the {\em bound variables} of $\RS$) and $\RS_{Y}$ a CCSP$_\rt$ expression.

The constant $0$ represents a process that is unable to perform any
action. The process $\alpha.E$ first performs the action $\alpha$ and then
proceeds as $E$. The process $E+F$ will behave as either $E$ or $F$.
$\spar{S}$ is a partially synchronous parallel composition operator; 
actions $a\in S$ must synchronise---they can occur only when both arguments
are ready to perform them---whereas actions $\alpha\notin S$ from both arguments are interleaved.
$\tau_I$ is an abstraction operator; it conceals the actions in $I$ by renaming them into the hidden
action $\tau$.
The operator $\Rn$ is a relational renaming: it renames a given action $a\in A$ into a choice between
all actions $b$ with $(a,b)\mathbin\in \Rn$. {\We} require that all sets ${\{b\mid (a,b)\in \Rn\}}$ are finite.
Finally, $\rec{X|\RS}$ represents the $X$-component of a solution of the system of recursive equations $\RS$.
A CCSP$_\rt$ expression $E$ is {\em closed} if every occurrence of a variable $X$ is \emph{bound}, i.e., occurs in a
subexpression $\rec{Y|\RS}$ of $E$ with $X \mathbin\in V_\RS$.

The interleaving semantics of CCSP$_\rt$ is given by the labelled transition relation
$\mathord\rightarrow \subseteq \IP\times Act \times\IP$
on the set $\IP$ of closed CCSP$_\rt$ terms or \emph{processes}, where the transitions 
{$P\goesto{a}Q$} are derived from the rules of \tab{sos CCSP}.
Here $\rec{E|\RS}$ for $E$ an expression and $\RS$ a recursive specification
denotes the expression $E$ in which $\rec{Y|\RS}$ has been substituted for the
variable $Y\!$, for all $Y \mathbin\in V_\RS$.

\begin{table}[t]
\vspace{-6pt}
\caption{Structural operational interleaving semantics of CCSP$_\rt$}
\label{tab:sos CCSP}
\begin{center}
\framebox{$\begin{array}{@{}c@{\qquad}c@{\qquad}c@{}}
\multicolumn{3}{c}{
\alpha.x \goesto{\alpha} x \qquad
\displaystyle\frac{x \goesto{\alpha} x'}{x+y \goesto{\alpha} x'} \qquad
\displaystyle\frac{y \goesto{\alpha} y'}{x+y \goesto{\alpha} y'} \qquad
\displaystyle\frac{x \goesto{\alpha} x'} {\Rn(x) \goesto{\beta} \Rn(x')}
~\left( \begin{array}{@{}r@{}} \scriptstyle\alpha=\beta=\tau \\[-3pt]
\scriptstyle \vee~~ \alpha=\beta=\rt\\[-3pt]\scriptstyle \vee~(\alpha,\beta)\in \Rn\!\end{array}\right)}\\[1.5em]

\displaystyle\frac{x \goesto{\alpha} x'}{x\spar{S} y \goesto{\alpha} x'\spar{S} y}~(\alpha\not\in S) &
\displaystyle\frac{x \goesto{a} x'\quad y \goesto{a} y'}{x\spar{S} y \goesto{a} x'\spar{S} y'} ~(a\in S) &
\displaystyle\frac{y \goesto{\alpha} y'}{x\spar{S} y \goesto{\alpha} x\spar{S} y'}~(\alpha\not\in S) \\[1.5em]

\displaystyle\frac{x \goesto{\alpha} x'}{\tau_I(x) \goesto{\alpha} \tau_I(x')}~(\alpha\not\in I) &
\displaystyle\frac{x \goesto{a} x'}{\tau_I(x) \goesto{\tau} \tau_I(x')}~(a\in I) &
\displaystyle\frac{\rec{\RS_{X}|\RS} \goesto{\alpha} y}{\rec{X|\RS}\goesto{\alpha}y} \\[1em]
\end{array}$}
\end{center}
\end{table}

The language CCSP is a common mix of the process algebras CCS \cite{Mi90ccs} and CSP \cite{BHR84,Ho85}.
It first appeared in \cite{Ol87}, where it was named following a suggestion by M. Nielsen.
The family of parallel composition operators $\|_S$ stems from \cite{OH86}, and incorporates
the two CSP parallel composition operators from \cite{BHR84}.
The relation renaming operators $\Rn(\_\!\_)$ stem from \cite{Va93}; they combine both the
(functional) renaming operators that are common to CCS and CSP, and the inverse image operators of CSP\@.
The remaining constructs are common to CCS and CSP\@. The syntactic form of inaction $0$, action
prefixing $\alpha.E$ and choice $E+F$ follows CCS, whereas the syntax of abstraction $\tau_I(\_\!\_)$
and recursion $\rec{X|\RS}$ follows ACP \cite{BW90}. The only addition by {\us} is the prefixing
operator $\rt.\_\!\_\,$; so far it is merely an action that admits no synchronisation, concealment,
or renaming.

\begin{definition}{bisimulation}
A \emph{strong bisimulation} is a symmetric relation $\R$ on $\IP$, such that, for all
$(P,Q)\mathbin\in\R$ and $\alpha\mathbin\in Act$,\vspace{-2pt}
\begin{itemize}
\item if $P \goesto{\alpha} P'$ then $Q\goesto{\alpha} Q'$ for some $Q'$ with $P'\R Q'$.\vspace{-2pt}
\end{itemize}
Processes $P,Q\mathbin\in\IP$ are \emph{strongly bisimilar}, $P \bis{} Q$, if $P \R Q$ for some strong bisimulation $\R$.
\end{definition}

\noindent
Strong bisimilarity lifts in the standard way to \emph{open} CCSP$_\rt$ expressions, containing free with variables:
$E \bis{} F$ iff $P \bis{} Q$ for each pair of closed substitution instances $P$ and $Q$ of $E$ and $F$.

The common \emph{strong bisimulation semantics} of process algebras like CCSP$_\rt$ interprets
closed expressions as $\bis{}$-equivalence classes of processes, thereby identifying strongly
bisimilar processes. Operators, or more generally open terms $E$, then denote $n$-ary operations on such
equivalences classes, with $n$ the number of free variables occurring in $E$.
To make sure that the meaning of the operators is independent of the choice of
representative processes within their equivalences classes, it is essential that $\bis{}$ is a
congruence for all $n$-ary operators $f$:
\begin{quote}
if $P_i \bis{} Q_i$ for all $i=1,\dots,n$ then $f(P_1,\dots,P_n) \bis{} f(Q_1,\dots,Q_n)$.
\end{quote}
This property holds for CCSP$_\rt$ since the structural operation semantics of the recursion-free
fragment of the language,
displayed in \tab{sos CCSP}, fits the \emph{tyft/tyxt format} of \cite{GrV92}.
Likewise, to make sure that also the meaning of the recursion construct is independent of the choice of
representative expressions within their equivalences classes, $\bis{}$ needs to be \emph{full
  congruence} for recursion \cite{vG17b}:
\begin{quote}
if $\RS_Y \bis{} \RS'_Y$ for all $Y\in V_\RS=V_{\RS'}$ then $\rec{X|\RS} \bis{} \rec{X|\RS'}$
\end{quote}
for all recursive specifications $\RS$ and $\RS'$ and variables $X$ with $X \in V_\RS=V_{\RS'}$.
Again, this property holds since the semantics of \tab{sos CCSP} fits the \emph{tyft/tyxt format
with recursion} of \cite{vG17b}.

\begin{definition}{guarded}
Given a recursive specification $\RS$, write $X \goesto{u} Y$, for variables $X,Y\in V_\RS$, if $Y$
occurs in the expression $\RS_X$ outside of all subexpressions $\alpha.E$ of $\RS_X$.
$\RS$ is \emph{guarded} iff there is no infinite chain $X_0 \goesto{u} X_1 \goesto{u} X_2 \goesto{u} \dots$.
A process is \emph{guarded} if all recursive specifications called by it are guarded.
\end{definition}
\noindent
Often, one restricts attention to guarded processes.
\begin{table}[t]
\vspace{-6pt}
\caption{A complete axiomatisation for strong bisimilarity on guarded CCSP$_\rt$ processes}
\label{tab:axioms CCSP}
\begin{center}
\framebox{$\begin{array}{@{}r@{~=~}l@{\qquad}r@{~=~}l@{}c r@{~=~}l@{}}
    x+(y+z) & (x+y)+z  & \tau_I(x+y) & \tau_I(x) + \tau_I(y)   && \Rn(x+y) & \Rn(x) + \Rn(y)
\\
x+y & y+x          & \tau_I(\alpha.x) & \alpha.\tau_I(x) & \mbox{\small if $\alpha\mathbin{\notin} I$}
                   & \Rn(\tau.x) & \tau.\Rn(x)
\\
x+x & x            & \tau_I(\alpha.x) & \tau.\tau_I(x) & \mbox{\small if $\alpha\mathbin\in I$}
                   & \Rn(\rt.x) & \rt.\Rn(x)
\\
x+0 & 0            & \rec{X|\RS} & \rec{\RS_X | \RS} &
                   & \Rn(a.x) & \plat{\hspace{-1em}$\displaystyle\sum_{\{b\mid (a,b)\in \Rn\}}\hspace{-1em} b.\Rn(x)$}
\\
\multicolumn{7}{l}{\mbox{If $\rule{0pt}{15pt}\displaystyle P= \sum_{i\in I}\alpha_i.P_i$ and $\displaystyle Q= \sum_{j\in J}\beta_j.Q_j$ then}}\\
\multicolumn{7}{c}{\displaystyle P \spar{S} Q = \sum_{i\in I,~\alpha_i \notin S}\alpha_i.(P_i \spar{S} Q) + \sum_{j\in J,~\beta_j\notin S}\beta_j.(P \spar{S} Q_j)
+ \!\!\!\sum_{i\in I,~j\in J,~ \alpha_i=\beta_j\in S}\!\!\! \alpha_i.(P_i \spar{S} Q_j)}\\
\hline
\multicolumn{7}{l}{\mbox{Recursive Specification Principle (RSP)} \qquad\qquad \rule{0pt}{15pt}\RS
  \Rightarrow X = \rec{X|\RS}} \qquad\qquad \mbox{($\RS$ guarded)}
\end{array}$}
\end{center}
\end{table}
The axioms of \tab{axioms CCSP} are \emph{sound} for $\bis{}$, meaning that writing $\bis{}$ for $=$,
and substituting arbitrary expressions for the variables $x,y,z$, or the meta-variables $P_i$ and $Q_j$,
turns them into true statements. In these axioms $\alpha,\beta$ range over $Act$ and $a,b$ over $A$.
All axioms involving variables are equations. The axiom involving $P$ and $Q$ is a template that
stands for a family of equations, one for each fitting choice of $P$ and $Q$. This is the CCSP$_\rt$ version
of the \emph{expansion law} from \cite{Mi90ccs}. The axiom $\rec{X|\RS} = \rec{\RS_X | \RS}$ is the
\emph{Recursive Definition Principle} (RDP) \cite{BW90}. It says that recursively defined processes
$\rec{X|\RS}$ satisfy their set of defining equations $\RS$. In particular, this entails that each
recursive specification has a solution. The axiom RSP \cite{BW90} is a conditional equation, with as
antecedents the equations of a guarded recursive specification $\RS$. It says that the $X$-component of any
solution of $\RS$---a vector of processes substituted for the variables $V_\RS$---equals $\rec{X | \RS}$.
This is equivalent to the statement that the solutions of guarded recursive specifications must be unique.

\begin{theorem}{completeness}
For guarded $P,Q\in\IP$, one has $P \bis{} Q$ iff $P=Q$ is derivable from the axioms of \tab{axioms CCSP}.
\end{theorem}

\begin{trivlist}
\item[\hspace{\labelsep}{\em Proof Sketch.}]
``If'', the \emph{soundness} of the axiomatisation of \tab{axioms CCSP}, is an immediate
consequence of the soundness of the individual axioms.

``Only if'', the \emph{completeness} of the axiomatisation:
Using the axioms from the first box of \tab{axioms CCSP} any guarded CCSP$_\rt$ process $P$ can be
brought in the form $\sum_{i\in I}\alpha_i.P_i$---a \emph{head normal form}.
In fact, $\{(\alpha_i,P_i)\mid i\in I\}$ can be chosen to be $\{(\alpha_i,P_i)\mid P \goesto\alpha P_i\}$.
Given a process $P$, let $\textit{reach}(P)$ be the smallest set of processes containing $P$, such that 
$R\goesto\alpha R'$ with $R \in \textit{reach}(P)$ implies $R'$ in $\textit{reach}(P)$.
Now dedicate to each $R \in \textit{reach}(P)$ a variable $X_R$, and consider the recursive
specification $\RS$ with $V_\RS$ the set of all those variables, and as equations $X_R = \sum_{k\in K}\alpha_k.X_{R_k}$,
using the head normal form of $R$. Using RSP, we derive $P=\rec{X_P|\RS}$. 
Likewise, we can derive $Q=\rec{X_Q|\RS'}$ for a recursive specification $\RS'$, consisting of
equations of the form $X_S = \sum_{l\in L}\beta_l.X_{S_l}$.

Once we have two such recursive specifications, as described for instance in \cite{Mi90ccs} we can
create a combined recursive specification $\RS''$ such that both $P$ and $Q$ are the X-components of
solutions of $\RS''$. With RSP one then derives $P=Q$.\vspace{1ex}
\hfill $\Box$\end{trivlist}
The above proof idea stems from Milner \cite{Mi90ccs} and can be found in various places in the
literature, but always applied to finite-state processes, or some other restrictive subset of a
calculus like CCSP$_\rt$. Since the set of true statements $P \bis{} Q$, with $P$ and $Q$ processes
in a process algebra like CCSP$_\rt$, is well-known to be undecidable, and even not recursively
enumerable, it was widely believed that no sound and complete axiomatisation of strong bisimilarity could exist.
Only in March 2017, Kees Middelburg \cite{Mid17} observed (in the setting of the process
algebra ACP \cite{BW90}) that the above standard proof applies verbatim to arbitrary guarded processes.
His result does not contradict the non-enumerability of the set of true statements $P \bis{} Q$,
due to the fact that RSP is a proof rule with infinitely many premises.

\section{Adding a time-out to CCSP}

The process algebra CCSP, just like CCS, CSP and ACP, is an \emph{untimed} process algebra, meaning
that its semantics abstracts from a quantification of the amounts of time that elapse during or
between the execution of the actions. Untimed process algebras do not deny that the execution of
processes takes time; they abstract from timing merely in not specifying how much time elapses here
and there. An untimed process could therefore be seen as a timed process, in which each occurrence
of a particular period of time is replaced by a nondeterministic choice allowing \emph{any} amount
of time.

From this perspective, one could wonder whether time elapses during the execution of the actions
$\alpha\in Act$, or between the actions, that is, in the states. In some works, notably \cite{GV87},
it is assumed that time happens during the execution of actions. This leads to a rejection of
equations like $a \| b = a.b+b.a$,\footnote{{\We} abbreviate $a.0$ by $a$ and $\spar{\emptyset}$ by $\|$.
Moreover, $+$ binds weaker than $a.\_\!\_\,$, and $\spar{S}$ binds weakest of all.}
which are fundamental for \emph{interleaving semantics} \cite{Mi90ccs,BHR84,Ho85,BW90,OH86},
for the left-hand side allows an overlap in time of the actions $a$ and $b$, whereas the right-hand
side does not.
The default opinion, however, is that the execution of actions is instantaneous, so that time must elapse between the actions.
The following formulation of this assumption is taken from Hoare \cite[Page 24]{Ho85}:
\begin{quote}\small
     The actual occurrence of each event in the life of an object should be regarded as an
     instantaneous or an atomic action without duration. Extended or time-consuming actions should
     be represented by a pair of events, the first denoting its start and the second denoting its
     finish.
\end{quote}
In \cite{GV97} it is pointed out that simply replacing each occurrence of a durational action $\alpha$ by a
sequence $\alpha^+.\alpha^-$ of two instantaneous actions, one denoting the start and one the end of $\alpha$,
is insufficient to adequately capture the durational nature of actions---it even makes a difference
whether actions are split into two or into three parts. The problem arises from the possibility that
two actions $\alpha$ may occur independently in parallel, and the above splitting allows for the
possibility that the end of one such action is confused with the end of the other.
A possible solution is to model a durational action $\alpha$ as a process $\sum_{i\in I}\alpha^+_i.\alpha^-_i$,
where each start action $\alpha^+$ is equipped with a tag $i\mathbin\in I$ that has to be matched by the
corresponding end action $\alpha^-$. When the choice $I$ of tags is large enough, this suffices to
model durational actions in terms of instantaneous actions and durational states. It is for that
reason that in the present paper {\we} focus on time elapsing in states only.

Let $s$ be the state between the $a$- and $\tau$-transitions in the process $a.\tau.P$.
Assume that time may elapse in state $s$, in the sense that after the execution of $a$ and
before the execution of $\tau$, the modelled system, for some positive amount of time, does nothing whatsoever.
Then it is hard to imagine what could possibly trigger the system to resume activity after this amount of time.
A system that does nothing whatsoever may be regarded as dead, and remains dead.

The philosophy of Milner \cite{Mi90ccs} appears to be that once a system reaches state $s$, it will
proceed with the $\tau$-transition immediately. However, once it reaches a state $\sum_{i\in I}a_i.P_i$,
where the $a_i$ are visible actions, it remains idle until the environment of the system triggers
the execution of one of the actions $a_i$. The system is \emph{reactive}, in the sense that visible
actions occur only in reaction to stimuli from the environment. The latter could be modelled as the
environment being a user of the system that may press buttons labelled $a$, for $a\mathbin\in A$.
In a state where the system is not able to engage in an $a$-transition, the user may exercise
pressure on the $a$ button, but it will not go down.
In this philosophy, all time that elapses is due to the system waiting on its environment.

An unsatisfactory aspect of the above philosophy is that there is an asymmetry between a system and
its environment. The latter can spontaneously cause delays, but the former may delay only in 
reaction to the latter. When we would model the environment also as a closed process algebraic
expression, and put it in parallel with the system, in such a way that no action depends on triggers
from outside the system/environment composition, then the resulting composition performs all its
actions instantaneously, possibly until it reaches a deadlock, from which it cannot recover.

The present paper offers an alternative account on the passage of time, with a symmetric view on
systems and their environments. In fact, it tries to be as
close as possible to the philosophy of Milner described above, but adds the expressiveness to model
delays caused by a system rather than its environment. Here there are two overriding design
decisions. First, amounts of time will not be quantified, for {\we} aim to remain in the realm of
untimed process algebras. Second, the assumption of \emph{progress} \cite{GH19} needs to be maintained.
Progress says that from state $s$ above one will sooner or later reach the process $P$, i.e., an
infinite amount of waiting is ruled out. This assumption is essential for the formulation and
verification of liveness properties, saying that a system will reach some goal eventually \cite{GH19}.

Following \cite{vG01,vG93}, {\we} explicitly allow the environment of a system to exercise pressure on
multiple buttons at the same time. The actions for which, at a given time, the environment is
exercising pressure, are \emph{allowed} at this time, whereas the others are \emph{blocked}.
Instead of a model with buttons, one may imagine a model with switches for each visible action, where
the environment, at discrete points in time, can toggle the switches of any action between \emph{allowed}
and \emph{blocked}. In particular the occurrence of a visible action may trigger the environment to
rearrange it switches.

To model time elapsing in a state, {\we} assume that the state may be equipped with one of more time-outs.
A \emph{time-out} is a kind of clock that performs some internal activity, from which {\we} abstract,
and at the end of this activity emits a signal that causes a state-change.
The time-out always runs for an unspecified, but positive and finite amount of time.
For each time-out of a state there is an outgoing transition labelled $\rt$ that models this state-change.
A $\rt$-transition is not observable by the environment. So, apart from the associated time-out, it
is similar to a $\tau$-transition.

The state $a.P + b.Q + \rt.R$, for instance, has one outgoing $\rt$-transition, corresponding with
one time-out. If, upon reaching this state, the environment is allowing one of $a$ or $b$, this action will
happen immediately, reaching the follow-up state $P$ or $Q$; in case the environment allows both $a$
and $b$, the choice between them is made nondeterministically, i.e., by means that are not part of
{\our} model. If, upon reaching the state $a.P + b.Q + \rt.R$, the environment is not allowing $a$
or $b$, the system idles, until either the environment changes and allows $a$ or $b$ after
all, or the time-out goes off, triggering the $\rt$-transition.

A state like $a.P + \rt.Q + \rt.R$ simply has two time-outs; which one goes off first is not predetermined.
If the transition to $R$ occurs, the other time-out is simply discarded.
A state like $a.P + b.Q$ has no time-outs. Upon reaching this state the system idles until the
environment allows $a$ or $b$ to happen.

Note that in $a.P + b.Q + \rt.R$ the action $\rt$ should not be seen as a durational $\tau$-action,
for that would suggest that as soon as that action starts, one has lost the option to perform $a$ or $b$.
Instead, the options $a$ and $b$ remain open until the instantaneous transition to $R$ occurs.

\subsection*{The expressive power of time-outs}

To model that a process like $a.\tau.P$ simply spends some time in the state $s$ right before the
$\tau$-action, one could write $a.\rt.\tau.P$. Furthermore, a durational action $\alpha$,
contemplated above, could be modelled as $\sum_{i\in I}\alpha^+_i.\rt.\alpha^-_i$.
Both uses do not involve placing a process $\rt.P$ in a $+$-context.
The latter is useful for modelling a simple priority mechanism.
Suppose one wants to model a process like $a.P + b.Q$, which can react on $a$- and $b$-stimuli from
the environment. However one has a preference for $a.P$, in the sense that if the environment allows
both possibilities, $a.P$ should be chosen. This can be modelled as
$a.P + \rt.b.Q$ or $a.P + \rt.(a.P + b.Q)$. In an environment allowing both $a$ and $b$, this system
will perform $a.P$. But in an environment just allowing $b$, first some idling occurs, and then $b.Q$.

\subsection*{Three crucial laws}

The most fundamental law concerning time-outs, which could be added as an axiom to \tab{axioms CCSP}, is
\begin{equation}\label{1}
  \tau.P + \rt.Q = \tau.P \;.
\end{equation}
It says that a choice between $\tau.P$ and $\rt.Q$ will always be resolved in favour of $\tau.P$,
because the $\tau$-action is not contingent on the environment and thus can happen immediately,
before the time-out associated with $\rt.Q$ has a chance to go off.

A process $\rt.(\tau.P + b.Q)$ will, after waiting for some time, make a choice between $\tau.P$ and $b.Q$.
In case the environment allows $b$, this choice is nondeterministic. In case the environment blocks $b$ at the
time the time-out associated with the $\rt$-transition goes off, the system will immediately take the
$\tau.P$ branch, thereby discarding $b.Q$. Henceforth {\we} will substitute $\tau_{\{b\}}(P)$ for $P$,
to indicate that this is a process that cannot perform a $b$-action. Now consider two such processes placed in parallel,
synchronising on the action $b$. Each of these processes acts as the environment in which the other
will be placed. In this paper {\we} assume that here the synchronisation on $b$  can \emph{not} occur:
\begin{equation}\label{2}
  \rt.(\tau.\tau_{\{b\}}(P) + b.Q)  \spar{\{b\}} \rt.(\tau. \tau_{\{b\}}(S)+ b.T) =
  \rt.\tau.\tau_{\{b\}}(P)\spar{\{b\}} \rt.\tau.\tau_{\{b\}}(S)\;.
\end{equation}
Intuitively, the reason is that the time-outs on the left and on the right go off at random points
in real time. {\We} explicitly assume no causal link of any kind between the relative timing of these events.
Hence the probability that the two time-outs go off at the very same moment is $0$, and this
possibility can be discarded. So either the left-hand side in the parallel composition reaches
the state $L:=\tau.\tau_{\{b\}}(P) + b.Q$ before the right-hand process reaches the state
$R:=\tau.\tau_{\{b\}}(S) + b.T$, or vice versa.
For reasons of symmetry {\we} only consider the former case.
When the state $L$ is reached, the environment of that process is not
yet ready to synchronise on $b$, so $\tau.\tau_{\{b\}}(P)$ will happen instead. This forces the right-hand side,
when reaching the state $R$, to choose $\tau.\tau_{\{b\}}(S)$, for $\tau_{\{b\}}(P)$ cannot synchronise on $b$.

Algebraically, (\ref{2}) can be derived from (\ref{1}) and the axioms of \tab{axioms CCSP}, when
taking for granted the obvious identity $\tau_{\{b\}}(x) \spar{\{b\}} \rt.(y+b.z) = \tau_{\{b\}}(x)\spar{\{b\}} \rt.y$\,
(whose closed instances with guarded recursion are derivable from the axioms of \tab{axioms CCSP}):
\[\begin{array}{l}
\rt.L  \spar{\{b\}} \rt.R \stackrel{\mbox{\scriptsize (Expansion Law)}}= \\
\rt.\left(L  \spar{\{b\}} \rt.R\right)
    + \rt.\left(\rt.L\spar{\{b\}}R\right)  \stackrel{\mbox{\scriptsize (Expansion Law)}}= \\
\rt.\left(\tau.\left(\tau_{\{b\}}(P)\spar{\{b\}} \rt.R\right)
    + \rt.\left(L  \spar{\{b\}}R\right)\right)
    + \rt.\left(\rt.L\spar{\{b\}}R\right) \stackrel{\mbox{\scriptsize (\ref{1})}}= \\
\rt.\left(\tau.\left(\tau_{\{b\}}(P)\spar{\{b\}} \rt.R\right)
    \right)
    + \rt.\left(\rt.L\spar{\{b\}}R\right) \stackrel{\mbox{\scriptsize (obvious identity; same reasoning on the right)}}= \\
\rt.\left(\tau.\left(\tau_{\{b\}}(P)\spar{\{b\}} \rt.\tau.\tau_{\{b\}}(S)\right) \right)
+ \rt.\left(\tau.\left(\rt.\tau.\tau_{\{b\}}(P)\spar{\{b\}} \tau_{\{b\}}(S)\right) \right) \stackrel{\mbox{\scriptsize (same way back)}}= \\
  \rt.\tau.\tau_{\{b\}}(P)\spar{\{b\}} \rt.\tau.\tau_{\{b\}}(S)\;.
\end{array}\]
To reject (\ref{2}) one would need to take into account seriously the possibility that the two
time-outs take exactly equally long. This corresponds to the addition of the rule
\[\displaystyle\frac{x \goesto{\rt} x'\quad y \goesto{\rt} y'}{x\spar{S} y \goesto{\rt} x'\spar{S} y'}\]
to the structural operational semantics of CCSP$_\rt$, and the summand
$$\mbox{} + \sum_{i\in I,~j\in J,~ \alpha_i=\beta_j=\rt} \rt.(P_i \spar{S} Q_j)$$
to the expansion law of CCSP$_\rt$.
While this could very well lead to an interesting alternative treatment of timeouts,
{\we} will not pursue this option here.

Another law, that could also be added as an axiom to \tab{axioms CCSP}, is
\begin{equation}\label{3}
  a.P + \rt.(Q + \tau.R + a.S) = a.P + \rt.(Q + \tau.R) \;.
\end{equation}
When the process $a.P + \rt.(Q + \tau.R + a.S)$ runs in an environment that allows $a$, the summand 
$\rt.(Q + \tau.R + a.S)$ will not be taken, so it does not matter if there is a summand $a.S$ after
the $\rt$-action. And when this process runs in an environment that keeps blocking $a$ long enough,
the summand  $a.S$ will not be taken. So the only way to execute $a.S$ is in an environment that
initially blocks $a$, but then makes a change to allow $a$.
In this case we need to consider two independently chosen points in time: time $t_e$ where the
environment starts allowing $a$, and time $t_s$ where the time-out goes off that makes the system
take the (unobservable) $\rt$-transition to the state $Q + \tau.R + a.S$. Following the same
reasoning as for law (\ref{2}), the probability that $t_e = t_s$ is $0$, so this possibility can
be ignored. In case $t_e$ occurs before $t_s$, the system will never reach the state $Q + \tau.R + a.S$.
In case $t_e$ occurs after $t_s$, the system will, in state $Q + \tau.R + a.S$, choose $\tau.R$ (or
a branch from $Q$) over $a.S$. In either case, the summand $a.S$ can just as well be omitted.

\subsection*{Time guardedness}

CCSP$_\rt$ allows expressions such as $\rec{X| X=\tau.X}$ or $\rec{X| X=a.X}$ where infinitely many
actions could happen instantaneously. When this is felt as a drawback, one could impose syntactic
restrictions on the language to rule out such behaviour. The following is an example of such a restriction.

\begin{definition}{time guarded}
Given a recursive specification $\RS$, write $X \goesto{t u} Y$, for variables $X,Y\in V_\RS$, if $Y$
occurs in expression $\RS_X$ outside of all subexpressions $\rt.E$ of $\RS_X$.
$\RS$ is \emph{time guarded} iff there is no infinite chain $X_0 \mathbin{\goesto{t u}} X_1 \mathbin{\goesto{t u}} \dots$.
A process is \emph{time guarded} if all recursive specifications called by it are time guarded.
\end{definition}
\noindent
Clearly, the time guarded CCSP$_\rt$ expressions, a subset of the guarded CCSP$_\rt$ expressions,
allow only finitely many actions to occur between any two time-out transitions.

\section{Trace and failures equivalence fail to be congruences}

The paper \cite{vG01} presents a spectrum of semantic equivalences on processes encountered in the literature.
It only deals with \emph{concrete} processes, not featuring the internal action $\tau$.
Of course, the time-out action $\rt$ does not occur in \cite{vG01} either.
Figure~\ref{spectrum} reproduces the spectrum of \cite[Figure~9]{vG01}, while omitting
those semantic equivalences that fail to be congruences for the
$\tau$-free fragment of CCSP---those were coloured red in \cite[Figure~9]{vG01}.%
\begin{figure}[ht]
\input{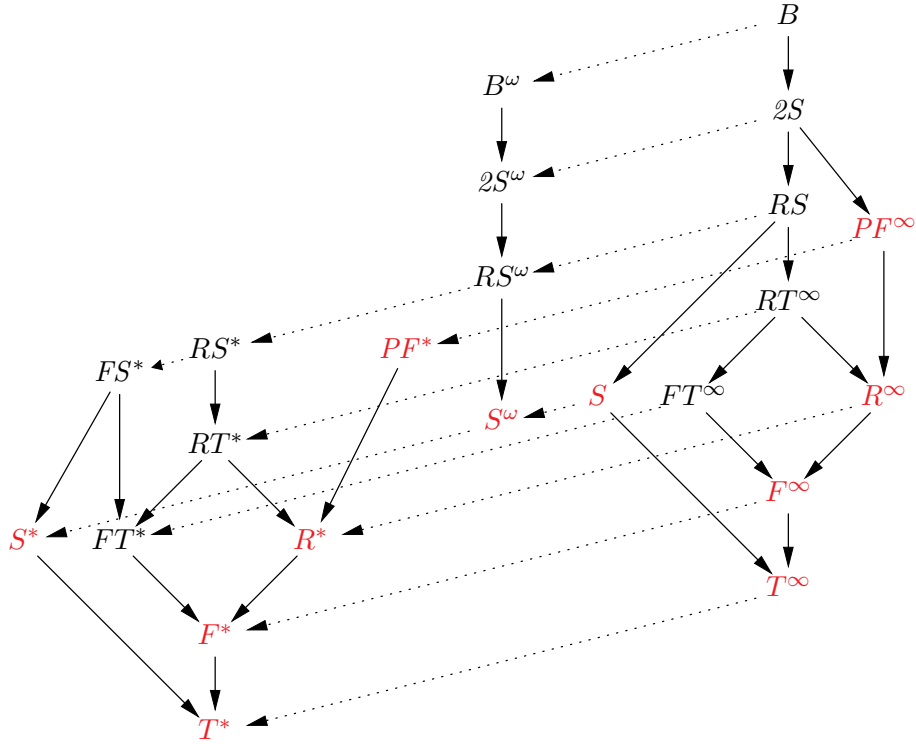}
\centerline{\box\graph}
\caption{The linear time -- branching time spectrum \cite{vG01}}
\label{spectrum}
\end{figure}

(Strong) bisimilarity ($B$) \cite{Mi90ccs,Pa81}, cf.\ \df{bisimulation}, is the finest semantics of
Figure~\ref{spectrum}, i.e., making the fewest identifications between processes, or yielding the smallest equivalence classes.
It sits at the \emph{branching time} side of the spectrum, as it distinguishes processes based on
the timing of choices between different execution sequences, \eg, $a.(b.c+b.d) \neq a.b.c + a.b.d$.

(Partial) trace semantics ($T^*$) \cite{Ho80} is the coarsest semantics of Figure~\ref{spectrum}, i.e., making
the most identifications. It sits at the \emph{linear time} side of the spectrum, as it abstracts
a process into the sequences of action occurrences that can be observed during its runs, its \emph{traces},
and identifies two processes when they have the same traces.

Failures semantics ($F^*$) is the default semantics of the process algebra CSP \cite{BHR84,Ho85}.
When restricting to concrete processes, its falls neatly between $T^*$ and $B$.
Similar to bisimulation semantics, failures semantics rejects the equation $a.(b+c) = a.b + a.c$ that
holds in trace semantics. The main reason for doing so is that in an environment where action $c$
remains blocked, the left-hand side will surely do the $b$-action, whereas the right-hand side may
deadlock instead. However, no such argument distinguishes $a.(b.c+b.d)$ and $a.b.c + a.b.d$, and
hence these processes are identified in failures semantics.

Failure trace semantics ($FT^*$) was proposed in \cite{vG01}, and first presented in \cite{Ba86},
as the natural completion of the square suggested by failures, readiness \cite{OH86} and ready trace semantics \cite{BBK87b}.
For finitely branching processes it coincides with \emph{refusal semantics}, introduced by Phillips in \cite{Ph87}.
Like failure semantics it rejects the identity $a.(b+c) = a.b + a.c$, but accepts $a.(b.c+b.d) = a.b.c + a.b.d$.
The standard example of two processes that are failures equivalent but not failure trace equivalent
is displayed in Figure~\ref{RvsFT}.
\begin{figure}[ht]
\input{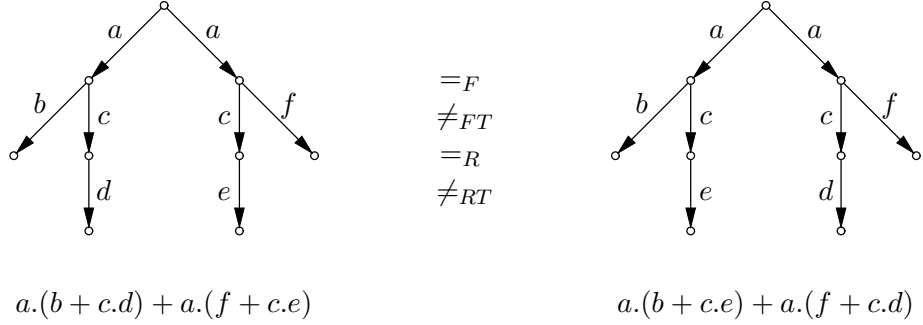}
\centerline{\raise 1em\box\graph}
\caption{Failures and ready equivalent, but not failure trace or ready trace equivalent\label{RvsFT}~\cite{vG01}}
\end{figure}

The following argument shows that neither trace nor failures equivalence are congruences for the
operators of CCSP$_\rt$. That is, there are failures equivalent processes $P$ and $Q$ (that are
therefore certainly trace equivalent) and a recursion-free CCSP$_\rt$-context
$\mathcal{C}[\_\!\_\,]$ such that  $\mathcal{C}[P]$ and $\mathcal{C}[P]$ are trace inequivalent.

\begin{example}{congruence}
Take $P:= a.(b+c.d) + a.(f+c.e)$ and $Q := a.(b+c.e) + a.(f+c.d)$.
\\
Use the context $\mathcal{C}[\_\!\_\,] := \tau_{\{a,b,c\}}(a.(b+\rt.c) \spar{\{a,b,c,f\}} \, \_\!\_ \,)$.
Then only $\mathcal{C}[Q]$ can ever perform the action~$d$.
Namely, the context, enforcing synchronisation on the actions $a$, $b$, $c$ and $f$, implements a priority
of $b$ over $c$ (while blocking $f$ altogether); whenever its argument process $P$ or $Q$ faces a choice between a $b$- and a
$c$-transition, within this context it must take the $b$-transition. This explains why
$\mathcal{C}[P]$ can never reach the action $d$. The abstraction operator $\tau_{\{a,b,c\}}$ is not
redundant, although $\tau_{\{b\}}$ would have worked as well. Without such an operator, the process 
$\mathcal{C}[P]$ might run in an environment in which $b$ is blocked, and here the $d$ could occur
after all.
\end{example}
\noindent
Interestingly, I could make this argument without even providing definitions of trace and failures
equivalence on LTSs that feature time-out transitions. The only two things we need to know about such
equivalences are that they must (a) distinguish a process that can never do the action $d$ from one that can,
and (b) identify the two processes of Figure~\ref{RvsFT}. Property (b) is satisfied when merely
requiring that on concrete process (without $\tau$ and $\rt$) these equivalences agree with the definitions in \cite{vG01}.
The argument is thus also independent of the question how these equivalences are extended to a setting
with the hidden action $\tau$.

The same example, with its above analysis, also shows that ready equivalence ($R^*$) fails to be a
congruence for CCSP$_\rt$, for it also identifies the processes $P$ and $Q$.

In Figure~\ref{spectrum} I have coloured red all those semantics that can be seen to fail the
congruence requirement by a similar argument. For the infinitary versions of trace, failures and
readiness semantics ($T^\infty$, $F^\infty$ and $R^\infty$), the same example applies.
For possible-futures equivalence ($PW^*$ and $PW^\infty$) one takes as witness the two processes of
\cite[Counterexample~7]{vG01}, in combination with the obvious context that gives priority of $b$
over $a$ as well as $c$. Simulation equivalence ($S$) and its finitary variants ($S^*$ and $S^\omega$)
make the identification $a.b + a = a.b$ \cite[Counterexample~2]{vG01}. Using this, the following
argument shows that these semantics also fail to be congruences for CCSP$_\rt$.
\begin{example}{simulation}
Take $P:= a.b + a$ and $Q=a.b$, and use the context
$\mathcal{C}[\_\!\_\,] := \tau_{\{a,b\}}(a.(b+ \rt.d) \spar{\{a,b\}} \, \_\!\_ \,)$.
Then only $\mathcal{C}[P]$ can ever perform the action $d$.
Here $d$ can be done by the context only, rather than by the processes $P$ or $Q$.
However, after synchronising on $a$, the context will proceed towards $d$ only when the environment of
$b+ \rt.d$, which is the process $P$ or $Q$ synchronising on $b$, blocks the $b$-transition.
This is possible for $P$ but not for $Q$.
In this example the possibility to deadlock (or the impossibility to perform a certain action)
is made observable through a context.
\end{example}
\noindent
It follows that failure trace semantics is in fact the coarsest semantics reviewed in \cite{vG01}
that could be a congruence for the operators of CCSP$_\rt$. By the arguments in Section~\ref{testing scenarios}
it also is the finest semantics with a convincing testing scenario, in the sense that all finer
semantics make distinctions that are hard to justify. Together, this tells us that failure trace
semantics may be the only reasonable semantics for CCSP$_\rt$.

The next section proposes an extension of the definition of failure trace semantics (in \cite{vG01}
only given for concrete processes, without the special actions $\tau$ and $\rt$) to arbitrary
LTSs featuring $\tau$ and $\rt$, and thereby to CCSP$_\rt$.
Of course on LTSs with $\tau$ but without $\rt$ it coincides with the definition given in \cite{vG93}.

Then Section~\ref{congruence} shows that, when restricting the choice operator $+$ of CCSP$_\rt$ to
guarded choice, failure trace equivalence is in fact a congruence for the operators of CCSP$_\rt$.
To obtain a congruence for the $+$, a \emph{rooted} version of failure trace equivalence is defined that
records more information on the behaviour of process around their initial states.
This situation is the same as for weak bisimilarity \cite{Mi90ccs} and virtually all other semantic
equivalences that partially abstract from the hidden action $\tau$.

In Section~\ref{coarsest} {\we} then show that for any two failure trace inequivalent processes $P$
and $Q$ one can find a recursion-free CCSP$_\rt$-context
$\mathcal{C}[\_\!\_\,]$ such that  $\mathcal{C}[P]$ and $\mathcal{C}[P]$ are trace inequivalent.
This means that the arguments from Examples~\ref{ex:congruence} and~\ref{ex:simulation} not only
apply to the semantics surveyed in \cite{vG01}, but in fact to \emph{any} semantics that is
incomparable with or coarser than failure trace semantics.

\section{Partial failure trace semantics of CCSP\texorpdfstring{$_\rt$}{}}
\newcommand{\stable}{{\it stable}}

A \emph{finite failure trace} is a sequence $\sigma \in (A \cup \Pow(A))^*$ of actions and sets of actions.
It represents an observation of a system. Each action $a \in A$ occurring in $\sigma$ represents the observation
of the instantaneous occurrence of that action. Each set $X\subseteq A$ occurring in $\sigma$
represents the observation of a period of idling (meaning that no actions occur) during which the
environment allows (i.e., is ready to synchronise with) exactly the actions in $X$.
Such a set is called a \emph{refusal set} \cite{BHR84,vG01}, because it it corresponds with an offer
$X$ from the environment that is refused by the system.
A \emph{partial failure trace} $\sigma \top$ consists of a finite failure trace $\sigma$ followed by
the tag $\top$. It represents a partial observation of a system. The symbol $\top$ indicates the
act, on the part of the observer, of ending the observation. 

\begin{table}[ht]
\caption{Operational failure trace semantics of CCSP$_\rt$}
\label{tab:FT CCSP}
\begin{center}
\framebox{$\begin{array}{c@{\qquad\!}c@{\qquad\!}c}
\top \in \fft(x) &
\displaystyle\frac{x \goesto{a} y \quad \rho \in \fft(y)}{a\rho \in \fft(x)}&
\displaystyle\frac{x \goesto{\tau} y \quad \rho \in \fft(y)}{\rho \in \fft(x)} \\[1.5em]
\displaystyle\frac{\begin{array}{@{}c@{}} x {\ngoesto{\alpha}} \mbox{\footnotesize~for all $\alpha\mathbin\in X \cup\{\tau\}$}\\
      \rho \in \fft(x)\end{array}}{X\rho \in \fft(x)} &
\displaystyle\frac{\begin{array}{@{}c@{}} x {\ngoesto{\alpha}} \mbox{\footnotesize~for all $\alpha\mathbin\in X \cup\{\tau\}$}\\
       x \goesto{\rt} y \quad X\rho \in \fft(y)\end{array}}{X\rho \in \fft(x)} &
\displaystyle\frac{\begin{array}{@{}c@{}} x {\ngoesto{\alpha}} \mbox{\footnotesize~for all $\alpha\mathbin\in X \cup\{\tau\}$}\\
       x \goesto{\rt} y \quad a\rho \in \fft(y)\end{array}}{X a\rho \in \fft(x)} ~(a\mathbin\in X)
\end{array}$}
\end{center}
\end{table}

\tab{FT CCSP} derives the set $\fft(P)$ of partial failure traces of a process $P$.
Here concatenation of sequences is denoted by juxtaposition.
The first rule testifies that it is always possible to end the observation and just record $\top$.
The next two rules say that visible actions are observed, whereas hidden actions are not.
The fourth rule says that, when the environment allows the set of actions $X$, 
idling in a given state $x$ is possible iff from that state no actions $a\in X$, nor the hidden
action $\tau$, can occur. These four rules stem in essence from \cite{vG01,vG93}.
The last two rules deal with the time-out $\rt$. To follow a transition $x \goesto{\rt} y$, the
system must be idling in state $x$ until the time-out occurs. Let $X$ be the set of actions
allowed by the environment at this time. As explained with Law (\ref{3}), {\we} ignore the possibility that
the environment changes the set of allowed actions at the exact same time as the occurrence of the
time-out. So right after the time-out occurs, $X$ is still the set of actions allowed by the environment.
At this time there are two possibilities. Either the system keeps idling, namely when none of the
actions from $X$ can occur in state $y$; or the system performs one of the actions $a\mathbin\in X$
(possibly after some $\tau$-transitions). These two possibilities correspond with the remaining two rules.

\begin{definition}{fft}
Processes $P,Q\mathbin\in \IP$ are \emph{partial failure trace equivalent}, $P\equiv^*_{FT}Q$, iff $\fft(P)\mathbin=\fft(Q)$.
\end{definition}

\subsection*{A path-based characterisation of partial failure traces}

\begin{definition}{path}
  A finite \emph{path} $\pi: P_0 \goesto{\alpha_1} P_1 \goesto{\alpha_2 } \cdots \goesto{\alpha_\ell} P_\ell$
  is an alternating sequence of states/processes and transitions, starting and ending with a state, each
  transition going from the state before it to the state after it. One says it is a path
  \emph{of} $P\in\IP$ if its first state $P_0$ is $P$.

  A process $P\in\IP$ is \emph{stable}, notation $\stable(P)$, if $P{\ngoesto\tau}$, i.e., if it has no outgoing $\tau$-transitions.
  The set $\I(P)$ of \emph{initial visible actions} of a process $P$ is $\{a \in A\mid P{\goesto{a}}\}$,
  provided $P$ is stable.  Here $P{\goesto{a}}$ means that $P \goesto{a} Q$ for some $Q\in\IP$.
  If $P$ is unstable, $\I(P):=\bot$.

  The \emph{concrete ready trace} $\textit{crt}(\pi)$ of a finite path $\pi$ as above is
  $\I(P_0) \alpha_1 \I(P_1) \alpha_2 \cdots \alpha_\ell \I(P_\ell)$.\\
  The (abstract) \emph{ready trace} $\emph{rt}(\pi)$ of $\pi$ is obtained from $\textit{crt}(\pi)$
  by leaving out all occurrences of $\tau$ and $\bot$.
\end{definition}
\noindent
Note that each partial failure trace $\sigma$ of a process $P$ is generated by a path $\pi$ of $P$,
and is in fact completely determined by $\textit{rt}(\pi)$.
Each occurrence of action $a$ in $\sigma$ originates from a transition $\goesto{a}$ occurring in
$\pi$ (by application of the second rule of \tab{FT CCSP}) and each occurrence of a refusal set $X$
in $\sigma$ originates (by application of the fourth or sixth rule) from a stable process $P_i$ in
$\pi$ such that $\I(P_i)\cap X=\emptyset$.
The last symbol $\top$ of $\sigma$ originates from the last state $P_\ell$ of $\pi$ (by application of the first rule).
The mapping from the action occurrences in $\sigma$ to the $A$-labelled transitions in $\pi$ must be bijective,
and any two elements (actions, refusal sets or $\top$) of $\sigma$ occur in the same order in
$\sigma$ as their originators occur in $\pi$. 

Let $\xi$ be the occurrence of an action, refusal set or $\top$ in $\sigma$, originating from $\goesto{\alpha_j}$ or $P_j$.
If $\xi$ is the first element of $\sigma$, the \emph{realm} of $\xi$ is the prefix of $\pi$ ending
in $P_j$. In case $a$ or $Y$ is the symbol in $\sigma$ preceding $\xi$, with $a$ originating from
$\goesto{\alpha_i}$ or $Y$ originating from $P_i$, then $j\geq i$ and the \emph{realm} of $\xi$ is the
subpath $P_i \goesto{\alpha_{i+1} } \cdots \goesto{\alpha_j} P_j$.
Thus, $\sigma$ partitions $\pi$ into realms, with each two adjacent realms having exactly one state in common.

Each transition in the realm of $\top$ must be labelled $\tau$, matching an application of the third
rule of \tab{FT CCSP}. Each transition $\goesto{\alpha_k}$ in the realm of a refusal set $X$ must be
labelled $\rt$ or $\tau$, matching applications of the third and fifth rules; if $\alpha_k\mathbin=\rt$ then
$P_{k-1}$ must be stable and $\I(P_{k-1})\cap X=\emptyset$. Only the first transition in the realm of
an action occurrence $a$ could be labelled $\rt$; if it is then in $\sigma$ action $a$ is preceded
with a refusal set $X$ such that $a\in X$. This matches an application of the sixth rule.

In particular, a path featuring a $\rt$-labelled transition leaving from an unstable state does not
generate any partial failure traces.

\subsection*{System- versus environment-ended periods of idling}

\begin{observation}{doubling}
$\sigma X \rho \in \fft(P) \Leftrightarrow \sigma XX \rho \in \fft(P)$.
\end{observation}

An occurrence of a set $X$ in a partial failure trace $\sigma$ denotes a period $p$ of idling, during
which $X$ is the set of actions allowed by the environment. During period $p$ the system idles
because none of the actions in $X$ is enabled by the system, and none of the others is allowed by
the environment. Such a period can end in exactly two ways:
either by a time-out within the modelled system, or by the environment changing the set of actions it allows.
The latter can be thought of as the occurrence of time-out outside the modelled system.
In this paper {\we} ignore the possibility that a time-out within the system occurs at the exact
same moment as a time-out outside the system. Hence there can never be ambiguity on which party ends
a period of idling.

{\We} will now show how from the shape of a partial failure trace one can deduce which party ends
a given period of idling. If, within $\sigma$, $X$ is followed by
a set $Y\neq X$, the period $p$ is followed by a period $q$ of idling, this time with $Y$ the set
of actions allowed by the environment. The transfer from period $p$ to period $q$ must therefore be
triggered by the environment, namely by changing the set of actions it allows to occur.
If, within $\sigma$, $X$ is followed by an action $a \in X$, the state of the system in which action
$a$ occurs must be different from the state of the system in which $X$ is refused, and action
$a$ not enabled. It follows that period $p$ must be ended by a time-out occurring within the system.
If, within $\sigma$, $X$ is followed by an action $a \notin X$, period $p$ must be ended by the
environment. For if the set of actions allowed by the environment is not changed after period $p$, the
action $a \notin X$ would still not be allowed.

\begin{definition}{system-ended}
An occurrence of $X$ in a partial failure trace $\sigma$ is \emph{system-ended} iff within $\sigma$,
$X$ is followed by an action $a \in X$.
\end{definition}

\section{Congruence properties}\label{congruence}

In this section {\we} show that partial failure trace equivalence is a congruence for the operators
of CCSP$_\rt$, with the exception of the choice operator $+$. {\We} also characterise the coarsest
equivalence included in $\ffteq$ that also is a congruence for $+$.

\begin{definition}{col}
Let $F\subseteq (A \cup \Pow(A))^* \top$, meaning $F$ is a set of partial failure traces with the end tag $\top$.
Then $\col(F)$ is the smallest set containing $F$ such that
$\sigma X X \rho\in \col(F) \Rightarrow \sigma X \rho\in \col(F)$.
\end{definition}

\begin{theorem}{congruence parallel}
$\equiv_{FT}^*$ is a congruence for the parallel composition operators $\spar{S}$.
\end{theorem}

\begin{proof}
{\We} need to show that $P\ffteq P'$ and $Q\ffteq Q'$ implies $P\spar{S}Q \ffteq P'\spar{S} Q'$.
This is equivalent to showing that $\fft(P \spar{S}Q)$ is completely determined by $\fft(P)$, $S$
and $\fft(Q)$.

Below I define, for each $S\subseteq A$, a binary operator $\spar{S}$ that takes as arguments 
sets of partial failure traces, and produces again 
a set of partial failure traces. In the appendix---\pr{spar explicit}---I prove that
\begin{equation}\label{spar explicit}
\fft(P \spar{S} Q) = \col(\fft(P) \spar{S} \fft(Q))\;.
\end{equation}
This yields \thrm{congruence parallel}.
\end{proof}
\noindent
In order to define the operator $\spar{S}$, I analyse how a partial failure trace of $P\spar{S}Q$ decomposes into partial failure
traces of $P$ and $Q$.

\begin{definition}{decomposition}
A \emph{decomposition} of a partial failure trace
$\sigma \mathbin= \sigma_1 \sigma_2 \cdots \sigma_n \top \mathbin\in (A \cup \Pow(A))^*\top$
w.r.t.\ a set $S\mathbin\subseteq A$ is a pair $\sigma_{\rm L},\sigma_{\rm R}$ of partial failure
traces, obtained by leaving out all $\star$-elements from sequences
$\sigma^{\rm L} = \sigma^{\rm L}_1 \sigma^{\rm L}_2 \cdots \sigma^{\rm L}_n \top$
and
$\sigma^{\rm R} = \sigma^{\rm R}_1 \sigma^{\rm R}_2 \cdots \sigma^{\rm R}_n \top \mathbin\in (A \cup \{\star\} \cup \Pow(A))^*\top$,
such that, for all $i=1,\dots,n$,
\begin{itemize}
\item if $\sigma_i \in A\setminus S$ then either $\sigma^{\rm L}_i =\sigma_i$ and $\sigma^{\rm R}_i= \star$,
      or $\sigma^{\rm R}_i =\sigma_i$ and $\sigma^{\rm L}_i= \star$,
\item if $\sigma_i \in S$ then $\sigma^{\rm L}_i = \sigma^{\rm R}_i =\sigma_i$, and
\item if $\sigma_i \in\Pow(A)$ then $\sigma^{\rm L}_i  \in\Pow(A)$ and $\sigma^{\rm R}_i  \in\Pow(A)$.
\end{itemize}
This decomposition is \emph{valid} if, for each $i=1,\dots,n$ for which $\sigma_i \in\Pow(A)$, one has
\begin{align}
\label{consistency}
   \hspace{-9pt}\mbox{$\sigma_i$ is system-ended in $\sigma$} \Leftrightarrow
   \big((\mbox{$\sigma_i^{\rm L}$ is system-ended in $\sigma_{\rm L}$}) \!\vee\! (\mbox{$\sigma_i^{\rm R}$ is system-ended in $\sigma_{\rm R}$})\big)\!\!\!
\\
\label{consistency2}
   \neg \big((\mbox{$\sigma_i^{\rm L}$ is system-ended in $\sigma_{\rm L}$}) \!\wedge\! (\mbox{$\sigma_i^{\rm R}$ is system-ended in $\sigma_{\rm R}$})\big)\!\!\!
\\
\label{decomposition}
   \sigma_i^{\rm L} \setminus S = \sigma_i \setminus S = \sigma_i^{\rm R} \setminus S \qquad\mbox{and}\qquad
   (\sigma_i^{\rm L} \cap S) \cup (\sigma_i^{\rm R} \cap S) = \sigma_i \cap S \;.
\end{align}
For $F,G \subseteq (A \cup \Pow(A))^* \top$ sets of partial failure traces, let
$F\spar{S}G \subseteq (A \cup \Pow(A))^* \top$ be the set of partial failure traces
$\sigma$, such that for some valid decomposition one has $\sigma_{\rm L} \in F$ and $\sigma_{\rm R} \in G$.
\end{definition}
\noindent
\df{decomposition} says that in any partial failure trace $\sigma\in\fft(P \spar{S}Q)$, each occurrence of an action $a\notin S$ in
$\sigma$ must stem from either $P$ or $Q$, whereas each occurrence of an action $a\in S$ must stem
from both.\linebreak[2] A set $X$ occurring in $\sigma$ denotes a period of idling; necessarily both $P$ and $Q$
idle during this period.
I justify (\ref{consistency})--(\ref{decomposition}) informally below; formally, these requirements
are fully justified by the fact that they allow me to prove (\ref{spar explicit}).

In a parallel composition $P\spar{S}Q$ there are three parties that can block actions, or end idle periods: $P$, $Q$
and the global environment $\mathcal{E}$ of the composition. The environment of $P$ consists of $Q$ and $\mathcal{E}$.

An idle period of $P\spar{S}Q$ ends through a single time-out. If this timeout occurs
in $P\spar{S}Q$ itself, rather then in its environment, it must occur in exactly one of the components $P$ and $Q$.
And if the time-out stems from the environment it occurs in neither component.
This explains conditions (\ref{consistency}) and (\ref{consistency2}).

An action $a\notin S$ can, from the perspective of one component, not be
blocked by the other component. So it is blocked by the environment of that component iff it is
blocked by the environment $\mathcal{E}$ of the parallel composition. Hence
\mbox{$a\in \sigma_i^{\rm L} \Leftrightarrow a\in \sigma_i \Leftrightarrow a\in \sigma_i^{\rm R}$}.
This gives the first formula of (\ref{decomposition}).

An action $a\in S$ requires cooperation of all three parties, $P$, $Q$ and $\mathcal{E}$,
so it can be blocked by any of them.
If such an action is in $\sigma_i^{\rm L}$ (resp.\ $\sigma_i^{\rm R}$), it is not blocked by the environment of $P$
(resp.\ $Q$), and thus certainly not by $\mathcal{E}$.
This gives direction $\subseteq$ of the second formula of (\ref{decomposition}).

To justify the other direction, note that each set $X$ occurring in a partial failure trace of a
process $R$ corresponds with path $R_0 \goesto{\alpha_1} R_1 \goesto{\alpha_2} \cdots \goesto{\alpha_n} R_n$
(its realm),
where $R_0$ is reachable from $R$, such that (i) all $\alpha_j$ are either $\tau$ or $\rt$, and (ii)
for each $j\in\{0,\dots,n{-}1\}$ such that $\alpha_{j+1}=\rt$, and also for $j=n$, one has $R_j \ngoesto{\beta}$
for all $\beta\in X \cup\{\tau\}$. In the special case that $R=P\spar{S} Q$, and where $X (= \sigma_i)$ is
decomposed into $X_{\rm L} (=\sigma_i^{\rm L})$ and $X_{\rm R}(=\sigma_i^{\rm R})$, each such $R_j$ has the
form $P_j\spar{S}Q_j$, with $P_j \ngoesto{\beta}$ for $\beta\in X_{\rm L} \cup\{\tau\}$ and 
$Q_j \ngoesto{\beta}$ for $\beta\in X_{\rm R} \cup\{\tau\}$. Since, for $a\in S$, $P_j\spar{S} Q_j\ngoesto{a}$
iff either $P_j\ngoesto{a}$ or $Q_j\ngoesto{a}$, one can always choose the decomposition of $X$ into
$X_{\rm L}$ and $X_{\rm R}$ in such a way that $a\in X$ implies $a \in X_{\rm L} \vee a \in X_{\rm R}$.
This gives direction $\supseteq$ of the second formula of (\ref{decomposition}).

There is one circumstance, however, where the above argument breaks down, namely if different
choices of $j\in \{0,\dots,n\}$ give rise to a different decomposition of $X$ into $X_{\rm L}$ and $X_{\rm R}$.
However, in such a case one can simply write the occurrence of $X$ as $XX\dots X$, such that now
each occurrence of $X$ has a unique decomposition. The original failure trace is then recovered by
the closure operator $\col$ in (\ref{spar explicit}).

The following example shows the necessity of Condition~(\ref{consistency}).

\begin{example}{consistency}
  Consider the process $a \spar{\emptyset} \rt.b$.
  Here $\{b\} a \top \in \fft(a)$ and $\{b\}b\top \in \fft(\rt.b)$.
  So without Condition~(\ref{consistency}) one would obtain $\{b\}ab\top \in\fft(a \spar{\emptyset} \rt.b)$.
  Yet $a \spar{\emptyset} \rt.b$ has no such partial failure trace.
\end{example}
\noindent
Two similar examples, without and with synchronisation, show the necessity of Condition~(\ref{consistency2}).

\begin{example}{consistency2}
  Consider the process $\rt.a \spar{\emptyset} \rt.b$.
  Here $\{a,b\} a \top \in \fft(\rt.a)$ and $\{a,b\}b\top \in \fft(\rt.b)$.
  So without Condition~(\ref{consistency2}) one would obtain $\{a,b\}ab\top \in\fft(\rt.a \spar{\emptyset} \rt.b)$.
  Yet $\rt.a \spar{\emptyset} \rt.b$ has no such partial failure trace.
\end{example}

\begin{example}{consistency2S}
  Consider the process $\rt.(\tau + b) \spar{\{b\}} \rt.(\tau + b)$.
  Here $\{b\} b \top \in \fft(\rt.(\tau + b))$.
  So without Condition~(\ref{consistency2}) one would obtain $\{b\}b\top \in\fft(\rt.(\tau + b) \spar{\{b\}} \rt.(\tau + b))$.
  Yet, as explained at (\ref{2}), $\rt.(\tau + b) \spar{\{b\}} \rt.(\tau + b)$ has no such partial failure trace.
\end{example}
\noindent
The next example shows that the operation $\col$ in (\ref{spar explicit}) cannot be omitted.

\begin{example}{collapse required}
One has $\sigma:=\{b\}\{b\}a\top \in\fft(b+\rt.a)\spar{\{a,b\}}\fft(\rt.(b+\rt.a))$, taking
$\sigma_{\rm L}:= \emptyset\{b\}a\top$ and $\sigma_{\rm R}:= \{b\}\emptyset a\top$. Namely\hfill
\hfill $\top \in \fft(0)$, \hfill $a\top\in\fft(a)$,
\hfill $\{b\}a\top\in\fft(a)$, \hfill $\emptyset\{b\}a\top\in\fft(a)$, \hfill so \hfill $\emptyset\{b\}a\top\in\fft(b+\rt.a)$,
\hfill and \hfill
$\emptyset a\top\in\fft(a)$, \hfill $\emptyset a\top\in\fft(b+\rt.a)$,
\hfill $\emptyset a\top \in\fft(\rt.(b+\rt.a))$, \mbox{} so
\mbox{} $\{b\}\emptyset a\top \in\fft(\rt.(b+\rt.a))$.\\
Yet $\{b\}a\top \notin\fft(b+\rt.a)\spar{\{a,b\}}\fft(\rt.(b+\rt.a))$.
\end{example}

\begin{theorem}{congruence abstraction}
$\equiv_{FT}^*$ is a congruence for the abstraction operators $\tau_I$.
\end{theorem}

\begin{proof}
  A partial failure trace $\rho$ \emph{survives abstraction from} $I\subseteq A$ iff
  \begin{enumerate}[(i)]
  \item each set $X$ occurring in $\rho$ satisfies $I\subseteq X$,
  \item each subsequence $X c_0 \dots c_n Y$ with $c_0,\dots,c_n\in I$ satisfies $X=Y$, and
  \item each subsequence $X c_0 \dots c_n a$ with $c_0,\dots,c_n\in I$ and $a\notin I$ satisfies $a\in X$.
  \end{enumerate}
  If this is the case, then $\tau_I(\rho)$ is the result of contracting any subsequence $X c_0 \dots c_n X$
  of $\rho$ to $X$, and omitting all remaining occurrences of actions $c\in I$.
  If this is not the case $\tau_I(\rho)$ is undefined.
Let $\sigma\cup I$ denote the partial failure trace obtained from a partial failure trace
$\sigma$ by replacing each occurrence of $X$ in $\sigma$ by $X\cup I$. {\We} claim that
\begin{equation}\label{abstraction congruence}
\sigma\in \fft(\tau_I(P)) ~~\Leftrightarrow~~ \exists \rho \in \fft(P).~ \tau_I(\rho) =\sigma\cup I\;.
\end{equation}
This shows that $\fft(\tau_I(P))$ is completely determined by $\fft(P)$ and $I$,
which yields \thrm{congruence abstraction}.\linebreak[2]
A formal proof of (\ref{abstraction congruence}) is given in the appendix---\pr{abstraction explicit}.
\end{proof}
\noindent
Intuitively, an occurrence of $X$ in $\sigma$ denotes a period of idling in which $X$ is the set of
actions allowed by the environment $\mathcal{E}$ of $\tau_I(P)$. The environment $\mathcal{E}_P$ of
$P$ when $P$ is placed in a $\tau_I(\_\!\_)$-context always allows the actions from $I$, and for the
rest is like $\mathcal{E}$. This explains (i) and the use of $\sigma\cup I$ in (\ref{abstraction congruence}).
Conditions (ii) and (iii) express that the instantaneous occurrence of a sequence of internal actions does
not constitute an opportune moment for the environment to change its mind on which actions are allowed.
Since $c_0\in X$ (by (i)), the period $X$ of idling ends through a time-out action of the system, right
before the occurrence of $c_0$. Whereas in $P$ the action $c_0$ could synchronise with the
environment, and thus cause a change in the set of allowed actions, this option is no longer
available for $\tau_I(P)$. So after the sequence of internal actions, the same set of actions $X$ is
allowed by the environment, either resulting in further idling, or in the execution of an action $a\in X$.

An occurrence of $X$ in $\sigma$ imposes the requirement on some reachable states $\tau_I(P')$
of $\tau_I(P)$ that $\tau_I(P')\ngoesto{\alpha}$ for all $\alpha\in X\cup\{\tau\}$. It is
satisfied iff $P'\ngoesto{\alpha}$ for all $\alpha\in X\cup I \cup \{\tau\}$.
The particular way of hiding actions $c\in I$ in $\tau_I(\rho)$, together with Conditions (ii) and
(iii), exactly matches the two rules for bridging a time-out transition in \tab{FT CCSP}.

\begin{theorem}{congruence renaming}
$\equiv_{FT}^*$ is a congruence for the relational renaming operators $\Rn$.
\end{theorem}

\begin{proof}
  For $X\subseteq A$, let $\Rn^{-1}(X):=\{a\in A \mid \exists b \in X.~ (a,b)\in \Rn\}$.
  Furthermore, let $\Rn^{-1}(\sigma)$ denote the set of partial failure traces obtained from $\sigma$
  by (i) replacing each occurrence of a set $X$ by the set $\Rn^{-1}(X)$, and
  (ii) replacing each occurrence of an action $b$ by some action $a$ with $(a,b)\in \Rn$,
  subject to the condition that if the occurrence of $b$ is preceded by an occurrence of a set $X$,
  then $b\in X$ iff $a\in \Rn^{-1}(X)$.  Now, as shown by \pr{renaming explicit} in the appendix,
  \begin{equation}\label{renaming congruence}
  \sigma\in \fft(\Rn(P)) ~~\Leftrightarrow~~ \exists \rho \in \fft(P).~ \rho \in \Rn^{-1}(\sigma) \;.
  \end{equation}
  The side condition ensures that for each period of idling, the party that ends it (system or
  environment) does not differ between $\sigma$ and $\rho$.
  This shows that $\fft(\Rn(P))$ is completely determined by $\Rn$ and $P$, which yields \thrm{congruence renaming}.
\end{proof}
\noindent
The following example shows that the side condition in the definition of $\Rn^{-1}(\sigma)$ cannot be omitted.
\begin{example}{renaming}
Let $P=\rt.a$ and $\Rn=\{(a,b),(a,c)\}$. Then $\rho=\{a\}a\top \in \fft(P)$, so without the side
condition one would obtain $\sigma=\{b\}c\top \in\fft(\Rn(P))$. However, $\Rn(P)$ has no such failure
trace, for if the environment allows just $b$ when the time-out occurs, $\Rn(P)$ will proceed with
$b$ rather than $c$.
\end{example}

\subsection*{Rooted partial failure trace semantics}

Partial failure trace equivalence, however, like all default equivalences that abstract from
internal activity~\cite{Mi90ccs},
fails to be a congruence for the choice operator $+$. In particular, $\fft(b) = \fft(\tau.b)$, yet
$\fft(a+b) \neq \fft(a+\tau.b)$, for only $a+\tau.b$ has a partial failure trace $\{a\}\top$, and
only $a+b$ has a partial failure trace $\emptyset a\top$.
The classical way to solve this problem for languages like CCSP is to add one bit of information to
the semantics of processes. Two CCSP processes could be called \emph{partial failure trace congruent}
iff $\fft(P)=\fft(Q) \wedge \textit{stable}(P)\Leftrightarrow\textit{stable}(Q)$.
This turns out to be a congruence for CCSP, and moreover the coarsest congruence finer than partial failure trace equivalence.
In the presence of time-outs this one-bit modification of partial failure trace equivalence is insufficient,
for $\fft(\rt.b) = \fft(\rt.\rt.b)$, yet $\fft(a+\rt.b) \neq \fft(a+\rt.\rt.b)$. Namely, only $a+\rt.\rt.b$
has a partial failure trace $\{b\}\{a,b\}b$.

To make partial failure trace equivalence a congruence, one needs additionally to keep track of
partial failure traces of the form $X\sigma$ that are lifted by the fifth rule of \tab{FT CCSP}
through an initial $\rt$-transition. {\We} also use an additional bit, telling that a stable state
is reachable along a nonempty path of $\tau$-transitions.
\begin{definition}{rfft}
Let $\rfft(P) :=$\vspace{-1ex} \[\fft(P) \begin{array}[t]{@{}l@{}}
\mbox{} \cup \{\st \mid P {\ngoesto\tau}\} \cup \{\rt X \sigma \mid \exists P'.~ P\goesto\rt P' \wedge X \sigma\in\fft(P') \wedge P {\ngoesto\tau} \wedge \I(P)\cap X=\emptyset\}\\
\mbox{} \cup \{\pst \mid P{\goesto\tau} \wedge \emptyset \top\in\fft(P) \}.
\end{array}\vspace{-1ex}\]
Processes $P,Q\in \IP$ are \emph{rooted partial failure trace equivalent}, $P \rffteq Q$, iff $\rfft(P)=\rfft(Q)$.
\end{definition}
\noindent
One has $\rt.b \not\rffteq \rt.\rt.b$ because $ \rt \{a,b\} b \in \rfft(\rt.\rt.b)$ but $\rt \{a,b\} b \notin \rfft(\rt.b)$.

For the purpose of defining $\rffteq$, the $\pst$ bit is redundant,
as it is clearly determined by the other elements of $\rfft(P)$.
However, it will turn out to be useful in defining a rooted partial failure trace preorder
$\sqsubseteq^{r*}_{FT}$ in Section~\ref{preorders}.

\begin{observation}{rfft}
Rooted partial failure trace equivalence is finer than partial failure trace equivalence.
\end{observation}

\noindent
Let $\ini{Z}(F)$, for $Z\mathbin\subseteq A$ and $F\mathbin\subseteq (A \cup \Pow(A))^* \top$, be the least set such that
(a) $F\mathbin\subseteq \ini{Z}(F)$, (b) $\top\mathbin\in \ini{Z}(F)$, and (c) if $\sigma \in \ini{Z}(F)$ and $X\cap Z = \emptyset$ then
$X\sigma\in \ini{Z}(F)$. It represents the set of partial failure traces of a stable process $P$
derivable by the first and fourth rule of \tab{FT CCSP} when $F\mathbin\subseteq \fft(P)$ and $\I(P)\mathbin=Z$.

\begin{theorem}{congruence}
Rooted partial failure trace equivalence ($\rffteq$) is a congruence for the operators of CCSP$_\rt$.
\end{theorem}

\begin{proof}
It suffices to show that $\rfft(P+Q)$ is fully determined by $\rfft(P)$ and $\rfft(Q)$, that
$\rfft(\alpha.P)$ is fully determined by $\alpha$ and $\rfft(P)$, and similarly for $\rfft(P\spar{S}Q)$, $\rfft(\tau_I(P))$ and $\rfft(\Rn(P))$.
The following equations establish this.
\begin{align}
\rfft(a.P) ~~=~~&  \ini{\{a\}}\big(\{a\sigma \mid \sigma \in \fft(P) \}\big) \cup \{\st\}\label{action}\\
\rfft(\tau.P) ~~=~~& \fft(P) \cup \{\pst \mid \emptyset\top \in\fft(P)\}\\
\rfft(\rt.P) ~~=~~& \begin{array}[t]{@{}l@{}}\ini\emptyset\big(\{X \sigma \mid X \sigma \in \fft(P)\} \cup
                                            \{X a \sigma \mid a\sigma \in \fft(P) \wedge a\in X\}\big)\\
                     \mbox{} \cup \{\st\} \cup \{\rt X \sigma \mid X \sigma \in \fft(P)\}\end{array}\label{time-step}\\
\rfft(P+Q) ~~=~~& \!\!\left\{\begin{array}{@{}ll@{}}
                        \rfft(P) \cup \rfft(Q) & \mbox{if}~~ \neg\stable(P) \wedge \neg\stable(Q)\\
                        \{a\sigma\mid a\sigma\in \rfft(P)\} \cup \rfft(Q) & \mbox{if}~~ \stable(P) \wedge \neg\stable(Q)\\
                        \rfft(P) \cup \{a\sigma\mid a\sigma\in \rfft(Q)\} & \mbox{if}~~ \neg\stable(P) \wedge \stable(Q)\\
                        \ini{\I(P)\cup \I(Q)}\big(H\big) \cup \{\st\} \cup K & \mbox{if}~~ \stable(P) \wedge \stable(Q)\\
                        \end{array}\right.\label{sum}\\
\rfft(P\spar{S}Q) ~~=~~&  \col(\rfft(P) \spar{S} \rfft(Q)) \label{rpar}
\end{align}
\begin{align}
\rfft(\tau_I(P)) ~~=~~& \begin{array}[t]{@{}ll@{}}
                        \{\sigma \mid (\st\neq\sigma\neq\pst) \wedge \exists \rho \in \rfft(P).~ \tau_I(\rho) =\sigma\cup I\} \\
                        \mbox{}\cup \{\st\mid\st\in\rfft(P) \wedge \I(P)\cap I=\emptyset\} \\
                        \mbox{}\cup \{\pst \mid \begin{array}[t]{@{}ll@{}}
                        (\exists c_0c_1\dots c_n I\top \in \rfft(P) \mbox{~where~}n\geq 0 \mbox{~and all~}c_i\in I)\hspace{-2em}\\
                          \mbox{} \vee (\pst\in \rfft(P) \wedge I\top\in \rfft(P) )\}
                        \end{array}
                        \end{array} \label{rabs}\\
\rfft(\Rn(P)) ~~=~~& \{\sigma \mid \exists \rho \in \rfft(P).~ \rho \in \Rn^{-1}(\sigma)\} \label{rren}
\end{align}
Here $H:= \begin{array}[t]{@{}l@{}}
                  \{a\sigma\mid a\sigma\in \rfft(P) \cup \rfft(Q)\} \cup \mbox{}\\
                  \{X \sigma \mid \rt X \sigma \in \rfft(P) \cup \rfft(Q) \wedge (\I(P)\cup \I(Q)) \cap X=\emptyset\}  \cup \mbox{}\\
                  \{X a\sigma \mid X a\sigma\in \rfft(P)\cup \rfft(Q) \wedge a \mathbin\in X \wedge (\I(P)\cup \I(Q)) \cap X=\emptyset\}
                  \end{array}$\\
and $K= \{\rt X \sigma \mid \rt X \sigma \in \rfft(P) \cup \rfft(Q) \wedge (\I(P)\cup \I(Q)) \cap X=\emptyset\}$.
\\[1ex]
The first four equations follow immediately from \df{rfft} and \tab{FT CCSP}.
In (\ref{time-step}) the argument of $\ini\emptyset$ generates the partial failure traces
contributed by the fifth and sixth rule of \tab{FT CCSP}; the operation $\ini\emptyset$ then closes
this set under applications of the first and fourth rule.
In the fourth case of (\ref{sum}), the three lines of $H$ generate the partial failure traces
contributed by the second, fifth and sixth rule of \tab{FT CCSP}, respectively; the operation
$\ini{\I(P)\cup \I(Q)}$ then closes this set under applications of the first and fourth rule.

To see that (\ref{action})--(\ref{sum}) fulfil their intended purpose, note that $\fft(P)$ is fully determined by $\rfft\!(P)$,
namely by dropping the bits $\st$ and $\pst$, and all traces $\rt\sigma$.
Moreover $\stable(P)$ holds iff $\st\mathbin\in\rfft(P)$ and for stable $P$ the set $\I(P)$ can be computed
by $\I(P) = \{a\mathbin\in A \mid \exists \sigma.~ a\sigma \mathbin\in \rfft(P)\}$.

To explain (\ref{rpar}), the definition of decomposition needs to be extended to rooted partial
failure traces $\rt X\sigma$. In the first bullet-point of \df{decomposition},
``$\sigma_i\in A\setminus S$'' should now be read as ``$\sigma_i\in (A\setminus S) \cup \{\rt\}$''.
The definition of $F \spar{S} G$ is unchanged, except that rooted
partial failure traces $\rt X \sigma$ are dropped from $F \spar{S} G$ unless $\st\in F$ and $\st \in G$;
furthermore $\st$ is deemed to be in  $F \spar{S} G$ iff it is in both $F$ and $G$, and
$\pst$ is deemed to be in  $F \spar{S} G$ if either $\pst\in F \wedge \pst\in G$, or 
$\pst\in F \wedge \st\in G$, or $\st\in F \wedge \pst\in G$.
Moreover, the operator $\col$ extends to sets of rooted partial failure traces.
In the appendix it is shown that with these definitions (\ref{rpar}) holds---see \pr{spar explicit rooted}.

In (\ref{rabs}), the concept ``survives abstraction from $I\subseteq A$'' from the proof of
\thrm{congruence abstraction} is applied verbatim to rooted partial failure traces $\rt X\sigma$,
as is the operation $\tau_I$. The validity of (\ref{rabs}) is shown in the
appendix---\pr{abstraction explicit rooted}.

In (\ref{rren}), the operator $\Rn^{-1}$ extends verbatim to rooted partial failure traces.
The validity of (\ref{rabs}) is shown in the appendix---\pr{renaming rooted}.
\end{proof}
As remarked in \cite{vG94a}, the countable choice operator $\sum_{i\in \INs}P_i$ is expressible in terms
of the binary choice operator $+$ and recursion:
$\sum_{i\in \INs}P_i = \rec{X_0|\RS}$, where $X_i = P_i + X_{i+1}$ for $i\in\IN$.
An equation similar to (\ref{sum}) shows that $\rffteq$ also is a congruence for the countable choice.

\begin{theorem}{full abstraction +}
Rooted partial failure trace equivalence ($\rffteq$) is the largest congruence for the operators of CCSP$_\rt$
that is included in $\ffteq$.
\end{theorem}

\begin{proof}
In view of \obsref{rfft} and Theorem~\ref{thm:congruence} it suffices to find, for any two
rooted partial failure trace inequivalent processes $P \not\rffteq Q$, a CCSP$_\rt$-context
$\mathcal{C}$ such that $\mathcal{C}(P) \not\ffteq \mathcal{C}(Q)$.

First of all, one can apply a bijective renaming operator on $P$ and $Q$, mapping $A$ to $A\setminus\{f\}$
for some action $f\mathbin\in A$. Consequently, {\we} may assume that $f$ is \emph{fresh} in the sense that
for each process $R$ reachable from $P$ or $Q$ one has $R{\ngoesto{f}}$.
For any such $R$ one has
\[
\sigma X \rho \in \fft(R) \Leftrightarrow \sigma (X{\cup}\{f\}) \rho \in \fft(R).
\tag{*}
\]

Let $\sigma \in \rfft(Q) \setminus \rfft(P)$. (The other case proceeds by symmetry.)
The only relevant cases are that $\sigma\mathbin=\mbox{\scriptsize\sc (post)}\st$ or $\sigma\mathbin=\rt X\rho$,
since in all other cases one immediately has $\sigma \in \fft(Q) \setminus \fft(P)$.

Let $\sigma =\st$, that is, $Q$ is stable but $P$ is not. Then $\emptyset f \top \in \fft(Q+f)$ but
$\emptyset f \top \notin \fft(P+f)$. So it suffices to take the context $\mathcal{C}(\_\!\_\,) = \_\!\_\, + f$.

Now let $\sigma\mathbin=\rt X\rho$. So $Q{\ngoesto\alpha}$
for all $\alpha \mathbin\in X\cup\{\tau\}$, $Q \mathbin{\goesto\rt} Q'$ and $X\rho\mathbin\in\fft(Q')$.
{\We} may assume that $P$ is stable, as the case that $Q$ is stable but $P$ is not has been treated already.
By \obsref{doubling} one has $XX\rho\in\fft(Q')$, and (*) yields $X^- X^+ \rho \in\fft(Q')$, where
$X^- := X\setminus \{f\}$ and $X^+ := X\cup\{f\}$.
The fifth rule of \tab{FT CCSP} gives $X^- X^+ \rho \in\fft(Q{+}f)$, using that $Q{+}f{\ngoesto\alpha}$
for all $\alpha \in X^-\cup\{\tau\}$.
It suffices to show that $X^- X^+ \rho\notin\fft(P{+}f)$, for then the context
$\mathcal{C}(\_\!\_\,) = \_\!\_\, + f$ finishes the proof.

So assume, towards a contradiction, that $X^- X^+ \rho\in\fft(P{+}f)$.
This could not have been derived by the fourth rule of \tab{FT CCSP}, since surely
$X^+ \rho\notin\fft(P{+}f)$. Hence the fifth rule must have been used, so that
$P{\ngoesto\alpha}$ for all $\alpha \mathbin\in X\cup\{\tau\}$,
$P \mathbin{\goesto\rt} P'$ and $X^- X^+\rho\mathbin\in\fft(P')$.
Now $X X \rho \in \fft(P')$ by (*), and $X\rho \mathbin\in \fft(P')$ by \obsref{doubling}.
Thus $\rt X\rho \mathbin\in \rfft(P)$, yielding the required contradiction.

Finally, let $\sigma\mathbin=\pst$. Then $\{f\}\top \in \fft(Q+f)$, yet $\{f\}\top \notin \fft(P+f)$.
\end{proof}
\noindent
In the quest for a congruence, an alternative solution to changing partial failure trace equivalence
into rooted partial failure trace equivalence, is to restrict the expressiveness of CCSP$_\rt$.
Let CCSP$^g_\rt$ be the version of  CCSP$_\rt$ where the binary operator $+$ is replaced by guarded
sums $\sum_{i=1}^n\alpha_i. P_i$, for $n\mathbin\in\IN\cup\{\infty\}$. 
Here each guarded sum $\sum_{i=1}^n\alpha_i.\_\!\_\,$ counts as a separate $n$-ary operator.
The constant $0$ and action prefixing $\alpha.P$ are the special cases with $n=0$ and $n=1$.
The congruence property for guarded sums demands
that $P_i\ffteq Q_i$ for $i=1,\dots,n$ implies $\sum_{i=1}^n\alpha_i.P_i = \sum_{i=1}^n\alpha_i. Q_i$.
That this holds is an immediately consequence of \thrm{fft precongruence} below.

\section{Partial failure trace preorders}\label{preorders}

Here {\we} introduce (rooted) partial failure trace preorders $\sqsubseteq^*_{FT}$ and
$\sqsubseteq^{r*}_{FT}$ in such a way that (rooted) partial failure trace equivalence is its
kernel, i.e., $P \ffteq Q \Leftrightarrow (P\sqsubseteq^*_{FT} Q \wedge Q \sqsubseteq^*_{FT} P)$
and, likewise, $P \rffteq Q \Leftrightarrow (P\sqsubseteq^{r*}_{FT} Q \wedge Q \sqsubseteq^{r*}_{FT} P)$.
These preorders should be defined in such a way that they are precongruences.

\begin{definition}{fftpre}
Write $P\sqsubseteq^*_{FT}Q$, for processes $P,Q\mathbin\in \IP$, iff $\fft(P)\supseteq \fft(Q)$.
\end{definition}
\noindent
The orientation of the symbol $\sqsubseteq^*_{FT}$ aligns with that of the safety preorder,
discussed in the next section.
The following proposition says that the preorder $\sqsubseteq^*_{FT}$ is a \emph{precongruence} for the
operators of CCSP$^g_\rt$, or that recursion-free CCSP$^g_\rt$-contexts are \emph{monotone} w.r.t.\ $\sqsubseteq^*_{FT}$.
\begin{theorem}{fft precongruence}
  Let $\mathcal{C}$ be a unary recursion-free CCSP$^g_\rt$-context.\\ If $P\sqsubseteq^*_{FT}Q$ then
  $\mathcal{C}(P)\sqsubseteq^*_{FT}\mathcal{C}(Q)$.
\end{theorem}
\begin{proof}
   Let $G:=\sum_{i\in I}\alpha.P_i$. In case $\alpha_i\neq\tau$ for all $i\in I$, then
   (\ref{action})--(\ref{sum}) imply (or generalise to)\vspace{-1ex} {\small
$$\fft(G) = \ini{I(G)}\!\left(\bigcup_{\begin{array}{@{}c@{}}\raisebox{5pt}[0pt][0pt]{$\scriptstyle i\in I$} \\[-11pt] \scriptstyle\alpha_i\in A\end{array}} \{\alpha_i\sigma \mid \sigma \mathbin\in \fft(P_i)\} \cup
            \bigcup_{\begin{array}{@{}c@{}}\raisebox{5pt}[0pt][0pt]{$\scriptstyle i\in I$} \\[-11pt] \scriptstyle\alpha_i=\rt \end{array}} \left(
                \begin{array}{@{}l@{}}
                \{X\sigma \mid X\sigma \mathbin\in \fft(P_i) \wedge \I(G)\cap X = \emptyset\} \cup \mbox{}\\
                \left\{X a\sigma \left|\, a\sigma \mathbin\in \fft(P_i) \begin{array}{@{}l}\mbox{}\wedge \I(G)\cap X = \emptyset\\
                                                  \mbox{}\wedge a\mathbin\in X\end{array}\right\}\right.
                \end{array}\right)\!\right)\!.$$}
   Here $\I(G) :=\{\alpha_i \mid i\in I\}\setminus\{\rt\}$.
   Moreover, if $\alpha_i=\tau$ for at least one $i\in I$, then
$$\fft(G) = \bigcup_{\begin{array}{@{}c@{}}\raisebox{5pt}[0pt][0pt]{$\scriptstyle i\in I$} \\[-11pt] \scriptstyle\alpha_i\in A\end{array}} \{\alpha_i\sigma \mid \sigma \in \fft(P_i)\} \cup
             \bigcup_{\begin{array}{@{}c@{}}\raisebox{5pt}[0pt][0pt]{$\scriptstyle i\in I$} \\[-11pt] \scriptstyle\alpha_i=\tau\end{array}} \fft(P_i)\;.$$
  There equations immediately imply the monotonicity of the guarded choice operators w.r.t.\ $\sqsubseteq^*_{FT}$.
  The monotonicity of $\spar{S}$, $\tau_I$ and $\Rn$ follows from Properties (\ref{spar explicit}),
  (\ref{abstraction congruence}) and (\ref{renaming congruence}).
\end{proof}

\begin{definition}{rfftpre}
  Write $P\sqsubseteq^{r*}_{FT}Q$, for $P,Q\mathbin\in \IP$, iff $\rfft(P)\supseteq\rfft(Q)$.
\end{definition}
\noindent
In this definition, the $\pst$ bit of \df{rfft} plays a crucial r\^ole. Without it, we
would have \mbox{$b \sqsubseteq^{r*}_{FT} \tau.b$}, since $\rfft(b)= \rfft(\tau.b) \uplus \{\st\}$. This would cause a
failure of monotonicity, as $a+b \not\sqsubseteq^{r*}_{FT} a+\tau.b$, for only the latter process has a partial
failure trace $\{a\}\top$. With the $\pst$ bit one obtains
$b \not\sqsubseteq^{r*}_{FT} \tau.b$, since $\pst\in\rfft(\tau.b)$, yet $\pst\notin\rfft(b)$.
The preorder $\sqsubseteq^{r*}_{FT}$ still relates processes of which only one is stable;
for instance $b \sqsubseteq^{r*}_{FT} \rec{X|X=b+\tau.X}$.

\begin{lemma}{initials equal}
If $P\sqsubseteq^{r*}_{FT}Q$ (or $P\sqsubseteq^{*}_{FT}Q$) and $P$ and $Q$ are both stable, then $\I(P)=\I(Q)$.
\end{lemma}
\begin{proof}
If $a\in \I(Q)$ then $a\top\in \rfft(Q)\subseteq \rfft(P)$, so $a\mathbin\in \I(P)$.\\
Moreover, if $a\mathbin{\notin} \I(Q)$ then $\{a\}\top\mathbin\in \rfft(Q)\subseteq \rfft(P)$, so $a\mathbin{\notin} \I(P)$.
\end{proof}

\begin{theorem}{rfft precongruence}
  Let $\mathcal{C}$ be a unary recursion-free CCSP$_\rt$-context.\\ If $P\sqsubseteq^{r*}_{FT}Q$ then
  $\mathcal{C}(P)\sqsubseteq^{r*}_{FT}\mathcal{C}(Q)$.
\end{theorem}
\begin{proof}
  The monotonicity of $\spar{S}$, $\tau_I$ and $\Rn$ follows from Properties
  (\ref{rpar})--(\ref{rren}), using \lemref{initials equal}.
  The monotonicity of $\alpha.\_\!\_\,$ follows from (\ref{action})--(\ref{time-step}),
  using that $\rfft(P)\mathbin\supseteq\rfft(Q)$ implies $\fft(P)\mathbin\supseteq\fft(Q)$.

  The monotonicity of $+$ can be formulated as $P\sqsubseteq^{r*}_{FT}Q \Rightarrow P+R\sqsubseteq^{r*}_{FT}Q+R$
  (the requirement $R+P\sqsubseteq^{r*}_{FT}R+Q$ then follows by symmetry).
  So let $P\sqsubseteq^{r*}_{FT}Q$, i.e., $\rfft(P)\supseteq\rfft(Q)$.
  \begin{itemize}
  \item Let $\neg\stable(P)$ and $\neg\stable(R)$. Then $\st\notin\rfft(P)\supseteq\rfft(Q)$, so $\neg\stable(Q)$.
    Hence\vspace{-1ex}
     \[\rfft(P+R) = \rfft(P) \cup \rfft(R) \supseteq  \rfft(Q) \cup \rfft(R) = \rfft(Q+R).\vspace{-1ex}\]
  \item The case $\neg\stable(P)$ and $\stable(R)$ proceeds likewise.\\[1ex]
    \mbox{}\hspace{-10pt}In all remaining cases {\we} assume $\stable(P)$.
    Note that $\pst\notin\rfft(P)\supseteq\rfft(Q)$.
  \item Let $\stable(Q)$ and $\neg\stable(R)$. Then
    $\rfft(P+R) = \{a\sigma \mid a\sigma\in \rfft(P)\} \cup \rfft(R)\linebreak \supseteq \{a\sigma \mid a\sigma\in \rfft(Q)\} \cup \rfft(R) = \rfft(Q+R).$
  \item Let $\neg\stable(Q)$ and $\neg\stable(R)$. Now all $\sigma\in \rfft(Q)$ must have the form $\top$ or $a\zeta$.
    Namely if $X\zeta \in \rfft(Q)$ then one would have $\pst\in\rfft(Q)$.
    Hence\\
    $\rfft(P+R) = \{a\sigma \mid a\sigma\in \rfft(P)\} \cup \rfft(R) \supseteq \rfft(Q) \cup \rfft(R) = \rfft(Q+R)$.
  \item Let $\stable(Q)$ and $\stable(R)$. Then $\I(P)=\I(Q)$ by \lemref{initials equal}.   
    Using this, the desired monotonicity property follows from the fourth case of (\ref{sum}).
  \item Let $\neg\stable(Q)$ and $\stable(R)$. Again all $\sigma\in \rfft(Q)$ must have the form $\top$ or $a\zeta$.
    Let $H$ and $K$ be as in (\ref{sum}), but with $R$ substituted for $Q$.
    Now
    \begin{align*}
    \rfft(P+R) &= \ini{\I(P){\cup} \I(Q)}\big(H\big) {\cup} \{\st\} {\cup} K \supseteq H {\cup}\{\top\}
    \\& \supseteq \rfft(Q) {\cup}  \{a\sigma \mid a\sigma\mathbin\in \rfft(R)\} = \rfft(Q{+}R).
    \tag*{\qedhere}
    \end{align*}
\end{itemize}
\end{proof}

\begin{theorem}{full abstraction preorder}
$\sqsubseteq^{r*}_{FT}$ is the largest precongruence for the operators of CCSP$_\rt$
that is included in $\sqsubseteq^{*}_{FT}$.
\end{theorem}

\begin{proof}
\thrm{rfft precongruence} shows that $\sqsubseteq^{r*}_{FT}$ is a precongruence for the operators of CCSP$_\rt$.
By \df{rfft}, $\sqsubseteq^{r*}_{FT}$ is included in $\sqsubseteq^{*}_{FT}$,
i.e., $P \sqsubseteq^{r*}_{FT} Q$ implies $P\sqsubseteq^{*}_{FT} Q$.
Thus it suffices to find, for any two processes $P$ and $Q$ with $P \not\sqsubseteq^{r*}_{FT} Q$, a CCSP$_\rt$-context
$\mathcal{C}$ such that $\mathcal{C}(P) \not\sqsubseteq^{*}_{FT} \mathcal{C}(Q)$.
This proceeds exactly as in the proof of \thrm{full abstraction +}. Of course the symmetric case ``$\sigma\in \rfft(P) \setminus \rfft(Q)$''
is skipped, as it is not needed here.
\end{proof}

\section{The coarsest precongruence respecting safety properties}\label{coarsest}

In \cite{vG10} I proposed a way to define refinement preorders $\sqsubseteq$ on processes,
where $P \sqsubseteq Q$ says that for all practical purposes under consideration, $Q$ is at least as
suitable as $P$, i.e., it will never harm to replace $P$ by $Q$. In the stepwise design of systems,
$P$ may be closer to a specification, and $Q$ to an implementation. Each such a refinement preorder also
yields a semantic equivalence $\equiv$, with $P\equiv Q$ saying that for practical purposes $P$ and
$Q$ are equally suitable; i.e., one can be replaced by the other without untoward side effects.
Naturally, $P \equiv Q$ iff both $P \sqsubseteq Q$ and $Q \sqsubseteq P$.

The method of \cite{vG10} equips the choice of $\sqsubseteq$ with two parameters:
a class $\fG$ of good properties of processes, ones that may be required of processes in a given
range of applications, and a class $\fO$ of useful operators for combining processes.
The first requirement on $\sqsubseteq$ is
\mbox{$P \sqsubseteq Q \Rightarrow \forall \varphi\in \fG.\; (P \models \varphi \Rightarrow Q \models \varphi)$}
where $P\models\varphi$ denotes that process $P$ has the property $\varphi$.
If this holds, $\sqsubseteq$ \emph{respects} or \emph{preserves} the properties in $\fG$. 
The second requirement is that $\sqsubseteq$ is a \emph{precongruence} for $\fO$:
$P \sqsubseteq Q \Rightarrow \mathcal{C}[P] \sqsubseteq \mathcal{C}[Q]$ for any context $\mathcal{C}[\_\!\_\,]$ built from
operators from $\fO$. Given the choice of $\fG$ and $\fO$, the preorder recommended by \cite{vG10}
is the \emph{coarsest}, or largest, one satisfying both requirements.
This preorder is called \emph{fully abstract} w.r.t.\ $\fG$ and $\fO$, and
is characterised by $$P \sqsubseteq Q ~~\Leftrightarrow~~ \forall \fO\mbox{-context } \mathcal{C}[\_\!\_].\;
\forall\varphi\in \fG.\;  (\mathcal{C}[P] \models \varphi \Rightarrow \mathcal{C}[Q] \models \varphi).$$
The corresponding semantic equivalence identifies processes only when this is enforced by the two
requirements above.

Naturally, increasing the class $\fO$ of operators makes the resulting fully abstract preorder
\emph{finer}, or smaller. The same holds when increasing the class $\fG$ of properties.
In \cite{vG10} I employed a class $\fO$ of operators that was effectively equivalent to the
operators of CCSP$^g$. Here I use the operators of CCSP$^g_\rt$, thus adding time-outs.
In this section, following \cite[Section 3]{vG10}, I take as class $\fG$ of good properties the
\emph{safety properties}, saying that something bad will never happen.
What exactly counts as a safety property could be open to some debate.
However, the \emph{canonical safety property}, proposed in \cite{vG10}, is definitely in this class.
{\We} first characterise the preorder $\sqsubseteq_{\it safety}$ that is fully abstract w.r.t.\ the
operators of CCSP$^g_\rt$ and this single safety property. It turns out to be $\sqsubseteq^{*}_{FT}$.
Then {\we} argue that adding any other safety properties to the class $\fG$ could not possible
make the resulting preorder any finer.
An immediately corollary of this and \thrm{full abstraction preorder} is that $\sqsubseteq^{r*}_{FT}$
is fully abstract for the operators of CCSP$_\rt$ and the class of safety properties.

\subsection*{Full abstraction w.r.t.\ the canonical safety property}

To formulate the canonical safety property, assume that the alphabet $A$
of visible actions contains one specific action $b$, whose occurrence
is \emph{bad}. The property now says that \textbf{$b$ will never happen}.

\begin{defi}\label{df:canonical safety}\rm
A process $P$ satisfies the \emph{canonical safety property}, notation
$P\models \textit{safety}(b)$, if no partial failure trace of $P$ contains the action $b$.
\end{defi}

\begin{lemma}{initial}
If $a\eta\in\fft(P)$ then $P{\goesto{a}}$ or $P{\goesto{\tau}}$.
\end{lemma}
\begin{proof}
That $a\eta\in\fft(P)$ can be derived only from the second or third rule of \tab{FT CCSP}.
\end{proof}

\begin{corollary}{initial}
If $X a\eta\in\fft(P)$ is obtained from the fourth rule of \tab{FT CCSP}, then $a\notin X$.
\end{corollary}
\begin{proof}
  The fourth rule requires that $P{\ngoesto\alpha}$ for all $\alpha\in X\cup\{\tau\}$, and
  that $a\eta\in\fft(P)$. The latter implies $P{\goesto{a}}$ or $P{\goesto{\tau}}$ by \lemref{initial}.
  So $P{\goesto{a}}$ and $a\notin X$.
\end{proof}
\noindent
As encountered already in the proof of \thrm{full abstraction +}, {\we} call an action $b$ \emph{fresh}
for a process $P$ if for each process $R$ reachable from $P$ one has $R{\ngoesto{b}}$.
In that case $P\models \textit{safety}(b)$.
{\We} call it \emph{fresh} for a partial failure trace $\sigma \in (A \cup \Pow(A))^*\top$ if $b$ does not occur in
$\sigma$, either as action or inside a refusal set\ in $\sigma$.

\begin{theorem}{may}
For each $\sigma \in (A \cup \Pow(A))^*\top$ for which $b$ is fresh, there exists a CCSP$^g_\rt$ process $T_\sigma$
such that\vspace{-2ex}
$$\tau_{A\setminus\{b\}}(T_\sigma \spar{A\setminus\{b\}} P)  \not\models\textit{safety}(b) ~~\Leftrightarrow~~ \sigma \in \fft(P)$$
for each  CCSP$^g_\rt$ process $P$ for which $b$ is fresh.
\end{theorem}

\begin{proof}
Let $B := A{\setminus}\{b\}$, the set of all visible actions except $b$.
Define $T_\sigma$ with structural induction on $\sigma$:\\[1ex]
\mbox{}\hfill
$T_\top := \rt.b$
\hfill
$T_{c\rho} := \tau + c.T_\rho$
\hfill
$T_{X\eta} := \rt.T_\eta +\sum_{a\in X} a$
\hfill
$T_{X d\rho} := \rt.(d.T_\rho +\sum_{a\in X\setminus\{d\}} a)+\sum_{a\in X} a$
\hfill\mbox{}\\[1ex]
where, $c,d\in B$, $X\subseteq B$, $d\in X$, and $\eta$ is not of the form $d\rho$ with $d\mathbin\in X$.
Note that $\tau_{B}(T_\sigma \spar{B} P)  \not\models\textit{safety}(b)$ iff there is a path
$\pi := \tau_{B}(T_\rho \spar{B} P) \goesto{\alpha_1} Q_1 \goesto{\alpha_2} \cdots \goesto{\alpha_n} Q_\ell \goesto{b} Q'$,
and this path gives rise to a partial failure trace $\beta b \top \in \fft(\tau_{B}(T_\sigma \spar{B} P))$
by application of the rules of \tab{FT CCSP}. Due to the operator $\tau_B$, each of the $\alpha_i$
must be $\tau$ or $\rt$.

To obtain ``$\Rightarrow$'' {\we} assume $\tau_{B}(T_\sigma \spar{B} P) \not\models\textit{safety}(b)$
and apply structural induction on $\sigma$.

\begin{itemize}
\item Let $\sigma=\top$. Then $\sigma \in \fft(P)$ for all $P$.
\item Let $\sigma=c\rho$.
  For the process $\tau_{B}(T_\sigma \spar{B} P)$ to ever reach a $b$-transition, the component
  $T_\sigma = \tau + c.T_\rho$ must take the $c$-transition to $T_\rho$. 
  So the beginning of $\pi$ must have the form
  $$\tau_{B}(T_\sigma \spar{B} P) \goesto{\alpha_1}
   \tau_{B}(T_\sigma \spar{B} P_1) \goesto{\alpha_2}
   \cdots                       \goesto{\alpha_n}
   \tau_{B}(T_\sigma \spar{B} P_n) \goesto{\tau}
   \tau_{B}(T_\rho \spar{B} P')$$
  for some $n\geq 0$, where $\tau_{B}(T_\rho \spar{B} P')  \not\models\textit{safety}(b)$. 
  Let $P_0:=P$.
  All of these $\alpha_i$ must be $\tau$, for $\rt$-transitions lose out from the outgoing
  $\tau$-transition from $\tau_{B}(T_\sigma \spar{B} P_{i-1})$. Technically, if $\alpha_{i}=\rt$,
  for $1\leq i\leq n$,
  then the fifth and sixth rules of \tab{FT CCSP} do not allow any partial failure trace of 
  $\tau_{B}(T_\rho \spar{B} P_{i})$ to give rise to a partial failure trace of $\tau_{B}(T_\sigma \spar{B} P_{i-1})$.
  The above path must stem from a path
  $$\qquad T_\sigma \spar{B} P \goesto{\beta_1}
   T_\sigma \spar{B} P_1 \goesto{\beta_2}
   \cdots               \goesto{\beta_n}
   T_\sigma \spar{B} P_n \goesto{c}
   T_\rho \spar{B} P'$$
  that in turns stems from a path 
  $P\goesto{\beta_1} P_1\goesto{\beta_2} \cdots \goesto{\beta_n} P_n \goesto{c} P'$.
  All of these $\beta_i$ must be $\tau$, as the synchronisation on $B$ stops all visible actions
  stemming from $P$.

  By induction $\rho \in \fft(P')$.
  The second and third rules of \tab{FT CCSP} yield $\sigma \in \fft(P)$.
\item Let $\sigma=X\eta$, where $X\subseteq B$ and $\eta$ is not of the form $d\rho$ with $d\mathbin\in X$.
  For the process $\tau_{B}(T_\sigma \spar{B} P)$ to ever reach a $b$-transition, the component
  $T_\sigma = \rt.T_\eta +\sum_{a\in X} a$ must take the $\rt$-transition to $T_\eta$. 
  So the beginning of $\pi$ must have the form
  $$\tau_{B}(T_\sigma \spar{B} P) \goesto{\alpha_1}
   \tau_{B}(T_\sigma \spar{B} P_1) \goesto{\alpha_2}
   \cdots                       \goesto{\alpha_n}
   \tau_{B}(T_\sigma \spar{B} P_n) \goesto{\rt}
   \tau_{B}(T_\eta \spar{B} P_n)$$
  for some $n\geq 0$, where $\tau_{B}(T_\eta \spar{B} P_n)  \not\models\textit{safety}(b)$. 
  By induction $\eta \in \fft(P_n)$.
  As above, again using that visible actions from $P$ are blocked by $\spar{B}$,
  it must be that $P\goesto{\alpha_1} P_1\goesto{\alpha_2} \cdots \goesto{\alpha_n} P_n$.\linebreak[3]
  I proceed with induction on $n$.

  \emph{Base case:} It must be that $P=P_n {\ngoesto\alpha}$ for all $\alpha\in X\cup\{\tau\}$, for
  otherwise $\tau_{B}(T_\sigma \spar{B} P_n)$ would have an outgoing $\tau$-transition that wins
  from the $\rt$-transition to $\tau_{B}(T_\eta \spar{B} P_n)$. The fourth rule of \tab{FT CCSP} yields $\sigma \in \fft(P)$.

  \emph{Induction step:} By induction  $\sigma \in \fft(P_1)$. In case $\alpha_1=\tau$ one obtains
  $\sigma \in \fft(P)$ by the third rule of \tab{FT CCSP}. In case $\alpha_1=\rt$, it must be that
  $P {\ngoesto\alpha}$ for all $\alpha\in X\cup\{\tau\}$, for otherwise $\tau_{B}(T_\sigma \spar{B} P)$
  would have an outgoing $\tau$-transition that wins from the $\rt$-transition to $\tau_{B}(T_\sigma \spar{B} P_1)$. 
  The fifth rule of \tab{FT CCSP} yields $\sigma \in \fft(P)$.

\item Let $\sigma\mathbin=X d\rho$ with $d\mathbin\in X$. Writing $R$ for $d.T_\rho +\sum_{a\in X\setminus\{d\}} a$,
  the first part of $\pi$ must have the form\vspace{-3pt}
  $$\tau_{B}(T_\sigma \spar{B} P) \goesto{\alpha_1}
   \tau_{B}(T_\sigma \spar{B} P_1) \goesto{\alpha_2}
   \cdots                       \goesto{\alpha_n}
   \tau_{B}(T_\sigma \spar{B} P_n) \goesto{\rt}
   \tau_{B}(R \spar{B} P_n)$$
\vspace{-3ex}
  $$\tau_{B}(R \spar{B} P_n) \goesto{\alpha_{n+1}}
   \tau_{B}(R \spar{B} P_{n+1}) \goesto{\alpha_{n+2}}
   \cdots                       \goesto{\alpha_m}
   \tau_{B}(R \spar{B} P_m) \goesto{\tau}
   \tau_{B}(T_\rho \spar{B} P')$$
  for some $m \geq n\geq 0$, where $\tau_{B}(T_\rho \spar{B} P')  \not\models\textit{safety}(b)$. 
  By induction $\rho \in \fft(P')$.
  As above, it must be that 
  $P\goesto{\alpha_1} P_1\goesto{\alpha_2} \cdots \goesto{\alpha_n} P_n \goesto{\alpha_{n+1}} P_{n+1}\goesto{\alpha_{n+2}} \cdots \goesto{\alpha_m} P_m \goesto{d} P'$.
  Necessarily, $P_n {\ngoesto\alpha}$ for all $\alpha\in X\cup\{\tau\}$, for otherwise $\tau_{B}(T_\sigma \spar{B} P_n)$
  would have an outgoing $\tau$-transition that wins from the $\rt$-transition to $\tau_{B}(R \spar{B} P_n)$. 
  In particular, $P_n{\ngoesto d}$, whereas $P_m{\goesto d}$, so $n\neq m$.
  Moreover, $\alpha_{n+1}=\rt$, for $\alpha_{n+1}=\tau$ contradicts with $P_n{\ngoesto \tau}$.
  Let $k\leq m$ be the highest index such that $\alpha_{k}=\rt$.
  Then $d\rho\in\fft(P_k)$, by the second and third rules of \tab{FT CCSP}.
  For each index $i$ with $\alpha_i=\rt$ one must have  $P_{i-1} {\ngoesto\alpha}$ for all $\alpha\in X\cup\{\tau\}$,
  for the same reason as above. The sixth rule of \tab{FT CCSP} yields $X d\rho\in\fft(P_{k-1})$.
  Successive applications of the third and fifth rules yield  $X d\rho\in\fft(P_{j})$ for all $j<k$,
  so in particular $\sigma\in\fft(P)$.
\end{itemize}

To obtain ``$\Leftarrow$'', {\we} will show that 
$\sigma \mathbin\in \fft(P)$ implies $A b \top \in \fft(\tau_{B}(T_\sigma \spar{B} P))$.
Here $A$ is the set of all visible actions.
In doing so, {\we} apply induction on the derivation of $\sigma \mathbin\in \fft(P)$.

Suppose $\sigma \mathbin= \top \mathbin\in\fft(P)$ is derived from the first rule of \tab{FT CCSP}.
Since $T_\sigma = \rt.b.0$ one has
$\tau_{B}(T_\sigma \spar{B} P) \goesto \rt \tau_{B}(b.0 \spar{B} P) \goesto b \tau_{B}(0 \spar{B} P)$.
Hence $A b\top \in \fft(\tau_{B}(T_\sigma \spar{B} P))$, using the first, second and sixth rules of \tab{FT CCSP},
and using that $\tau_{B}(T_\sigma \spar{B} P){\ngoesto\alpha}$ for all $\alpha\in A\cup\{\tau\}$.

Suppose $\sigma = c\rho\in\fft(P)$ is derived from the second rule of \tab{FT CCSP}.
Then $P\goesto c P'$ and $\rho\in\fft(P')$.
By induction, $A b \top \in \fft(\tau_{B}(T_\rho \spar{B} P'))$.
Since $\tau_{B}(T_\sigma \spar{B} P) \goesto\tau \tau_{B}(T_\rho \spar{B} P')$,
it follows with the third rule of \tab{FT CCSP}
that $A b \top \in \fft(\tau_{B}(T_\sigma \spar{B} P))$.

Suppose $\sigma\in\fft(P)$ is derived from the third rule of \tab{FT CCSP}.
Then $P\goesto\tau P'$ and $\sigma\in\fft(P')$.
By induction, $A b \top \in \fft(\tau_{B}(T_\sigma \spar{B} P'))$.
Since $\tau_{B}(T_\sigma \spar{B} P) \goesto\tau \tau_{B}(T_\sigma \spar{B} P')$,
it follows with the third rule of \tab{FT CCSP} that $A b \top \in \fft(\tau_{B}(T_\sigma \spar{B} P))$.

Suppose $\sigma\mathbin=X\eta\mathbin\in\fft(P)$ is derived from the fourth rule of \tab{FT CCSP}.
Then $P{\ngoesto{\alpha}}$ for all $\alpha\mathbin\in X\cup\{\tau\}$ and $\eta\in\fft(P)$.
In the special case that $\eta=d\rho$, by \corref{initial} $d\notin X$.
By induction $A b\top\mathbin\in \fft(\tau_{B}(T_\eta \spar{B} P))$.
Since $\tau_{B}(T_\sigma \spar{B} P) \goesto\rt \tau_{B}(T_\eta \spar{B} P)$,
and $\tau_{B}(T_\sigma \spar{B} P){\ngoesto\alpha}$ for all $\alpha\in A\cup\{\tau\}$,
the fifth rule of \tab{FT CCSP} yields $A b \top \in \fft(\tau_{B}(T_\sigma \spar{B} P))$.

Suppose $\sigma\mathbin=X\eta\mathbin\in\fft(P)$ is derived from the fifth rule of \tab{FT CCSP}.
Then $P{\ngoesto{\alpha}}$ for all $\alpha\mathbin\in X\cup\{\tau\}$,
$P \goesto\rt P'$ and $\sigma\in\fft(P')$.
So $\tau_{B}(T_\sigma \spar{B} P){\ngoesto\alpha}$ for all $\alpha\in A\cup\{\tau\}$,
and $A b\top\mathbin\in \fft(\tau_{B}(T_\sigma \spar{B} P'))$ by induction.
Since $\tau_{B}(T_\sigma \spar{B} P) \goesto\rt \tau_{B}(T_\sigma \spar{B} P')$,
the fifth rule yields $A b \top \in \fft(\tau_{B}(T_\sigma \spar{B} P))$.

Suppose $\sigma\mathbin=X d\rho\mathbin\in\fft(P)$ is derived from the sixth rule of \tab{FT CCSP}.
Then $P{\ngoesto{\alpha}}$ for all $\alpha\mathbin\in X\cup\{\tau\}$, $d\mathbin\in X$,
$P \goesto\rt P'$ and $d\rho\in\fft(P')$. So $\tau_{B}(T_\sigma \spar{B} P){\ngoesto\alpha}$
and $\tau_{B}(R \spar{B} P){\ngoesto\alpha}$ for all $\alpha\in A\cup\{\tau\}$,
where $R := d.T_\rho +\sum_{a\in X\setminus\{d\}} a$.\vspace{-1pt}
By induction $A b\top\mathbin\in \fft(\tau_{B}(T_{d\rho} \spar{B} P'))$.
As $\tau_{B}(T_\sigma \spar{B} P) \goesto\rt \tau_{B}(R \spar{B} P)\linebreak[2] \goesto\rt \tau_{B}(R \spar{B} P')$,
two applications of the fifth rule of \tab{FT CCSP} yield $A b \top \in \fft(\tau_{B}(T_\sigma \spar{B} P))$.
\end{proof}

\begin{corollary}{safety}
$\sqsubseteq^{*}_{FT}$ is fully abstract w.r.t.\ the operators of CCSP$^g_\rt$ and the canonical safety property.
\end{corollary}

\begin{proof}
By \thrm{fft precongruence}, $\sqsubseteq^{*}_{FT}$ is a precongruence for the operators of CCSP$^g_\rt$.
By definition, $P \sqsubseteq^{*}_{FT} Q$ implies $P \models \textit{safety}(b) \Rightarrow Q \models \textit{safety}(b)$.
It remains to show that $\sqsubseteq^{*}_{FT}$ is the coarsest preorder with these properties.
To this end it suffices to find for any processes $P$ and $Q$ with $P \not\sqsubseteq^{*}_{FT} Q$
a recursion-free CCSP$^g_\rt$-context $\mathcal{C}[\_\!\_\,]$ such that $\mathcal{C}[P] \models\textit{safety}(b)$,
yet $\mathcal{C}[Q] \not\models\textit{safety}(b)$.

So assume $P \not\sqsubseteq^{*}_{FT} Q$. 
By applying a bijective renaming operator on $P$ and $Q$, mapping $A$ to $A\setminus\{b\}$,
{\we} may assume that $b$ is fresh for $P$ and $Q$.
Take $\sigma\in\fft(Q)\setminus\fft(P)$. By Property (*) in the proof of \thrm{full abstraction +}
I may assume that $b$ is fresh for $\sigma$. Now \thrm{may} yields
$\tau_{A\setminus\{b\}}(T_\sigma \spar{A\setminus\{b\}} P) \models\textit{safety}(b)$, yet
\plat{$\tau_{A\setminus\{b\}}(T_\sigma \spar{A\setminus\{b\}} Q) \not\models\textit{safety}(b)$}.
\end{proof}

\subsection*{Full abstraction w.r.t.\ general safety properties}

The above says that the unique preorder $\sqsubseteq_{\it safety}$ that is fully abstract w.r.t.\ the
operators of CCSP$^g_\rt$ and the canonical safety property coincides with $\sqsubseteq^{*}_{FT}$.
Without the time-out operator, $\sqsubseteq_{\it safety}$ would be way less discriminating, and coincides
with reverse weak partial trace inclusion \cite{vG10}.
{\We} now check whether one could get an even finer preorder than $\sqsubseteq^{*}_{FT}$ by also considering other safety properties.

To arrive at a general concept of safety property for LTSs, the paper \cite{vG10} assumes that some
notion of \emph{bad} is defined.  This induces the safety property saying that this bad thing will
never happen.  To judge whether a process $P$ satisfies this safety property, one should judge
whether $P$ can reach a state in which one would say that this bad thing had happened. But all
observable behaviour of $P$ that is recorded in an LTS until one comes to such a verdict, is the
sequence of visible actions performed until that point. Thus the safety property is completely
determined by the set sequences of visible actions that, when performed by $P$, lead to such a
judgement. Therefore, \cite{vG10} proposes to define the concept of a safety property in terms of such a
set: a safety property is given by a set $B\subseteq A^*$ of \emph{bad} partial traces.
A process $P$ \emph{satisfies} this property, notation $P\models\textit{safety}(B)$, if $P$ has no
partial trace in $B$. It follows immediately from this definition that partial trace equivalent
processes satisfy the same safety properties of this kind.

The present paper allows for a stronger concept of safety property. Here the observable behaviour of
a process that is recorded until one comes to the verdict that something bad has happened can be modelled
as a partial failure trace. Based on this, my definition is analogous to the one in \cite{vG10}:
\begin{definition}{safety}
A \emph{safety property} of CCSP$_\rt$ processes is given by a set $B\subseteq (A \cup \Pow(A))^*\top$
of \emph{bad} partial failure traces. A process $P$ \emph{satisfies} this property, notation $P\models\textit{safety}(B)$, 
if $\fft(P)\cap B = \emptyset$.
\end{definition}
\noindent
It follows immediately that $\sqsubseteq^{*}_{FT}$ respects all safety properties.
Consequently, as in \cite{vG10}, for the definition of $\sqsubseteq_{\it safety}$ it does not matter
whether all safety properties are considered, or only the canonical one.

\subsection*{Trace equivalence and its congruence closure}

A possible definition of a (weak) partial trace is simply a partial failure trace from which
all refusal sets have been omitted. Define two processes to be (weak) (partial) trace equivalent iff
they have the same partial traces. Now \corref{safety} implies that $\ffteq$ is the congruence closure
of trace equivalence, that is, the coarsest congruence for the operators of CCSP$_\rt$ that is finer
than trace equivalence.\vspace{-1ex}

\subsection*{May testing}

\emph{May testing} was proposed by De Nicola \& Hennessy in \cite{DH84} for the process algebra CCS\@.
Translated to CCSP$_\rt$ it works as follows. A \emph{test} is a CCSP$_\rt$ process that instead of
visible actions from the alphabet $A$, uses visible actions from the alphabet $A\uplus\{\omega\}$,
where $\omega$ is a fresh action reporting ``success''. A process $P$ \emph{may pass} a test $T$
iff $\tau_A(T\spar{A}P)$ has some partial failure trace $\beta\omega\top$, or, equivalently,
has a partial trace containing the action $\omega$. A process with such a trace represents a system
with an execution that eventually reports success. Now define the preorder
$\sqsubseteq_{\it may}$ on  CCSP$_\rt$ processes by
\begin{center}
$P \sqsubseteq_{\it may}Q$ iff each test $T$ that $P$ may pass, may also be passed by $Q$.
\end{center}
May testing contrasts with \emph{must testing}, which requires
that \emph{each} execution leads to a state that can report success.

\begin{theorem}{may2}
$P \sqsubseteq_{\it may}Q$ iff $Q \sqsubseteq^{*}_{FT} P$.
\end{theorem}
\begin{proof}
Note that $R$ has some partial failure trace $\beta\omega\top$ iff $R \not\models\textit{safety}(\omega)$.
Using this, \thrm{may2} is an immediately consequence of \thrm{may}, with $\omega$ in the r\^ole of $b$.
The side conditions of $b$ being fresh in \thrm{may} are redundant in this context, as $\omega$
is fresh by construction.
\end{proof}
\noindent
So the may-preorder is the converse of the safety preorder. This is related to calling $\omega$ a
success action, instead of $b$ a bad action. The may-preorder preserves the property that a process
\emph{can} do something good, whereas the safety preorder preserves the property that a process
\emph{can not} so something bad.

\section{Conclusion}

This paper extended the model of labelled transition systems with time-out transitions $s_1 \goesto\rt s_2$.
Such a transition is assumed to be instantaneous and unobservable by the environment, just like a
transition $s_1 \goesto\tau s_2$, but becomes available only a positive but finite amount of time
after the represented system reaches state $s_1$. This extension allows the implementation of a
simple priority mechanism, thereby increasing the expressiveness of the model.

To denote such labelled transition systems I extended the standard process algebra CCSP with the action
prefixing operator $\rt.\_\!\_\,$. Many semantic preorders $\sqsubseteq$ have been defined on 
labelled transition systems, and thereby on the expressions in CCSP and other process algebras.
$P \sqsubseteq Q$ says that process $Q$ is a \emph{refinement} of process $P$, where $P$ is closer to the
specification of a system, and $Q$ to its implementation.
The least one may wish to require from such a preorder is
\begin{enumerate}[(1)]
\item that it is a precongruence for the \emph{static} operators of CCSP: partially synchronous parallel
  composition, renaming, and abstraction by renaming visible actions into the hidden action
  $\tau$,
\item and that it respects (basic) safety properties, meaning that any safety property of a
  process $P$ also holds for its refinements $Q$.
\end{enumerate}
The coarsest preorder with these properties is well known to be the reverse weak partial trace inclusion
(see \eg\ \cite{vG10}). Accordingly, weak partial trace inclusion is the preorder generated by
\emph{may testing} \cite{DH84}.

After the extension with time-out transitions, the weak partial trace preorder is no longer a
precongruence for the static operators of CCSP\@. I here characterised the coarsest preorder satisfying
(1) and (2) above as the \emph{partial failure trace} preorder.
Again, its converse is the preorder generated by may testing.
All the above also applies to semantic equivalences, arising as the kernels of these preorders.

To obtain a (pre)congruence for all CCSP operators, including the CCS choice $+$, one needs to use
a \emph{rooted} version of the partial failure trace preorder, as is common in the study of preorders
that abstract from the hidden action $\tau$.

The work reported here remains in the realm of untimed process algebra, in the sense that the
progress of time is not quantified. In the study of Timed CSP \cite{RR87,DS95,Schn95,RR99},
similar work has been done in a setting where the progress of time is quantified.
Also there a form of failure trace semantics was found to be the right equivalence, and the
connection with may testing was made in \cite{Schn95}. Since the bookkeeping of time in 
\cite{RR87,DS95,Schn95,RR99} is strongly interwoven with the formalisation of failure traces, it is not
easy to determine whether my semantics can in some way be seen as a abstraction of the one for Timed
CCSP, for instance by instantiating each quantified passage of time with a nondeterministic choice allowing
\emph{any} amount of time. This appears to be a question worthy of further research.
A contribution of the present work is that the transformation of a failures based semantics of CSP
\cite{BHR84,Ho85} to a failure trace based semantics of Timed CSP \cite{RR87,DS95,Schn95,RR99}
is not really necessitated by quantification of time---often seen as the main difference
between CSP and Timed CSP---but rather by the introduction of a time-out operator (quantified or not).

Future work includes proving a congruence result for recursion, finding complete axiomatisations,
and extending the approach from partial to complete failure traces, so that liveness properties will
be respected. In that setting the expressiveness questions of Section~\ref{expressiveness} can be studied.
The adaption of strong bisimilarity to the setting with time-out transitions is studied in \cite{vG20b}.

\bibliographystyle{eptcsalpha}
\bibliography{XYZ13}

\appendix

\section{Proofs of congruence properties}

This appendix contains the proofs of the explicit characterisations of the operators $\spar{S}$,
$\tau_I$ and $\Rn$ on (rooted) partial failure trace sets: Properties (\ref{spar explicit}),
(\ref{abstraction congruence}), (\ref{renaming congruence}), (\ref{rpar}), (\ref{rabs}) and (\ref{rren}).
These are the central parts of the proofs showing that $\ffteq$ and $\rffteq$ are congruences for these operators.

\begin{proposition}{spar explicit}
$\fft(P \spar{S} Q) = \col(\fft(P) \spar{S} \fft(Q))$.
\end{proposition}

\begin{proof}
``$\supseteq$'': Let $\sigma\in\fft(P) \spar{S} \fft(Q)$, and fix a valid decomposition of
$\sigma$ such that $\sigma_{\rm L} \in \fft(P)$ and $\sigma_{\rm R} \in \fft(Q)$.
With \obsref{doubling} it suffices to show that $\sigma\in\fft(P\spar{S}Q)$.
{\We} proceed with structural induction on the derivations of
$\sigma_{\rm L} \in \fft(P)$ and $\sigma_{\rm R} \in \fft(Q)$ from the rules in \tab{FT CCSP},
making a case distinction on the first symbol of $\sigma$.
The case $\sigma=\top$ is trivial.
\begin{itemize}
\item
Let $\sigma= a\rho$ with $a \in S$.
Then $\sigma_{\rm L} = a\rho_{\rm L}$ and $\sigma_{\rm R} = a\rho_{\rm R}$.

Assume $\sigma_{\rm L} \in \fft(P)$ is obtained from the third rule of \tab{FT CCSP}.
Then $P \goesto{\tau} P'$ and $\sigma_{\rm L} \in \fft(P')$.
So $P\spar{S}Q \goesto\tau P'\spar{S}Q$ by \tab{sos CCSP}, and $\sigma \in \fft(P'\spar{S}Q)$ by induction.
It follows that $\sigma\in\fft(P\spar{S}Q)$.

The case that $\sigma_{\rm R} \in \fft(Q)$ is obtained from the third rule of \tab{FT CCSP} proceeds likewise.

So assume $\sigma_{\rm L} \in \fft(P)$ and $\sigma_{\rm R} \in \fft(Q)$ are both obtained from the second rule of \tab{FT CCSP}.
Then $P \goesto{a} P'$,  $Q \goesto{a} Q'$, $\rho_{\rm L} \in \fft(P')$ and  $\rho_{\rm R} \in \fft(Q')$.
So $P\spar{S}Q \goesto{a} P'\spar{S}Q'$ by \tab{sos CCSP}, and $\rho \in \fft(P'\spar{S}Q')$ by induction.
It follows that $\sigma\in\fft(P\spar{S}Q)$.

\item
Let $\sigma= a\rho$ with $a \notin S$.
Assume that $\sigma_{\rm L} = a\rho_{\rm L}$ and $\sigma_{\rm R} = \rho_{\rm R}$---the other case, that
$\sigma_{\rm L} = \rho_{\rm L}$ and $\sigma_{\rm R} = a\rho_{\rm R}$, will follow by symmetry.

Assume $\sigma_{\rm L} \in \fft(P)$ is obtained from the third rule of \tab{FT CCSP}.
Then $P \goesto{\tau} P'$ and $\sigma_{\rm L} \in \fft(P')$.
So $P\spar{S}Q \goesto\tau P'\spar{S}Q$ by \tab{sos CCSP}, and $\sigma \in \fft(P'\spar{S}Q)$ by induction.
It follows that $\sigma\in\fft(P\spar{S}Q)$.

Now assume $\sigma_{\rm L}  \mathbin\in \fft(P)$ is obtained from the second rule of \tab{FT CCSP}.
Then $P \goesto{a} P'$ and $\rho_{\rm L} \in \fft(P')$.
So $P\spar{S}Q \goesto{a} P'\spar{S}Q$ by \tab{sos CCSP}, and $\rho \in \fft(P'\spar{S}Q)$ by induction.
It follows that $\sigma\in\fft(P\spar{S}Q)$.

\item
Let $\sigma= X\rho$.
Then $\sigma_{\rm L} = X_{\rm L}\rho_{\rm L}$ and $\sigma_{\rm R} = X_{\rm R}\rho_{\rm R}$.
Both these statements must be derived from the fourth to sixth rule of \tab{FT CCSP}.
Hence $P {\ngoesto{\alpha}}$ for all $\alpha \in X_{\rm L}\cup\{\tau\}$,
and $Q {\ngoesto{\alpha}}$ for all $\alpha \in X_{\rm R}\cup\{\tau\}$.
Note that  (\ref{consistency})--(\ref{decomposition}) hold with $X$, $X_{\rm L}$ and $X_{\rm R}$ in the
r\^oles of $\sigma_i$, $\sigma_i^{\rm L}$ and $\sigma_i^{\rm R}$.
From (\ref{decomposition}) it follows that $P\spar{S}Q {\ngoesto{\alpha}}$ for all $\alpha \in X\cup\{\tau\}$.

Assume $\sigma_{\rm L} \in \fft(P)$ is obtained from the fifth rule of \tab{FT CCSP}.
Then $P \goesto{\rt} P'$ and $\sigma_{\rm L} \in \fft(P')$.
So $P\spar{S}Q \goesto{\rt} P'\spar{S}Q$ by \tab{sos CCSP}, and $\sigma \in \fft(P'\spar{S}Q)$ by induction.
It follows that $\sigma\in\fft(P\spar{S}Q)$.

The case that $\sigma_{\rm R} \in \fft(Q)$ is obtained from the fifth rule of \tab{FT CCSP} proceeds likewise.

Assume $\sigma_{\rm L} \mathbin\in \fft(P)$ and $\sigma_{\rm R} \in \fft(Q)$ are both obtained from the fourth rule of \tab{FT CCSP}.
Then $\rho_{\rm L} \mathbin\in \fft(P)$ and $\rho_{\rm R} \mathbin\in \fft(Q)$.
So $\rho \in \fft(P\spar{S}Q)$ by induction.
Hence $\sigma\in\fft(P\spar{S}Q)$.

Assume $\sigma_{\rm L} \mathbin\in \fft(P)$ is obtained from the sixth rule of \tab{FT CCSP} and
$\sigma_{\rm R} \mathbin\in \fft(Q)$ from the fourth.
Then $P \goesto{\rt} P'$, $\rho_{\rm L} \in \fft(P')$ and $\rho_{\rm L}$ has the form $a \eta_{\rm L}$
with $a \mathbin\in X_{\rm L}\subseteq X$---the latter using (\ref{decomposition}). Moreover, $\rho_{\rm R} \in \fft(Q)$.
Now $P\spar{S}Q \goesto{\rt} P'\spar{S}Q$ by \tab{sos CCSP}, and $\rho \in \fft(P'\spar{S}Q)$ by induction.
Since $X_{\rm L}$ is system-ended in $\sigma_{\rm L}$, by (\ref{consistency}) $X$ must be system-ended in $\sigma$.
Hence $\rho$ has the form $b\eta$ with $b\in X$.
Using the sixth rule  of \tab{FT CCSP} it follows that $\sigma\in\fft(P\spar{S}Q)$.

The case that $\sigma_{\rm L} \mathbin\in \fft(P)$ is obtained from the fourth rule of \tab{FT CCSP} and
$\sigma_{\rm R} \mathbin\in \fft(Q)$ from the sixth proceeds likewise.

In case $\sigma_{\rm L} \mathbin\in \fft(P)$ and $\sigma_{\rm R} \in \fft(Q)$
are both obtained from the sixth rule of \tab{FT CCSP},
$X_{\rm L}$ is system-ended in $\sigma_{\rm L}$ and $X_{\rm R}$ is system-ended in $\sigma_{\rm R}$.
So this case is excluded by condition (\ref{consistency2}).
\end{itemize}
``$\subseteq$'': Let $\ffte(x)$ be defined exactly as $\fft(x)$---see \tab{FT CCSP}---except that
the conclusion of the fifth rule reads ``$XX\rho \in \ffte(x)$''.
Now $\fft(P \spar{S} Q) \subseteq \col(\ffte(P \spar{S} Q))$, so, using that $\col$ is monotonous, it
suffices to show that $\ffte(P \spar{S} Q) \subseteq \fft(P) \spar{S} \fft(Q)$.

So let $\sigma \in \ffte(P \spar{S} Q)$. With structural induction on the derivation of
$\sigma \in \ffte(P \spar{S} Q)$ from the (amended) rules of Tables~\ref{tab:sos CCSP} and~\ref{tab:FT CCSP}
I show that $\sigma \in \fft(P) \spar{S} \fft(Q)$. This means I have to give a valid decomposition
of $\sigma$ such that $\sigma_{\rm L}\in \fft(P)$ and $\sigma_{\rm R} \in  \fft(Q)$.
\begin{itemize}
\item
  Let $\sigma=\top$, with $\sigma\in\ffte(P\spar{S}Q)$ derived from the first rule of \tab{FT CCSP}.
  Then $\sigma$ has only one decomposition, which is valid.
  Trivially, $\sigma_{\rm L}= \top \in \fft(P)$ and $\sigma_{\rm R} = \top \in  \fft(Q)$.
\item
  Let $\sigma\mathbin=a\rho$,  with $\sigma\in\ffte(P\spar{S}Q)$ derived from the second rule.
  Then $P \spar{S} Q \goesto{a} P' \spar{S} Q'$ and $\rho\in\ffte(P' \spar{S} Q')$.
  By induction there is a valid decomposition of $\rho$ such that
  $\rho_{\rm L}\in \fft(P')$ and $\rho_{\rm R} \mathbin\in  \fft(Q')$.
  Depending on which rule of \tab{sos CCSP} was employed to derive
  $P \spar{S} Q \goesto{a} P' \spar{S} Q'$,
  \begin{enumerate}[(i)]
  \item either $a \notin S$, $P \goesto{a} P'$ and $Q'=Q$,
  \item or $a\in S$, $P \goesto{a} P'$ and $Q \goesto{a} Q'$,
  \item or $a \notin S$, $P'=P$ and $Q \goesto{a} Q'$.
  \end{enumerate}
  In case (i) take $\sigma_{\rm L}:=a\rho_{\rm L}$ and
  $\sigma_{\rm R} := \rho_{\rm R} \in \fft(Q)$.
  As the given decomposition of $\rho$ is valid, so is the newly constructed one of $\sigma$.
  The second rule of \tab{FT CCSP} yields $\sigma_{\rm L} \in \fft(P)$.
  In case (ii) $\sigma_{\rm L}\mathbin{:=}a\rho_{\rm L} \mathbin\in \fft(P)$ and  $\sigma_{\rm R} \mathbin{:=} a\rho_{\rm R}\mathbin\in \fft(Q)$.
  Also this decomposition is valid.
  In case (iii) $\sigma_{\rm L}\mathbin{:=}\rho_{\rm L} \mathbin\in \fft(P)$ and $\sigma_{\rm R} \mathbin{:=} a\rho_{\rm R}\mathbin\in \fft(Q)$.
  Also this decomposition is valid.

\item
  Suppose $\sigma\in\ffte(P\spar{S}Q)$ is derived from the third rule of \tab{FT CCSP}.
  Then $P \spar{S} Q \goesto{\tau} P' \spar{S} Q'$ and $\sigma\in\ffte(P' \spar{S} Q')$.
  By induction there is a valid decomposition of $\sigma$ such that
  $\sigma_{\rm L}\in \fft(P')$ and $\sigma_{\rm R} \mathbin\in  \fft(Q')$.
  Depending on which rule of \tab{sos CCSP} was employed to derive
  $P \spar{S} Q \goesto{\tau} P' \spar{S} Q'$,
  \begin{enumerate}[(i)]
  \item either $P \goesto{\tau} P'$ and $Q'=Q$,
  \item or $P'=P$ and $Q \goesto{\tau} Q'$.
  \end{enumerate}
  In either case the third rule of \tab{FT CCSP} yields
  $\sigma_{\rm L}\in \fft(P)$ and $\sigma_{\rm R} \in  \fft(Q)$.

\item
  Let $\sigma= X\rho$, with $\sigma\in\ffte(P\spar{S}Q)$ derived from the fourth rule of \tab{FT CCSP}.
  Then $P\spar{S}Q {\ngoesto{\alpha}}$ for all $\alpha \in X\cup\{\tau\}$, and $\rho\in\ffte(P\spar{S}Q)$.
  By induction there is a valid decomposition of $\rho$ such that
  $\rho_{\rm L}\in \fft(P)$ and $\rho_{\rm R} \mathbin\in  \fft(Q)$.
  Let $X_{\rm L} := (X\setminus S) \cup \{a\in X\cap S \mid P {\ngoesto{a}}\}$
  and $X_{\rm R} := (X\setminus S) \cup \{a\in X\cap S \mid Q {\ngoesto{a}}\}$.
  Considering that, for $a \mathbin\in X\cap S$, $P\spar{S}Q{\ngoesto{a}}$ iff $P {\ngoesto{a}} \vee Q {\ngoesto{a}}$,\linebreak[3]
  Condition (\ref{decomposition}), with $X$, $X_{\rm L}$ and $X_{\rm R}$ in the
  r\^oles of $\sigma_i$, $\sigma_i^{\rm L}$ and $\sigma_i^{\rm R}$, is satisfied.
  Moreover, $P{\ngoesto{\alpha}}$ for all $\alpha \in X_{\rm L}\cup\{\tau\}$, and
  $Q{\ngoesto{\alpha}}$ for all $\alpha \mathbin\in X_{\rm R}\cup\{\tau\}$.\linebreak[3]
  Take $\sigma_{\rm L}\mathbin= X_{\rm L}\rho_{\rm L}$ and $\sigma_{\rm R}\mathbin= X_{\rm R}\rho_{\rm R}$.
  The fourth rule of \tab{FT CCSP} yields $\sigma_{\rm L}\mathbin\in \fft(P)$ and $\sigma_{\rm R} \mathbin\in  \fft(Q)$.
  It remains to show that this decomposition satisfies Conditions (\ref{consistency}) and (\ref{consistency2}).

  Suppose $\rho_{\rm L}$ has the form $a\eta_{\rm L}$. Then, by \corref{initial}, $a\notin X_{\rm L}$.
  It follows that $X_{\rm L}$ is not system-ended in $\sigma_{\rm L}$.
  In the same manner it follows that $X_{\rm R}$ is not system-ended in $\sigma_{\rm R}$,
  and that $X$ is not system-ended in $\sigma$.

\item
  Let $\sigma= XX\rho$, with $\sigma\in\ffte(P\spar{S}Q)$ derived from the amended fifth rule of \tab{FT CCSP}.
  Then $P\spar{S}Q {\ngoesto{\alpha}}$ for all $\alpha \in X\cup\{\tau\}$,
  $P \spar{S} Q \goesto{\rt} P' \spar{S} Q'$ and $X\rho\in\ffte(P' \spar{S} Q')$.
  By induction there is a valid decomposition $(X_{\rm L}\rho_{\rm L},X_{\rm R}\rho_{\rm R})$ of
  $X\rho$ such that $X_{\rm L}\rho_{\rm L}\in \fft(P')$ and $X_{\rm R}\rho_{\rm R} \mathbin\in  \fft(Q')$.
  Note that $P'{\ngoesto{a}}$ for all $a \in X_{\rm L} \cup \{\tau\}$.
  Depending on which rule of \tab{sos CCSP} was employed to derive
  $P \spar{S} Q \goesto{\rt} P' \spar{S} Q'$,
  \begin{enumerate}[(i)]
  \item either $P \goesto{\rt} P'$ and $Q'=Q$,
  \item or $P'=P$ and $Q \goesto{\rt} Q'$.
  \end{enumerate}
  For reasons of symmetry, {\we} may assume the former.

  Let $X_{\rm L}' := \{a \in X_{\rm L} \mid P {\ngoesto{a}}\}$ and
  $X_{\rm R}' := (X\setminus S) \cup \{a \in X\cap S \mid Q {\ngoesto{a}}\}$.
  Using the fourth rule of \tab{FT CCSP}, 
  $\sigma_{\rm L} := X_{\rm L}'X_{\rm L}\rho_{\rm L} \in \fft(P')$,
  and by the fifth rule $\sigma_{\rm L} \in \fft(P)$, also using that $P{\ngoesto\tau}$.
  Moreover, by the fourth rule of \tab{FT CCSP}, 
  $\sigma_{\rm R} := X_{\rm R}'X_{\rm R}\rho_{\rm R} \in \fft(Q)$, using that $Q{\ngoesto\alpha}$
  for all $\alpha\mathbin\in (X{\setminus} S) \cup \{\tau\}$.
  It remains to show the validity of the decomposition of $\sigma$ that yields
  $\sigma_{\rm L}$ and $\sigma_{\rm R}$.
  Conditions~(\ref{consistency}) and~(\ref{consistency2}) hold trivially.
  So it remains to check Condition~(\ref{decomposition}).

  Let $a \in X\setminus S$. Then $a \in X_{\rm L}\setminus S$ and $P\spar{S}Q {\ngoesto{a}}$, so
  $P {\ngoesto{a}}$, and $a\in X'_{\rm L}\setminus S$.

  Conversely, $X'_{\rm L}\setminus S \subseteq X_{\rm L}\setminus S = X\setminus S$,
  and by definition $X'_{\rm R}\setminus S = X\setminus S$.

  Let $a \in X\cap S$. Then $P\spar{S}Q {\ngoesto{a}}$, so $P {\ngoesto{a}}$ or $Q {\ngoesto{a}}$.
  In case $Q {\ngoesto{a}}$, by definition $a \in X'_{\rm R}\cap S$.
  In case $Q{\goesto{a}}$, it is not possible that $a \in X_{\rm R}$.
  Hence $a \in X_{\rm L}$, by Property~(\ref{decomposition}) of the given valid decomposition of $X\rho$.
  Moreover, $P {\ngoesto{a}}$, so $a \in X'_{\rm L}\cap S$.

  Conversely, by definition $X'_{\rm R}\cap S \subseteq X\cap S$ and
  $X'_{\rm L}\cap S \subseteq X_{\rm L}\cap S \subseteq X\cap S$.

\item
  Let $\sigma= X a\rho$, with $\sigma\in\ffte(P\spar{S}Q)$ derived from the sixth rule of \tab{FT CCSP}.
  Then $P\spar{S}Q {\ngoesto{\alpha}}$ for all $\alpha \in X\cup\{\tau\}$, $a \in X$,
  $P \spar{S} Q \goesto{\rt} P' \spar{S} Q'$ and $a\rho\in\ffte(P' \spar{S} Q')$.
  By induction there is a valid decomposition of $\zeta:=a\rho$ such that
  $\zeta_{\rm L}\in \fft(P')$ and $\zeta_{\rm R} \mathbin\in  \fft(Q')$.
  Let $X_{\rm L} := (X\setminus S) \cup \{b\in X\cap S \mid P {\ngoesto{b}}\}$
  and $X_{\rm R} := (X\setminus S) \cup \{b\in X\cap S \mid Q {\ngoesto{b}}\}$.
  Then $P{\ngoesto{\alpha}}$ for all $\alpha \in X_{\rm L}\cup\{\tau\}$, and
  $Q{\ngoesto{\alpha}}$ for all $\alpha \mathbin\in X_{\rm R}\cup\{\tau\}$.
  Take $\sigma_{\rm L}\mathbin= X_{\rm L}\zeta_{\rm L}$ and $\sigma_{\rm R}\mathbin= X_{\rm R}\zeta_{\rm R}$.
  {\We} have to show that this is a valid decomposition of $\sigma$, and that
  $\sigma_{\rm L} \in \fft(P)$ and $\sigma_{\rm R} \in \fft(Q)$.

  Considering that, for all $b \mathbin\in X\cap S$, $P\spar{S}Q{\ngoesto{b}}$ iff $P {\ngoesto{b}} \vee Q {\ngoesto{b}}$,
  Condition (\ref{decomposition}) is satisfied.

  Depending on which rule of \tab{sos CCSP} was employed to derive
  $P \spar{S} Q \goesto{\rt} P' \spar{S} Q'$,
  \begin{enumerate}[(i)]
  \item either $P \goesto{\rt} P'$ and $Q'=Q$,
  \item or $P'=P$ and $Q \goesto{\rt} Q'$.
  \end{enumerate}
  For reasons of symmetry, {\we} may assume the former.
  By the fourth rule of \tab{FT CCSP}, $\sigma_{\rm R}\in\fft(Q)$.

  {\We} now show that $X_{\rm R}$ is not system-ended in $\sigma_{\rm R}$.
  Namely, if $\zeta_{\rm R}$ is of the form $b\eta_{\rm R}$, then 
  $b\notin X_{\rm R}$ by \corref{initial}. Thus Condition~(\ref{consistency2}) holds.

  \begin{itemize}
  \item
    As a further case distinction, assume $a\notin S$.
    Since $a \in X_{\rm}$ and thus $a\in X_{\rm L}$ and $a\in X_{\rm R}$, by the above it follows
    that $\zeta_{\rm R}$ is not of the form $a\eta_{\rm R}$. Hence, $\zeta_{\rm L}$ has the form $a \eta_{\rm L}$.
    The sixth rule of \tab{FT CCSP} yields $\sigma_{\rm L}\in\fft(P)$.
    Since $X_{\rm L}$ is system-ended in $\sigma_{\rm L}$
    and $X$ is system-ended in $\sigma$, Condition~(\ref{consistency}) holds as well.
  \item
    Finally assume $a\in S$. Then $\zeta_{\rm L}$ has the form $a \eta_{\rm L}$ and $\zeta_{\rm R}$
    has the form $a\eta_{\rm R}$. Since $X_{\rm R}$ is not system-ended in $\sigma_{\rm R}$, one has
    $a \notin X_{\rm R}$, that is, $Q{\goesto{a}}$. As $P\spar{S}Q {\ngoesto{a}}$, it follows that
    $P{\ngoesto{a}}$, and thus $a \in X_{\rm L}$.
    The sixth rule of \tab{FT CCSP} yields $\sigma_{\rm L}\in\fft(P)$.
    Since $X_{\rm L}$ is system-ended in $\sigma_{\rm L}$
    and $X$ is system-ended in $\sigma$, Condition~(\ref{consistency}) holds as well.
  \qedhere
  \end{itemize}
\end{itemize}
\end{proof}

\begin{proposition}{abstraction explicit}
$\sigma\in \fft(\tau_I(P)) ~~\Leftrightarrow~~ \exists \rho \in \fft(P).~ \tau_I(\rho) =\sigma\cup I$.
\end{proposition}
\begin{proof}
Since each state $R$ reachable from $\tau_I(P)$ satisfies $R{\ngoesto{\alpha}}$ for all $\alpha\in I$,
a trivial induction shows that
\[
\sigma\in \fft(\tau_I(P)) \Leftrightarrow \sigma\cup I \in \fft(\tau_I(P)).
\tag{*}
\]
``$\Leftarrow$'': Let $\rho \in \fft(P)$ and $\rho$ survives abstraction from $I$.
By Condition (i) in the definition of abstraction survival,
$\tau_I(\rho)$ is of the form $\sigma\cup I$. In view of (*), it suffices to show that 
$\tau_I(\rho)\in \fft(\tau_I(P))$. {\We} do so with structural induction on the
derivation of $\rho \in \fft(P)$. The case $\rho=\top$ is trivial.

Assume $\rho \in \fft(P)$ is obtained from the third rule of \tab{FT CCSP}.
Then $P \goesto{\tau} P'$ and $\rho \in \fft(P')$.
By induction $\tau_I(\rho) \in \fft(\tau_I(P'))$.
By \tab{sos CCSP} $\tau_I(P) \goesto\tau \tau_I(P')$.
Hence  $\tau_I(\rho) \in \fft(\tau_I(P))$ by the third rule of \tab{FT CCSP}.

Assume $\rho=a\eta \in \fft(P)$ is obtained from the second rule of \tab{FT CCSP}.
Then $P \goesto{a} P'$ and $\eta \in \fft(P')$.
By induction $\tau_I(\eta) \in \fft(\tau_I(P'))$.
By \tab{sos CCSP} $\tau_I(P) \goesto\alpha \tau_I(P')$, with $\alpha=\tau$ if $a\in I$ and $\alpha=a$ otherwise.
In the first case $\tau_I(\rho) =\tau_I(\eta) \in \fft(\tau_I(P))$ by the third rule of \tab{FT CCSP}.
In the second case $\tau_I(\rho) = a\tau_I(\eta) \in \fft(\tau_I(P))$ by the second rule of \tab{FT CCSP}.

Assume $\rho\mathbin=X\eta \mathbin\in \fft(P)$ is obtained from the fourth rule of \tab{FT CCSP}.
Then $P{\ngoesto{\alpha}}$ for all $\alpha\mathbin\in X\cup\{\tau\}$ and $\eta \in \fft(P)$.
  If $\eta$ has the form $a\zeta$, then $a\notin X$ by \corref{initial}.
  In particular, $a\notin I$, using Condition (i) in the definition of abstraction survival.
Consequently, $\tau_I(\rho) = X\tau_I(\eta)$, for any possibly form of $\eta$.
Furthermore, $\tau_I(P){\ngoesto{\alpha}}$ for all $\alpha\mathbin\in X\cup\{\tau\}$.
Here {\we} use that $I\subseteq X$.
By induction $\tau_I(\eta) \in \fft(\tau_I(P))$.
Hence  $\tau_I(\rho) \in \fft(\tau_I(P))$ by the fourth rule of \tab{FT CCSP}.

Assume $\rho\mathbin=X\eta \mathbin\in \fft(P)$ is obtained from the fifth rule of \tab{FT CCSP}.
Then $P{\ngoesto{\alpha}}$ for all $\alpha\mathbin\in X\cup\{\tau\}$,
$P \goesto{\rt} P'$ and $\rho \in \fft(P')$.
So $\tau_I(P) \goesto{\rt} \tau_I(P')$ and $\tau_I(P){\ngoesto{\alpha}}$ for all
$\alpha\mathbin\in X\cup\{\tau\}$, using that $I\subseteq X$.
By induction $\tau_I(\rho) \in \fft(\tau_I(P'))$.
Note that $\tau_I(\rho)$ has the form $X\zeta$, although not necessarily with $\zeta=\tau_I(\eta)$.
The fifth rule of \tab{FT CCSP} yields $\tau_I(\rho) \in \fft(\tau_I(P))$.

Assume $\rho\mathbin=X\eta \mathbin\in \fft(P)$ is obtained from the sixth rule of \tab{FT CCSP}.
Then $P{\ngoesto{\alpha}}$ for all $\alpha\mathbin\in X\cup\{\tau\}$,
$P \goesto{\rt} P'$ and $\eta \in \fft(P')$.
So  $\tau_I(P) \goesto{\rt} \tau_I(P')$ and $\tau_I(P){\ngoesto{\alpha}}$ for all $\alpha\mathbin\in X\cup\{\tau\}$.
By induction $\tau_I(\eta) \in \fft(\tau_I(P'))$.
Using that $\rho$ survives abstraction, {\we} consider three alternatives for $\eta$.
\begin{itemize}
\item
  Let $\eta = c_0 \cdots c_n \top$ with the $c_i\in I$.
  Then $\tau_I(\rho) = X\top$.
  The first and fourth rules of \tab{FT CCSP} yield $\tau_I(\rho) \in \fft(\tau_I(P))$.
\item
  Let $\eta = c_0 \cdots c_n X \zeta$ with the $c_i\in I$.
  Then $\tau_I(\rho) = \tau_I(\eta) = \tau_I(X\zeta)$, and $\tau_I(\eta)$ has the form $X\zeta'$.
  The fifth rule of \tab{FT CCSP} yields $\tau_I(\rho) \in \fft(\tau_I(P))$.
\item
  Let $\eta = c_1 \cdots c_n a \zeta$ with the $c_i\mathbin\in I$ and $a\mathbin{\notin} I$.
  Then $a \mathbin\in X$ (using Condition (iii) of abstraction survival if $n>0$, and the side
  condition of the sixth rule of \tab{FT CCSP} if $n=0$)
  and $\tau_I(\rho) \mathbin= X\tau_I(\eta)$ with $\tau_I(\eta) \mathbin= a\tau_I(\zeta)$.
  The sixth rule of \tab{FT CCSP} yields $\tau_I(\rho) \in \fft(\tau_I(P))$.
\end{itemize}
``$\Rightarrow$'': Let $\sigma\in \fft(\tau_I(P))$; by (*) {\we} may assume that $\sigma=\sigma\cup I$.
{\We} have to find a $\rho\in\fft(P)$ such that $\tau_I(\rho)=\sigma$.
{\We} do so with structural induction on the derivation of $\sigma\in \fft(\tau_I(P))$.

Assume $\sigma\in \fft(\tau_I(P))$ is obtained from the first rule of \tab{FT CCSP}.
Take $\rho=\top$, so that $\tau_I(\rho)=\top=\sigma$.
The first rule of \tab{FT CCSP} yields $\rho\in\fft(P)$.

Assume $\sigma\mathbin\in \fft(\tau_I(P))$ is obtained from the second rule of \tab{FT CCSP}.
Then $\sigma\mathbin=a\zeta$, $\tau_I(P) \goesto{a} \tau_I(P')$ and $\zeta\mathbin\in \fft(\tau_I(P'))$.
So $P \goesto{a} P'$ and $a\mathbin{\notin} I$.
By induction, there is a $\rho\mathbin\in\fft(P')$ such that $\tau_I(\rho)\mathbin=\zeta$.
The second rule of \tab{FT CCSP} yields $a\rho\in\fft(P)$.
Moreover, $\tau_I(a\rho)=a\zeta=\sigma$.

Assume $\sigma\mathbin\in \fft(\tau_I(P))$ is obtained from the third rule of \tab{FT CCSP}.
Then $\tau_I(P) \goesto{\tau} \tau_I(P')$ and $\sigma\mathbin\in \fft(\tau_I(P'))$.
So $P \goesto{\alpha} P'$ with $\alpha\in I\cup\{\tau\}$.
By induction, there is a $\rho\mathbin\in\fft(P')$ such that $\tau_I(\rho)=\sigma$.
In case $\alpha=\tau$, the third rule of \tab{FT CCSP} yields $\rho\in\fft(P)$.
In case $\alpha\neq\tau$, the second rule of \tab{FT CCSP} yields $\alpha\rho\in\fft(P)$;
moreover, $\tau_I(\alpha\rho)=\tau_I(\rho)=\sigma$.

Assume $\sigma\mathbin\in \fft(\tau_I(P))$ is obtained from the fourth rule of \tab{FT CCSP}.
Then $\sigma\mathbin=X\zeta$, $\zeta\mathbin\in \fft(\tau_I(P))$ and
$\tau_I(P){\ngoesto{\alpha}}$ for all $\alpha\mathbin\in X\cup\{\tau\}$.
So $P{\ngoesto{\alpha}}$ for all $\alpha\mathbin\in X\cup\{\tau\}$.
By induction, there is a $\rho\mathbin\in\fft(P)$ such that $\tau_I(\rho)\mathbin=\zeta$.
Hence $X\rho\in\fft(P)$ by the fourth rule of \tab{FT CCSP}.
In case $\rho$ has the form $\top$ or $Y\eta$, one has $\tau_I(X\rho)=X\tau_I(\rho)=X\zeta=\sigma$,
which needed to be shown. 
  So suppose $\rho$ has the form $a\eta$. Then $a\notin X$ by \corref{initial}.
  In particular, $a\notin I$, using Condition (i) in the definition of abstraction survival.
Consequently, $\tau_I(X\rho)=X\tau_I(\rho)=X\zeta=\sigma$.

Assume $\sigma\mathbin\in \fft(\tau_I(P))$ is obtained from the fifth rule of \tab{FT CCSP}.
Then $\sigma\mathbin=X\zeta$, $\tau_I(P) \goesto{\rt} \tau_I(P')$, $\sigma\mathbin\in \fft(\tau_I(P'))$ and
$\tau_I(P){\ngoesto{\alpha}}$ for all $\alpha\mathbin\in X\cup\{\tau\}$.
So $P \goesto{\rt} P'$ and $P{\ngoesto{\alpha}}$ for all $\alpha\mathbin\in X\cup\{\tau\}$.
By induction, there is a $\rho\mathbin\in\fft(P')$ such that $\tau_I(\rho)\mathbin=\sigma$.
This $\rho$ must have the form $X\eta$.
Hence $\rho\in\fft(P)$ by the fifth rule of \tab{FT CCSP}.

Finally, assume $\sigma\mathbin\in \fft(\tau_I(P))$ is obtained from the sixth rule of \tab{FT CCSP}.
Then $\sigma\mathbin=X a\zeta$, $a \in X$, $\tau_I(P) \goesto{\rt} \tau_I(P')$, $a\zeta\mathbin\in \fft(\tau_I(P))$ and
$\tau_I(P){\ngoesto{\alpha}}$ for all $\alpha\mathbin\in X\cup\{\tau\}$.
So $P \goesto{\rt} P'$ and $P{\ngoesto{\alpha}}$ for all $\alpha\mathbin\in X\cup\{\tau\}$.
By induction, there is a $\rho\mathbin\in\fft(P')$ such that $\tau_I(\rho)=a\zeta$.
So $a\notin I$ and $\rho$ has the form $c_1 \dots c_n a\eta$ with $c_i\in I\subseteq X$
for $i=1,\dots,n$ and $\tau_I(\eta)=\zeta$.
The sixth rule of \tab{FT CCSP} yields $X\rho\in\fft(P)$.
Moreover, $\tau_I(X\rho) = X a\tau_I(\eta)=X a \zeta = \sigma$.
\end{proof}

\begin{proposition}{renaming explicit}
  $\sigma\in \fft(\Rn(P)) ~~\Leftrightarrow~~ \exists \rho \in \fft(P).~ \rho \in \Rn^{-1}(\sigma)$.
\end{proposition}
\begin{proof}
``$\Leftarrow$'': Let $\rho \in \fft(P)$.
With structural induction on the derivation of $\rho \in \fft(P)$
{\we} show that, for any $\sigma$ with $\rho \in \Rn^{-1}(\sigma)$,
one has $\sigma\in \fft(\Rn(P))$.

If $\rho \mathbin\in \fft(P)$ is obtained from the first rule of \tab{FT CCSP}, then $\sigma=\rho=\top$ and $\sigma\mathbin\in\fft(\Rn(P))$.

Assume $\rho \mathbin= a\eta \mathbin\in \fft(P)$ is obtained from the second rule.
Then $P \goesto{a} P'$ and $\eta \mathbin\in \fft(P')$.
Moreover, $\sigma\mathbin= b\zeta$ with $(a,b)\mathbin\in\Rn$ and $\eta\mathbin\in \Rn^{-1}(\zeta)$.
So $\Rn(P) \goesto{b} \Rn(P')$ by \tab{sos CCSP}, and $\zeta \mathbin\in \fft(\Rn(P'))$ by induction.
It follows that $\sigma\mathbin\in\fft(\Rn(P))$.

Assume $\rho \mathbin\in \fft(P)$ is obtained from the third rule of \tab{FT CCSP}.
Then $P \goesto{\tau} P'$ and $\rho \mathbin\in \fft(P')$.
So $\Rn(P) \goesto\tau \Rn(P')$ by \tab{sos CCSP}, and $\sigma \mathbin\in \fft(\Rn(P'))$ by induction.
It follows that $\sigma\mathbin\in\fft(\Rn(P))$.

Assume $\rho \mathbin= Y \eta \mathbin\in \fft(P)$ is obtained from the fourth rule.
Then $P{\ngoesto{\alpha}}$ for all $\alpha\mathbin\in Y\cup\{\tau\}$ and $\eta \in \fft(P)$.
Moreover, $\sigma\mathbin= X \zeta$ with $Y=\Rn^{-1}(X)$ and $\eta\mathbin\in \Rn^{-1}(\zeta)$.
So $\Rn(P){\ngoesto{\alpha}}$ for all $\alpha\mathbin\in X\cup\{\tau\}$,
and $\zeta \mathbin\in \fft(\Rn(P))$ by induction.
It follows that $\sigma\mathbin\in\fft(\Rn(P))$.

Assume $\rho \mathbin= Y \eta \mathbin\in \fft(P)$ is obtained from the fifth rule.
Then $P{\ngoesto{\alpha}}$ for all $\alpha\mathbin\in Y\cup\{\tau\}$,
$P \goesto{\rt} P'$ and $\rho \in \fft(P')$.
Moreover, $\sigma\mathbin= X \zeta$ with $Y=\Rn^{-1}(X)$.
So $\Rn(P){\ngoesto{\alpha}}$ for all $\alpha\mathbin\in X\cup\{\tau\}$,
$\Rn(P) \goesto{\rt} \Rn(P')$ and $\sigma \mathbin\in \fft(\Rn(P'))$ by induction.
It follows that $\sigma\mathbin\in\fft(\Rn(P))$.

Assume $\rho \mathbin= Y a\eta \mathbin\in \fft(P)$ is obtained from the sixth rule.
Then $P{\ngoesto{\alpha}}$ for all $\alpha\mathbin\in Y\cup\{\tau\}$,
$P \goesto{\rt} P'$, $a\mathbin\in Y$ and $a\eta \in \fft(P)$.
Moreover, $\sigma\mathbin= X b\zeta$ with $Y\mathbin=\Rn^{-1}(X)$, $(a,b)\mathbin\in\Rn$,
$b\mathbin\in X$ and $a\eta\mathbin\in \Rn^{-1}(b\zeta)$.
So $\Rn(P){\ngoesto{\alpha}}$ for all $\alpha\mathbin\in X\cup\{\tau\}$,
$\Rn(P) \goesto{\rt} \Rn(P')$ and $b\zeta \mathbin\in \fft(\Rn(P'))$ by induction.
It follows that $\sigma\mathbin\in\fft(\Rn(P))$.

``$\Rightarrow$'': Let $\sigma \in \fft(\Rn(P))$.
With structural induction on the derivation of $\sigma \in \fft(\Rn(P))$ {\we} show that
there exists a $\rho \in \fft(P)$ with $\rho \in \Rn^{-1}(\sigma)$.

If $\sigma =\top \mathbin\in \fft(\Rn(P))$ is obtained from the first rule of \tab{FT CCSP}, take
$\rho=\top$. Now $\rho\mathbin\in\fft(P)$.

Assume $\sigma=b\zeta\in\fft(\Rn(P))$ is obtained from the second rule of \tab{FT CCSP}.
Then $\Rn(P)\goesto{b}\Rn(P')$ and $\zeta\in\fft(\Rn(P'))$.
So $P \goesto{a} P'$ for some $a$ with $(a,b)\mathbin\in\Rn$, and by induction there is an
$\eta \mathbin\in\fft(P')$ with $\eta \in \Rn^{-1}(\zeta)$. 
Take $\rho=a\eta$. Then $\rho \in \fft(P)$ and $\rho \in \Rn^{-1}(\sigma)$.

Assume $\sigma\in\fft(\Rn(P))$ is obtained from the third rule of \tab{FT CCSP}.
Then $\Rn(P)\goesto{\tau}\Rn(P')$ and $\sigma\in\fft(\Rn(P'))$.
So $P \goesto{\tau} P'$, and by induction there is a
$\rho \mathbin\in\fft(P')$ with $\rho \in \Rn^{-1}(\sigma)$. 
Hence $\rho \in \fft(P)$.

Assume $\sigma=X\zeta\in\fft(\Rn(P))$ is obtained from the fourth rule of \tab{FT CCSP}.
Then $\Rn(P){\ngoesto{\alpha}}$ for all $\alpha\mathbin\in X\cup\{\tau\}$ and $\zeta\in\fft(\Rn(P))$.
So $P{\ngoesto{\alpha}}$ for all $\alpha\mathbin\in \Rn^{-1}(X)\cup\{\tau\}$.
By induction there is an $\eta \mathbin\in\fft(P)$ with $\eta \in \Rn^{-1}(\zeta)$. 
Take $\rho = \Rn^{-1}(X)\eta$. Then $\rho \in \fft(P)$ by the fourth rule of \tab{FT CCSP}.
In case $\zeta$ has the form $b\zeta'$, then $\eta$ has the form
$a\eta'$ with $(a,b)\in\Rn$; moreover, \corref{initial} yields  $b\notin X$ as well as $a\notin \Rn^{-1}(X)$.
Consequently, $\rho\in\Rn^{-1}(\sigma)$.

Assume $\sigma=X\zeta\in\fft(\Rn(P))$ is obtained from the fifth rule of \tab{FT CCSP}.
Then $\Rn(P){\ngoesto{\alpha}}$ for all $\alpha\mathbin\in X\cup\{\tau\}$, $\Rn(P) \goesto{\rt} \Rn(P')$
and $\sigma\in\fft(\Rn(P'))$.
So $P{\ngoesto{\alpha}}$ for all $\alpha\mathbin\in \Rn^{-1}(X)\cup\{\tau\}$, $P \goesto{\rt} P'$,
and by induction there is a $\rho \mathbin\in\fft(P')$ with $\rho \in \Rn^{-1}(\sigma)$. 
Of course $\rho$ must have the form $\Rn^{-1}(X)\eta$.
So $\rho \in \fft(P)$ by the fifth rule of \tab{FT CCSP}.

Finally assume $\sigma=X b\zeta\in\fft(\Rn(P))$ is obtained from the sixth rule of \tab{FT CCSP}.
Then $\Rn(P){\ngoesto{\alpha}}$ for all $\alpha\mathbin\in X\cup\{\tau\}$, $b\in X$, $\Rn(P) \goesto{\rt} \Rn(P')$
and $b\zeta\in\fft(\Rn(P'))$.
So $P{\ngoesto{\alpha}}$ for all $\alpha\mathbin\in \Rn^{-1}(X)\cup\{\tau\}$, $P \goesto{\rt} P'$,
and by induction there is an $a\eta \mathbin\in\fft(P')$ with $(a,b)\in\Rn$ and $\eta \in \Rn^{-1}(\zeta)$. 
So $a\mathbin\in \Rn^{-1}(X)$. Take $\rho= \Rn^{-1}(X)a\eta$. Then $\rho \in \fft(P)$ by the sixth rule of \tab{FT CCSP}.
Moreover, $\rho \in \Rn^{-1}(\sigma)$.
\end{proof}

\begin{proposition}{spar explicit rooted}
$\rfft(P \spar{S} Q) = \col(\rfft(P) \spar{S} \rfft(Q))$.
\end{proposition}

\begin{proof}
``$\supseteq$'': Let $\sigma\in\rfft(P) \spar{S} \rfft(Q)$, and fix a valid decomposition of
$\sigma$ as a witness.
With the clearly valid generalisation of \obsref{doubling} to $\rfft(P)$ it suffices to show that $\sigma\in\rfft(P\spar{S}Q)$.
\begin{itemize}
\item
Let $\sigma= \st$.  Then $\st \in \rfft(P)$ and $\st \in \rfft(Q)$.
So $\stable(P)$ and $\stable(Q)$. Indeed, this implies $\stable(P\spar{S} Q)$, and thus $\sigma\in\rfft(P\spar{S}Q)$.

\item
  Let $\sigma\mathbin=\top$, $\sigma\mathbin= a\rho$ or $\sigma \mathbin= X\rho$.
  Then $\sigma_{\rm L} \in \fft(P)$ and $\sigma_{\rm R} \in \fft(Q)$,
  so from (the proof of) \pr{spar explicit} one obtains $\sigma \in \fft(P\spar{S}Q) \subseteq \rfft(P\spar{S}Q)$.

\item
Let $\sigma= \rt X\rho$. Assume that $\sigma_{\rm L} \mathbin= \rt X_{\rm L}\rho_{\rm L}\mathbin \in \rfft(P)$ and
$\sigma_{\rm R} \mathbin= X_{\rm R}\rho_{\rm R}\mathbin\in \rfft(Q)$; the other case follows by symmetry.
Then $\sigma_{\rm R} \mathbin\in \fft(Q)$. Moreover, $\st\mathbin\in \rfft(Q)$, so $Q$ is stable.
By \df{rfft}, $P \goesto{\rt} P'$, $X_{\rm L}\rho_{\rm L} \in \fft(P')$
and $P{\ngoesto\alpha}$ for $\alpha\in X_{\rm L}\cup\{\tau\}$.
Thus $P\spar{S}Q \goesto{\rt} P'\spar{S}Q$ by \tab{sos CCSP}, and $X\rho \in \fft(P'\spar{S}Q)$ by
(the proof of) \pr{spar explicit}.
Furthermore, $X_{\rm R}\rho_{\rm R}\in \fft(Q)$ must be derived by the fourth to sixth rule of \tab{FT CCSP},
so $Q{\ngoesto\alpha}$ for all $\alpha\in X_{\rm R}\cup\{\tau\}$.
Hence $P\spar{S}Q{\ngoesto\alpha}$ for all $\alpha\in X\cup\{\tau\}$, using (\ref{decomposition}).
Therefore, $\sigma \in \rfft(P\spar{S}Q)$.

\item
  Let $\sigma\mathbin=\pst$. Then either $\pst\mathbin\in \rfft(P)\wedge \pst\mathbin\in \rfft(Q)$, or 
  $\pst\mathbin\in \rfft(P) \wedge \st\mathbin\in \rfft(Q)$, or $\st\mathbin\in \rfft(P) \wedge \pst\mathbin\in \rfft(Q)$.
  Either way, $P$ and $Q$ both have paths of $\tau$-transitions to a stable state, and at least one
  of them is nonempty. It follows that $P\spar{S}Q$ has a nonempty path of $\tau$-transitions to a
  stable state, so $\sigma \in \rfft(P\spar{S}Q)$.
\end{itemize}
``$\subseteq$'': Let $\rffte(P)$ be defined exactly as $\rfft(P)$---see \df{rfft}---but writing
$\rt X X \sigma$ instead of $\rt X \sigma$, and employing
$\ffte(P)$ and $\ffte(P')$, as defined in the proof of \pr{spar explicit}, instead of $\fft(P)$ and $\fft(P')$.
Obviously, $\rfft(P \spar{S} Q) \subseteq \col(\rffte(P \spar{S} Q))$, so, using that $\col$ is monotonous, it
suffices to show that $\rffte(P \spar{S} Q) \subseteq \rfft(P) \spar{S} \rfft(Q)$.

So let $\sigma \in \rffte(P \spar{S} Q)$. With structural induction on the derivation of
$\sigma \in \rffte(P \spar{S} Q)$ from the (amended) rules of Tables~\ref{tab:sos CCSP} and~\ref{tab:FT CCSP}
I show that $\sigma \in \rfft(P) \spar{S} \rfft(Q)$. This means that in case $\sigma \mathbin{\neq} \st,\pst$
I have to give a valid decomposition of $\sigma$ such that $\sigma_{\rm L}\in \rfft(P)$ and $\sigma_{\rm R} \in  \rfft(Q)$;
moreover, I have to show that $\st\in \rfft(P)$ and $\st\in  \rfft(Q)$ in case $\sigma$ has the form $\rt X \rho$.
\begin{itemize}
\item
  Let $\sigma\mathbin=\top$, $\sigma\mathbin= a\rho$ or $\sigma \mathbin= X\rho$.
  Then $\sigma \in \ffte(P \spar{S} Q)$. So by the proof of \pr{spar explicit}
  there is a valid decomposition of $\sigma$ such that $\sigma_{\rm L}\in \fft(P) \subseteq \rfft(P)$
  and $\sigma_{\rm R} \in \fft(Q)\subseteq \rfft(Q)$.

\item
  Let $\sigma=\st$.
  Then $P\spar{S}Q{\ngoesto\tau}$, so $P{\ngoesto\tau}$ and $Q{\ngoesto\tau}$.
  So $\st \in \rfft(P)$ and $\st \in \rfft(Q)$.

\item
  Let $\sigma=\rt XX \rho$. By \df{rfft},
  $P\spar{S}Q {\ngoesto{\alpha}}$ for all $\alpha \in X\cup\{\tau\}$,
  $P \spar{S} Q \goesto{\rt} P' \spar{S} Q'$ and $X\rho\in\ffte(P' \spar{S} Q')$.
  It follows that $P {\ngoesto{\tau}}$ and $Q {\ngoesto{\tau}}$, so $\st\in \rfft(P)$ and $\st\in \rfft(P)$.
  By the proof of \pr{spar explicit} there is a valid decomposition of $X\rho$ such that
  $X_{\rm L}\rho_{\rm L}\in \fft(P')$ and $X_{\rm R}\rho_{\rm R} \mathbin\in \fft(Q')$.
  Depending on which rule of \tab{sos CCSP} derived
  $P \spar{S} Q \goesto{\rt} P' \spar{S} Q'$,
  \begin{enumerate}[(i)]
  \item either $P \goesto{\rt} P'$ and $Q'=Q$,
  \item or $P'=P$ and $Q \goesto{\rt} Q'$.
  \end{enumerate}
  For reasons of symmetry, {\we} may assume the former.

  Let $X_{\rm L}' := \{a \in X_{\rm L} \mid P {\ngoesto{a}}\}$ and
  $X_{\rm R}' := (X\setminus S) \cup \{a \in X\cap S \mid Q {\ngoesto{a}}\}$.
  Using the fourth rule of \tab{FT CCSP}, 
  $X_{\rm L}'X_{\rm L}\rho_{\rm L} \in \fft(P')$, so by \df{rfft}
  $\sigma_{\rm L} := \rt X_{\rm L}'X_{\rm L}\rho_{\rm L} \in \rfft(P)$.
  Moreover, by the fourth rule of \tab{FT CCSP}, 
  $\sigma_{\rm R} := X_{\rm R}'X_{\rm R}\rho_{\rm R} \in \fft(Q)\subseteq \rfft(Q)$, using that $Q{\ngoesto\alpha}$
  for all $\alpha\mathbin\in (X{\setminus} S) \cup \{\tau\}$.
  It remains to show the validity of the decomposition of $\sigma$ into $\sigma_{\rm L}$ and $\sigma_{\rm R}$.
  This proceeds exactly as in the corresponding case in the proof of \pr{spar explicit}
  (the case $\sigma= X X\rho$ in direction ``$\subseteq$'').

\item
  Let $\sigma=\pst$. Then $P\spar{S}Q$ has a nonempty path of $\tau$-transitions to a stable state.
  This path projects to paths of $\tau$-transitions from $P$ and from $Q$ to stable states, and at
  least one of them must be nonempty. It follows that either $\pst\mathbin\in \rfft(P)\wedge \pst\mathbin\in \rfft(Q)$, or 
  $\pst\mathbin\in \rfft(P) \wedge \st\mathbin\in \rfft(Q)$, or $\st\mathbin\in \rfft(P) \wedge \pst\mathbin\in \rfft(Q)$.
\qedhere
\end{itemize}
\end{proof}

\begin{proposition}{abstraction explicit rooted}
Let $\st\neq\sigma\neq\pst$.
Then $$\sigma\in \rfft(\tau_I(P)) ~~\Leftrightarrow~~ \exists \rho \in \rfft(P).~ \tau_I(\rho) =\sigma\cup I.$$
Moreover, $\st\in \rfft(\tau_I(P)) ~~\Leftrightarrow~~ \st\in\rfft(P) \wedge \I(P)\cap I=\emptyset$,\\
and $\pst\in \rfft(\tau_I(P)) ~~\Leftrightarrow~~ \begin{array}[t]{@{}ll@{}}
                        (\exists c_0c_1\dots c_n I\top \in \rfft(P) \mbox{~where~}n\geq 0 \mbox{~and all~}c_i\in I)\\
                          \mbox{} \vee (\pst\in \rfft(P) \wedge I\top\in \rfft(P) ).
                        \end{array}$
\end{proposition}
\begin{proof}
Let $\st\neq\sigma\neq\pst$. Since each state $R$ reachable from $\tau_I(P)$ satisfies $R{\ngoesto{\alpha}}$ for all $\alpha\in I$,
a trivial induction shows that
\[
\sigma\in \rfft(\tau_I(P)) \Leftrightarrow \sigma\cup I \in \rfft(\tau_I(P))
.
\tag{*}
\]
``$\Leftarrow$'': Let $\rho \mathbin\in \rfft(P)$ and $\rho$ survives abstraction from $I$.
By Condition (i) in the definition of abstraction survival,
$\tau_I(\rho)$ is of the form $\sigma\cup I$. In view of (*), it suffices to show that 
$\tau_I(\rho)\mathbin\in \rfft(\tau_I(P))$.

  Suppose $\rho\mathbin=\top$, $\rho\mathbin= a\eta$ or $\rho \mathbin= X\eta$.
  Then $\rho \in \fft(P)$. So $\tau_I(\rho)\in \fft(\tau_I(P)) \subseteq \rfft(\tau_I(P))$
  by the proof of \pr{abstraction explicit}.

  Suppose $\rho\mathbin=\rt X\eta$.
  Then $P{\ngoesto\alpha}$ for all $\alpha\in X\cup\{\tau\}$, $P\goesto\rt P'$ and $X\eta \in \fft(P')$, by \df{rfft}.
  So $\tau_I(X\eta)\in \fft(\tau_I(P'))$ by the proof of \pr{abstraction explicit}.
  Moreover, $\tau_I(P) \goesto{\rt} \tau_I(P')$ and $\tau_I(P){\ngoesto{\alpha}}$ for all
  $\alpha\mathbin\in X\cup\{\tau\}$, using that $I\subseteq X$.
  Note that $\tau_I(X\eta)$ has the form $X\zeta$, although not necessarily with $\zeta=\tau_I(\eta)$.
  \df{rfft} yields $\tau_I(\rho) = \rt\tau_I(X\eta) \in \rfft(\tau_I(P))$.
\\[1ex]
``$\Rightarrow$'': Let $\sigma\in \rfft(\tau_I(P))$; by (*) {\we} may assume that $\sigma=\sigma\cup I$.
{\We} have to find a $\rho\in\rfft(P)$ such that $\tau_I(\rho)=\sigma$.

  Suppose $\sigma\mathbin=\top$, $\sigma\mathbin= a\eta$ or $\sigma \mathbin= X\eta$.
  Then $\sigma\in \fft(\tau_I(P))$. So by the proof of \pr{abstraction explicit}
  there is a $\rho\in\fft(P)\subseteq\rfft(P)$ such that $\tau_I(\rho)=\sigma$.

  Suppose $\sigma\mathbin=\rt X\zeta$. Then $\tau_I(P){\ngoesto\alpha}$ for all $\alpha\in X\cup\{\tau\}$,
  $\tau_I(P)\goesto\rt \tau_I(P')$ and $X\zeta \in \fft(\tau_I(P'))$, by \df{rfft}.
  Hence $P{\ngoesto\alpha}$ for all $\alpha\in X\cup\{\tau\}$, and $P\goesto\rt P'$.
  By the proof of \pr{abstraction explicit} there is a $X\eta\in\fft(P')$ such that $\tau_I(X\eta)=X\zeta$.
  So $\rho:=\rt X \eta \in\rfft(P)$ by \df{rfft}, and $\tau_I(\rho)=\sigma$.
\\[1ex]
The second statement of \pr{abstraction explicit rooted} is trivial.
Now consider the third.
\\[1ex]
``$\Leftarrow$'': Suppose $\exists c_0c_1\dots c_n I\top \in \rfft(P) \mbox{~where~}n\geq 0 \mbox{~and all~}c_i\in I$.
Then $P$ has a nonempty path, all of which transitions are labelled $\tau$ or $c_i\in I$, ending in a
state $P'$ satisfying $P'{\ngoesto\alpha}$ for all $\alpha\in I\cup \{\tau\}$.
Consequently, $\tau_I(P)$ has a nonempty part, all of which transitions are labelled $\tau$, ending in
a state $\tau(P')$ satisfying $\tau_I(P'){\ngoesto\tau}$. Consequently, $\tau_I(P){\goesto\tau}$ and
$\emptyset\top\in\fft(\tau_I(P))$. Thus $\pst\in\rfft(\tau_I(P))$.

Now suppose $\pst\in \rfft(P) \wedge I\top\in \rfft(P)$.
Then $P{\goesto\tau}$ and $P$ has a path, all of which transitions are labelled $\tau$, ending in
a state $P'$ satisfying $P'{\ngoesto\alpha}$ for all $\alpha\in I\cup \{\tau\}$.
This path must be nonempty. Again it follows that $\pst\in\rfft(\tau_I(P))$.
\\[1ex]
``$\Leftarrow$'': Suppose $\pst\in\rfft(\tau_I(P))$. Then $\tau_I(P)$ has a nonempty part, all of
which transitions are labelled $\tau$, ending in a state $\tau(P')$ satisfying $\tau_I(P'){\ngoesto\tau}$.
Consequently, $P$ has a nonempty path, all of which transitions are labelled $\tau$ or $c\in I$, ending in a
state $P'$ satisfying $P'{\ngoesto\alpha}$ for all $\alpha\in I\cup \{\tau\}$.
In case on this path some transitions are labelled $c\mathbin\in I$, one obtains
$\exists c_0c_1\dots c_n I\top \in \rfft(P) \mbox{~where~}n\mathbin{\geq} 0$ and all $c_i\mathbin\in I$.
In case there are no such transitions,  $\pst\in \rfft(P)$ and $I\top\in \rfft(P)$.
\end{proof}

\begin{proposition}{renaming rooted}
$\sigma\in \rfft(\Rn(P)) ~~\Leftrightarrow~~ \exists \rho \in \rfft(P).~ \rho \in \Rn^{-1}(\sigma)$.
\end{proposition}

\begin{proof}
``$\Leftarrow$'': Let $\rho \in \rfft(P)$, and let $\sigma$ satisfy $\rho \in\Rn^{-1}(\sigma)$.
{\We} have to show that $\sigma\in \rfft(\Rn(P))$.

Let $\rho\mathbin=\top$, $\rho\mathbin= a\eta$ or $\rho \mathbin= X\eta$.
Then $\rho \mathbin\in \fft(P)$, so by \pr{renaming explicit} $\sigma\mathbin\in \fft(\Rn(P)) \subseteq \rfft(\Rn(P))$.

Let $\rho=\rt Y \eta$.
Then $P{\ngoesto{\alpha}}$ for all $\alpha\mathbin\in Y\cup\{\tau\}$,
$P \goesto{\rt} P'$ and $Y\eta \in \fft(P')$.
Moreover, $\sigma\mathbin= \rt X \zeta$ with $Y\mathbin=\Rn^{-1}(X)$ and $Y\eta\in\Rn^{-1}(X\zeta)$.
Consequently $\Rn(P){\ngoesto{\alpha}}$ for all $\alpha\mathbin\in X\cup\{\tau\}$,
$\Rn(P) \goesto{\rt} \Rn(P')$ and $X\zeta \mathbin\in \fft(\Rn(P'))$ by \pr{renaming explicit}.
It follows that $\sigma\mathbin\in\rfft(\Rn(P))$.

Let $\rho=\st$. Then $\stable(P)$, so $\stable(\Rn(P))$ and $\sigma=\st\in\rfft(\Rn(P))$.

Let $\rho=\pst$. Then $P{\goesto\tau}$ and $\emptyset\top\in\fft(P)$.
So $\Rn(P){\goesto\tau}$ and $\emptyset\top\in\fft(\Rn(P))$.
Therefore $\sigma=\pst\in\rfft(\Rn(P))$.

``$\Rightarrow$'': Let $\sigma\in \rfft(\Rn(P))$. {\We} have to find a $\rho \in \rfft(P)$ with $\rho \in \Rn^{-1}(\sigma)$.

Let $\sigma\mathbin=\top$, $\sigma\mathbin= a\eta$ or $\sigma \mathbin= X\eta$.
Then $\sigma\in \fft(\Rn(P))$, so by \pr{renaming explicit} there is a $\rho \in\fft(P)\subseteq \rfft(P)$ with $\rho \in \Rn^{-1}(\sigma)$.

Let $\sigma = \rt X\zeta$. Then $\Rn(P){\ngoesto{\alpha}}$ for all $\alpha\mathbin\in X\cup\{\tau\}$,
$\Rn(P) \goesto{\rt} \Rn(P')$ and $X\zeta \mathbin\in \fft(\Rn(P'))$.
By \pr{renaming explicit} there is a $Y\eta \in\fft(P')$ with $Y\eta \in \Rn^{-1}(X\zeta)$.
In particular $Y=\Rn^{-1}(X)$. Therefore, $P{\ngoesto{\alpha}}$ for all $\alpha\mathbin\in Y\cup\{\tau\}$.
Moreover, $P \goesto{\rt} P'$. Take $\rho:=\rt Y \eta$. Then $\rho \in \rfft(P)$ and $\rho \in \Rn^{-1}(\sigma)$.

Let $\sigma=\st$. Then $\stable(\Rn(P))$, so $\stable(P)$ and $\rho=\st\in\rfft(P)$.

Let $\sigma=\pst$. Then $\Rn(P){\goesto\tau}$ and $\emptyset\top\in\fft(\Rn(P))$.
So $P{\goesto\tau}$ and $\emptyset\top\in\fft(P)$. Therefore $\rho=\pst\in\rfft(P)$.
\end{proof}
\end{document}